\documentclass[envcountsame,runningheads]{llncs}

\usepackage{bbding}

\makeatletter
\RequirePackage[bookmarks,unicode,colorlinks=true]{hyperref}%
    \def\@citecolor{blue}%
    \def\@urlcolor{blue}%
    \def\@linkcolor{blue}%

\def\orcidID#1{\smash{\href{http://orcid.org/#1}{\protect\raisebox{-1.25pt}{\protect\includegraphics{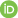}}}}}
\makeatother

\usepackage[utf8]{inputenc}
\usepackage[english]{babel}

\usepackage{amsmath}
\usepackage{amssymb}
\usepackage{mathtools}
\usepackage{thmtools}
\usepackage{stmaryrd}
\usepackage{wasysym}
\usepackage{xfrac}

\usepackage{tikz}
\usepackage{pgfplots}
\usepackage{lstautogobble}
\usepackage{listings}
\colorlet{punct}{red!60!black}
\definecolor{delim}{RGB}{20,105,176}
\colorlet{numb}{magenta!60!black}
\lstdefinelanguage{JSON}{
    autogobble=true,
    basicstyle=\normalfont\small\ttfamily,
    numbers=left,
    numberstyle=\scriptsize,
    stepnumber=1,
    showstringspaces=false,
    breaklines=true,
    frame=lines,
    literate=
     *{0}{{{\color{numb}0}}}{1}
      {1}{{{\color{numb}1}}}{1}
      {2}{{{\color{numb}2}}}{1}
      {3}{{{\color{numb}3}}}{1}
      {4}{{{\color{numb}4}}}{1}
      {5}{{{\color{numb}5}}}{1}
      {6}{{{\color{numb}6}}}{1}
      {7}{{{\color{numb}7}}}{1}
      {8}{{{\color{numb}8}}}{1}
      {9}{{{\color{numb}9}}}{1}
      {:}{{{\color{punct}{:}}}}{1}
      {,}{{{\color{punct}{,}}}}{1}
      {\{}{{{\color{delim}{\{}}}}{1}
      {\}}{{{\color{delim}{\}}}}}{1}
      {[}{{{\color{delim}{[}}}}{1} 
      {]}{{{\color{delim}{]}}}}{1}, 
    morestring=[b]", 
    morestring=[d]',
}
\usepackage{xparse}
\usepackage{xspace}

\usepackage{csquotes}

\usepackage[nameinlink, capitalize]{cleveref}

\usepackage{todonotes}

\usepackage[noend]{algpseudocode}
\usepackage{algorithm}
\usepackage{tabularx}
\usepackage{booktabs}
\usepackage{subcaption}
\newcolumntype{R}{>{\raggedleft\arraybackslash}X}%

\usepackage{biblatex}

\usepackage{mathrsfs}

\usepackage[inline]{enumitem}

\usetikzlibrary{
    automata,
    arrows.meta,
    calc,
    positioning,
    backgrounds,
    shapes.geometric,
}
\tikzset {
    edge with arrow/.style = {
        ->,
        >=stealth,
        shorten >=1pt,
    },
    directed/.style = {
        edge with arrow,
        node distance=2cm,
        on grid,
        semithick,
        double distance=1.5pt,
    },
    automaton/.style = {
        directed,
        auto,
        initial text={},
        state/.append style = {
            ellipse,
        },
    },
    deltaZero/.style = {
        blue
    },
    deltaNotZero/.style = {
        teal
    }
}
\pgfplotsset{
    compat=newest,
    colormap={Blues}{rgb=(1, 1, 1) rgb=(0, 0, 1)}, 
    colormap={Reds}{rgb=(1, 1, 1) rgb=(1, 0, 0)}, 
    legend pos=outer north east
}

\NewDocumentCommand{\N}{}{\mathbb{N}}

\NewDocumentCommand{\prefixes}{m}{\mathit{Pref}{(#1)}} 
\NewDocumentCommand{\suffixes}{m}{\mathit{Suff}{(#1)}}
\NewDocumentCommand{\SigmaStarLl}{O{\ell}}{\Sigma_{L,#1}^*}
\NewDocumentCommand{\sizeOfSet}{m}{|#1|}

\RenewDocumentCommand{\iff}{}{\Leftrightarrow}
\RenewDocumentCommand{\implies}{}{\Rightarrow}

\NewDocumentCommand{\complexity}{}{\mathcal{O}}

\NewDocumentCommand{\emptyword}{}{\varepsilon}
\NewDocumentCommand{\automaton}{O{A}}{\mathcal{#1}}
\NewDocumentCommand{\nerodeCongruence}{}{\mathbin{\sim}}

\NewDocumentCommand{\equivalenceClass}{m}{\llbracket#1\rrbracket}
\NewDocumentCommand{\languageOf}{m}{\mathcal{L}{(#1)}}
\NewDocumentCommand{\lengthOf}{m}{\lvert#1\rvert}
\DeclareMathOperator*{\transition}{\mathbin{\longrightarrow}}

\NewDocumentCommand{\deltaZero}{o}{\IfNoValueTF{#1}{\delta_{=0}}{\delta^{=0}_{#1}}}
\NewDocumentCommand{\deltaNotZero}{o}{\IfNoValueTF{#1}{\delta_{>0}}{\delta^{>0}_{#1}}}
\NewDocumentCommand{\equivBGROCA}{}{\equiv}
\NewDocumentCommand{\visiblyDeltaZero}{}{\visibly{\delta}_{=0}}
\NewDocumentCommand{\visiblyDeltaNotZero}{}{\visibly{\delta}_{>0}}

\NewDocumentCommand{\pushdownAlphabet}{}{\visibly{\Sigma}}
\NewDocumentCommand{\callAlphabet}{}{\Sigma_{c}}
\NewDocumentCommand{\returnAlphabet}{}{\Sigma_{r}}
\NewDocumentCommand{\internalAlphabet}{}{\Sigma_{int}}
\NewDocumentCommand{\boundedHeights}{O{\ell}}{\Sigma_{0, #1}^*}
\NewDocumentCommand{\boundedLanguage}{O{\ell}}{L_{\leq {#1}}}
\NewDocumentCommand{\countervalue}{m}{cv{(#1)}}
\NewDocumentCommand{\counterAut}{m}{c_{\automaton}{(#1)}}
\NewDocumentCommand{\height}{m}{h{(#1)}}
\NewDocumentCommand{\heightAut}{mO{\automaton}}{h_{#2}{(#1)}}
\NewDocumentCommand{\signOf}{m}{\chi{(#1)}}
\NewDocumentCommand{\visibly}{m}{\overline{#1}}
\NewDocumentCommand{\toVisibly}{}{\lambda}

\NewDocumentCommand{\behaviorGraphAut}{O{\automaton}}{BG(#1)} 
\NewDocumentCommand{\limBehaviorGraphAut}{O{\ell}O{\automaton}}{BG_{#1}(#2)} 
\NewDocumentCommand{\behaviorGraph}{O{L}}{BG_{#1}}
\NewDocumentCommand{\limitedBehaviorGraph}{O{\ell}O{L}}{BG_{#2\mid#1}}

\NewDocumentCommand{\isomorphism}{}{\varphi}

\NewDocumentCommand{\representatives}{}{R}
\NewDocumentCommand{\separatorsCounters}{}{S}
\NewDocumentCommand{\separatorsMapping}{}{\widehat{S}}
\NewDocumentCommand{\Rf}{}{R_f}
\NewDocumentCommand{\RfSize}{}{\alpha}

\NewDocumentCommand{\SHatf}{}{\widehat{S}_f}
\NewDocumentCommand{\binWord}{}{$\bot$-word\xspace}
\NewDocumentCommand{\binWords}{}{$\bot$-words\xspace}
\NewDocumentCommand{\withoutBinWords}{m}{\overline{#1}}
\NewDocumentCommand{\hypothesis}{}{\automaton[H]}

\NewDocumentCommand{\behaviorGraphHypothesis}{O{\ell}O{L}}{\automaton[H]_{#2 \mid #1}}

\NewDocumentCommand{\LStar}{}{\ensuremath{{L}^*}\xspace}

\NewDocumentCommand{\observationTable}{O{\ell}}{\mathscr{O}_{#1}}

\NewDocumentCommand{\Approx}{}{\mathit{Approx}}
\NewDocumentCommand{\approxSet}{}{approximation set}
\NewDocumentCommand{\obsTable}{}{\observationTable}
\NewDocumentCommand{\obsMapping}{O{\ell}}{\mathcal{L}_{#1}}
\NewDocumentCommand{\obsCounters}{O{\ell}}{\mathcal{C}_{#1}}
\NewDocumentCommand{\equivTable}{O{\obsTable}}{\equiv_{#1}}
\NewDocumentCommand{\autTable}{O{\obsTable}}{{\automaton}_{#1}}
\NewDocumentCommand{\limAutTable}{O{\automaton}O{\obsTable}O{t}}{{#1}_{#2,#3}}
\NewDocumentCommand{\binClass}{O{\obsTable}}{\bot_{#1}}
\NewDocumentCommand{\prefTable}{O{\obsTable}}{\mathit{Pref}({#1})} 
\NewDocumentCommand{\prefLevel}{O{\ell}O{L}}{\mathit{Pref}_{#1}({#2})} 
\NewDocumentCommand{\prefixed}{}{$\bot$-consistent\xspace}

\NewDocumentCommand{\notPrefixed}{}{$\bot$-inconsistent\xspace}

\NewDocumentCommand{\inconsistencyBot}{}{$\bot$-inconsistency\xspace}

\NewDocumentCommand{\consistent}{}{$\Sigma$-consistent\xspace}

\NewDocumentCommand{\inconsistent}{}{$\Sigma$-inconsistent\xspace}
\NewDocumentCommand{\inconsistency}{o}{\IfNoValueTF{#1}{$\Sigma$-inconsistency}{$\Sigma$-(#1)-inconsistency}\xspace} 

\NewDocumentCommand{\mismatch}{}{mismatch\xspace}

\NewDocumentCommand{\lengthCe}{}{t} 
\NewDocumentCommand{\observationTree}{O{\ell}}{\mathcal{T}_{#1}}

\algblockdefx{DoWhile}{EndDoWhile}
    {\algorithmicdo}
    [1]{\algorithmicwhile\ #1}
\algnewcommand{\LineComment}[1]{\Statex \(\triangleright\) #1}

\title{Learning Realtime One-Counter Automata\thanks{This work was partially supported by the Belgian FWO \enquote{SAILor} project (G030020N). Ga\"etan Staquet is a research fellow (Aspirant) of the Fonds de la Recherche Scientifique -- FNRS.} 
}
\author{%
    Véronique Bruyère\inst{1}\orcidID{0000-0002-9680-9140} \and
    Guillermo A. P\'erez\inst{2}\orcidID{0000-0002-1200-4952} \and
    Gaëtan Staquet\inst{1,2}(\Envelope)\orcidID{0000-0001-5795-3265}
}
\authorrunning{V. Bruy\`ere, G. A. P\'erez, and G. Staquet}
\institute{%
    University of Mons (UMONS), Mons, Belgium
	\email{\{veronique.bruyere,gaetan.staquet\}@umons.ac.be}
    \and
    University of Antwerp (UAntwerp) -- Flanders Make, Antwerp, Belgium 
	\email{guillermoalberto.perez@uantwerpen.be}
}

\begin{document}

\maketitle

\begin{abstract}
We present a new learning algorithm for realtime one-counter automata.
Our algorithm uses membership and equivalence queries as in Angluin's $\LStar$ algorithm, as well as
counter value queries and partial equivalence queries. In a partial
equivalence query, we ask the teacher whether the language of a given
finite-state automaton coincides with a counter-bounded subset of the target
language. We evaluate an implementation of our algorithm on a number of random
benchmarks and on a use case regarding efficient JSON-stream validation.
\end{abstract}

\section{Introduction}
In \emph{active learning}, a \emph{learner} has to infer a model of an unknown machine by interacting with a \emph{teacher}. Angluin's seminal \LStar algorithm does precisely this for finite-state automata while using only \emph{membership} and \emph{equivalence queries}~\cite{DBLP:journals/iandc/Angluin87}. An important application of active learning is to learn black-box models from (legacy) software and hardware systems~\cite{DBLP:journals/igpl/GrocePY06,DBLP:journals/jalc/PeledVY02}. Though recent works have greatly advanced the state of the art in finite-state automata learning, handling real-world applications usually involves tailor-made abstractions to circumvent elements of the system which result in an infinite state space~\cite{DBLP:journals/fmsd/AartsJUV15}. This highlights the need for learning algorithms that focus on more expressive models.

One-counter automata are obtained by extending finite-state automata with an integer-valued variable, i.e., the counter, that can be increased, decreased, and tested for equality against a finite number of values.
The counter allows such automata to capture the behavior of some infinite-state systems. Additionally, their expressiveness has been shown sufficient to verify programs with lists~\cite{DBLP:journals/fmsd/BouajjaniBHIMV11} and validate  XML streams~\cite{DBLP:conf/webdb/ChiticR04}. To the best of our knowledge, there is no learning algorithm for general one-counter automata. 

For \emph{visibly} one-counter automata (that is, when the alphabet is such that letters determine whether the counter is decreased, increased, or not affected), Neider and L\"oding describe an elegant algorithm in~\cite{neider2010learning}. Besides the usual membership and equivalence queries, their algorithm uses \emph{partial equivalence queries}: given a finite-state automaton $\mathcal{A}$ and a bound $k$, the teacher answers whether the language of $\mathcal{A}$ corresponds to the $k$-bounded subset of the target language. (This sub-language intuitively corresponds to what the one-counter automaton accepts if we only allow its counter to take values of at most $k$.) Additionally, Fahmy and Roos~\cite{DBLP:conf/alt/FahmyR95} claim to have solved a more general case. Namely, the case of realtime one-counter automata (that is, when the automaton is assumed to be configuration-deterministic and no $\emptyword$-transitions are allowed).
However, we were unable to understand the algorithm and proofs in that paper
due to lack of precise formalization and detailed proofs.
We also found an example where the provided algorithm did not produce the expected results.
It is noteworthy that B\"ohm et al.~\cite{DBLP:journals/jcss/BohmGJ14} made similar remarks about related works of Roos~\cite{DBLP:conf/focs/BermanR87,roos1988deciding}.

\paragraph*{Our contribution.}
We present a learning algorithm for realtime one-counter automata. Our
algorithm uses membership, equivalence and partial equivalence queries. It
also makes use of \emph{counter value queries}. That is, we make the
assumption that we have an executable black box with observable counter
values. We prove that our algorithm runs in exponential time and
space and that it uses at most an exponential number of queries.

One-counter automata with counter-value observability were observed to have
desirable properties by Bollig in~\cite{DBLP:conf/fsttcs/Bollig16}.
Importantly, in the same paper Bollig highlights a connection between such
automata and visibly one-counter automata. We expose a similar connection and
are thus able to leverage Neider and L\"oding's learning algorithm for visibly
one-counter languages~\cite{neider2010learning} as a sort of sub-routine for
ours. This is why our algorithm uses a superset of the query types required
in~\cite{neider2010learning}. Nevertheless, due to the fact that the counter
values cannot be inferred from a given word, our learning algorithm is more
complex. 
Technically, the latter required us to extend the classical
definition of \emph{observation tables} as used in, e.g.,~\cite{DBLP:journals/iandc/Angluin87,neider2010learning}.
Namely, entries in
our tables are composed of Boolean language information as well as a counter
value or a \emph{wildcard} encoding the fact that we do not (yet) care about
the value of the corresponding word. (Our use of wildcards is reminiscent of
the work of Leucker and Neider~\cite{DBLP:conf/isola/LeuckerN12} on learning a
regular language from an ``inexperienced'' teacher who may answer queries in
an unreliable manner.) Due to these extensions, much work is required to prove
that it is always possible to make a table closed and consistent in finite
time. A crucial element of our algorithm is that we formulate queries for the
teacher in a way which ensures the observation table eventually induces a
right congruence refining the classical Myhill-Nerode congruence with
counter-value information. (This is in contrast
with~\cite{DBLP:conf/isola/LeuckerN12}, where the ambiguity introduced by
wildcards is resolved using SAT solvers.)

We evaluate an implementation of our algorithm on random benchmarks and a use
case inspired by~\cite{DBLP:conf/webdb/ChiticR04}. Namely, we learn a realtime
one-counter automaton model for a simple JSON schema validator --- i.e., a
program that verifies whether a JSON document satisfies a given JSON schema.
The advantage of having a finite-state model of such a validator is that
JSON-stream validation becomes trivially efficient (cf.\ automata-based
parsing~\cite{DBLP:books/aw/AhoSU86}).

\paragraph*{Related work.} Our assumption about counter-value observability
means that the system with which we interact is not really a black box.
Rather, we see it as a gray box. Several recent active-learning works make
such assumptions to learn complex languages or ameliorate query-usage bounds.
For instance, in~\cite{DBLP:conf/dlt/BerthonBPR21}, the authors assume they
have information about the target language $L$ in the form of a superset of
it. Similarly, in~\cite{DBLP:conf/nfm/0002R16}, the authors assume $L$ is
obtained as the composition of two languages, one of which they know in
advance. In~\cite{DBLP:conf/ecai/MichaliszynO20}, the teacher is assumed to
have an executable automaton representation of the (infinite-word) target
language. This helps them learn the automaton directly and to do so more
efficiently than other active learning algorithms for the same task. Finally,
in~\cite{DBLP:conf/ifm/GarhewalVHSLS20} it is assumed that constraints
satisfied along the run of a system can be made visible. They leverage this
(tainting technique) to give a scalable learning algorithm for register
automata.

\paragraph*{Structure.} The paper is structured as follows. In
\Cref{sec:Prelim}, we recall the concept of realtime one-counter
automaton and the $\LStar$ algorithm from~\cite{DBLP:journals/iandc/Angluin87}
for learning finite-state automata. We also recall what visibly
one-counter automata are and how they can be learned, as explained
in~\cite{neider2010learning}. In \Cref{sec:learningROCAs}, we present
our learning algorithm for realtime one-counter automata and its complexity in
terms of time, space and number of queries. 
In \Cref{sec:isomorphism} we prove several properties of
realtime one-counter automata useful for the correctness of our learning
algorithm. In particular, we establish an important connection with visibly
one-counter automata. In \Cref{sec:correctness}, the announced
complexity of our learning algorithm is completely justified. This is the most
technical section since, as mentioned before, we have to argue that our
extended observation tables can always be made closed and consistent.
In \Cref{sec:experiments}, we evaluate the implementation
of our algorithm on two kinds of benchmarks: random realtime one-counter
automata and a use case regarding validating JSON documents. We provide a
conclusion in \Cref{sec:conclusion}.

\section{Preliminaries}\label{sec:Prelim}

In this section we recall all necessary definitions to present our learning algorithm. In particular, we give a definition of realtime one-counter automata adapted from~\cite{DBLP:conf/alt/FahmyR95,DBLP:journals/jcss/ValiantP75} and formally define the learning task.

An \emph{alphabet} $\Sigma$ is a non-empty finite set of \emph{symbols}. A \emph{word} is a finite sequence of symbols from $\Sigma$, and the \emph{empty word} is denoted by $\emptyword$. The set of all words over $\Sigma$ is denoted by $\Sigma^*$. The \emph{concatenation} of two words $u, v \in \Sigma^*$ is denoted by $uv$. A \emph{language} $L$ is a subset of $\Sigma^*$. Given a word \(w \in \Sigma^*\) and a language \(L \subseteq \Sigma^*\), the \emph{set of prefixes of \(w\)} is $\prefixes{w} = \{u \in \Sigma^* \mid \exists v \in \Sigma^*, w = uv\}$  and the \emph{set of prefixes of \(L\)} is $\prefixes{L} = \bigcup_{w \in L} \prefixes{w}$. Similarly, we have the sets of \emph{suffixes} $\suffixes{w} = \{u \in \Sigma^* \mid \exists v \in \Sigma^*, w = vu \}$ and $\suffixes{L} = \bigcup_{w \in L} \suffixes{w}$. Moreover, $L$ is said to be \emph{prefix-closed} (resp.\ \emph{suffix-closed}) if $L = \prefixes{L}$ (resp.\ $L = \suffixes{L}$).

\begin{remark}\label{rem:nonempty}
    In this paper, we always work with non-empty languages $L$ to avoid having to treat particular cases.
\end{remark}

\subsection{Realtime one-counter automata}

\begin{definition}
	A \emph{realtime one-counter automaton (ROCA)} \(\automaton\) is a tuple \(\automaton = (Q, \Sigma, \deltaZero, \deltaNotZero, q_0, F)\) where:
	\begin{itemize}
		\item $\Sigma$ is an alphabet,
		\item \(Q\) is a non-empty finite set of states,
		\item \(q_0 \in Q\) is the initial state,
		\item \(F \subseteq Q\) is the set of final states, and
		\item \(\deltaZero\) and \(\deltaNotZero\) are two (total) transition functions defined as:
			\begin{align*}
				\deltaZero : {} & Q \times \Sigma \to Q \times \{0, +1\},\\
				\deltaNotZero : {} & Q \times \Sigma \to Q \times \{-1, 0, +1\}.
			\end{align*}
	\end{itemize}
\end{definition}

\noindent
The second component of the output of \(\deltaZero\) and \(\deltaNotZero\) gives the counter operation to apply when taking the transition. Notice that it is impossible to decrement the counter when it is already equal to zero.

A \emph{configuration} is a pair \((q, n) \in Q \times \N\), that is, it contains the current state and the current counter value. The \emph{transition relation} \(\transition\limits_{\automaton} \subseteq (Q \times \N) \times \Sigma \times (Q \times \N)\) is defined as follows:
	\[
		(q, n) \transition_{\automaton}^a (p, m) \iff \begin{cases}
			\deltaZero(q, a) = (p, c) \land m = n + c & \text{if \(n = 0\)},\\
			\deltaNotZero(q, a) = (p, c) \land m = n + c & \text{if \(n > 0\).}
		\end{cases}
	\]
When the context is clear, we omit \(\automaton\) to simply write \(\transition\limits^a\). We lift the relation to words in the natural way. Notice that this relation is \emph{deterministic} in the sense that given a configuration $(q,n)$ and a word $w$, there exists a unique configuration $(p,m)$ such that $(q,n) \transition\limits^{w} (p,m)$. 

Given a word $w$, let $(q_0,0) \transition\limits^{w} (q, n)$ be the \emph{run} on $w$. When $n = 0$ and $q \in F$, we say that this run is accepting. The \emph{language accepted} by \(\automaton\) is the set $\languageOf{\automaton} = \{w \in \Sigma^* \mid (q_0,0) \transition\limits^{w} (q, 0)$ with $q \in F\}$. If a language \(L\) is accepted by some ROCA, we say that \(L\) is a \emph{realtime one-counter language (ROCL)}.

Given \(w \in \Sigma^*\), we define the \emph{counter value} of \(w\) according to \(\automaton\), noted \(\counterAut{w}\), as the counter value $n$ of the configuration $(q,n)$ such that $(q_0, 0) \transition\limits^{w} (q, n)$. We define the \emph{height of \(w\) according to \(\automaton\)}, noted \(\heightAut{w}\), as the maximal counter value among the prefixes of \(w\), i.e., $\heightAut{w} = \max_{x \in \prefixes{w}} \counterAut{x}$.

\begin{example}\label{example:roca} 
	\begin{figure}
		\centering
		\begin{tikzpicture}[
    automaton,
    node distance=120pt,
]
    \node [state, initial]                  (q0)    {\(q_0\)};
    \node [state, accepting, right=of q0]   (q1)    {\(q_1\)};
    \node [state, accepting, right=of q1]   (q2)     {\(q_2\)};

    \path[deltaZero]
        (q0)    edge [loop above]   node {\(a, =0, +1\)}    ()
        (q0)    edge[bend right]    node {\(b, =0, 0\)} (q2)
        (q1)    edge                node [align=center] {\(a, =0, 0\)\\\(b, =0, 0\)} (q2)
        (q2) edge [loop above]    node [align=center] {\(a, =0, 0\)\\\(b, =0, 0\)}    ()
    ;
    \path [deltaNotZero]
        (q0)    edge [loop below]   node {\(a, \neq 0, +1\)}    ()
                edge                node {\(b, \neq 0, 0\)}     (q1)
        (q1)    edge [in=110,out=140,loop]   node {\(a, \neq 0, -1\)}    ()
                edge [in=40,out=70,loop]   node {\(b, \neq 0, 0\)}     (q2)
        (q2)    edge [loop below]    node [align=center] {\(a, \neq 0, 0\)\\\(b, \neq 0, 0\)}    ()
    ;
\end{tikzpicture}
		\caption{An example of an ROCA.}%
		\label{fig:example:roca}
	\end{figure}

	A 3-state ROCA \(\automaton\) over $\Sigma =\{a,b\}$ is given in \Cref{fig:example:roca}. The initial state $q_0$ is marked by a small arrow and the two final states \(q_1\) and \(q_2\) are double-circled.
    The transitions give the input symbol, the condition on the counter value, and the counter operation, in this order. In particular $\deltaZero$ is indicated in blue while $\deltaNotZero$ is indicated in green.

    Let us study the following run on \(w = aababaa\):
    \[
        (q_0, 0) \transition^a (q_0, 1) \transition^a (q_0, 2) \transition^b (q_1, 2) \transition^a (q_1, 1) \transition^b (q_1, 1) \transition^a (q_1, 0) \transition^a (q_2, 0).
    \]
    This run is accepting and thus \(aababaa\) is accepted by \(\automaton\).
    Moreover, \(\counterAut{w} = 0\) and \(\heightAut{w} = 2\).
    One can verify that $\languageOf{\automaton}$ is the language
    $\{ w \in {\{a,b\}}^* \mid \exists n \geq 0, \exists k_1, k_2, \ldots, k_n \geq 0, \exists u \in {\{a,b\}}^*, w = a^n b (b^{k_1} a b^{k_2} a \cdots b^{k_n} a) u \}$.
\qed\end{example}

\subsection{Learning deterministic finite automata}\label{subsec:Lstar}

The aim of this paper is to design a learning algorithm for ROCAs. 
Let us first recall the well-known concept of learning a deterministic finite automaton (DFA), as introduced in~\cite{DBLP:journals/iandc/Angluin87}.
Let \(L \subseteq \Sigma^*\) be a regular language. The task of the \emph{learner} is to construct a DFA \(\hypothesis\) such that \(\languageOf{\hypothesis} = L\) by interacting with the \emph{teacher}.
The two possible types of interactions are \emph{membership queries} (does \(w \in \Sigma^*\) belong to \(L\)?), and \emph{equivalence queries} (does the DFA \(\hypothesis\) accept \(L\)?). For the latter type, if the answer is negative, the teacher also provides a counterexample, i.e., a word \(w\) such that \(w \in L \iff w \notin \hypothesis\).

The so-called \emph{\LStar algorithm} of~\cite{DBLP:journals/iandc/Angluin87} (see \Cref{alg:lstar} in \Cref{app:Lstar}) learns step by step at least one representative per equivalence class of the \emph{Myhill-Nerode congruence} \(\nerodeCongruence\) of \(L\).
We recall that, for \(u, v \in \Sigma^*\),
\(u \nerodeCongruence v\) if and only if \(\forall w \in \Sigma^*, uw \in L \iff vw \in L\)~\cite{hu00}. More precisely, the learner uses an \emph{observation table $\obsTable[] = (\representatives, \separatorsCounters, \obsMapping[])$} to store his current knowledge, with:
\begin{itemize}
    \item \(\representatives \subseteq \Sigma^*\) a finite prefix-closed set of \emph{representatives},
    \item \(\separatorsCounters \subseteq \Sigma^*\) a finite suffix-closed set of \emph{separators},
    \item \(\obsMapping[] : (\representatives \cup \representatives \Sigma) \separatorsCounters \to \{0, 1\}\) such that \(\forall w \in (\representatives \cup \representatives \Sigma) \separatorsCounters, \obsMapping[](w) = 1 \iff w \in L\).
\end{itemize}
An equivalence relation over \(\representatives \cup \representatives\Sigma\) can be defined from a table \(\obsTable[]\) as follows: for any \(u, v \in \representatives \cup \representatives \Sigma\), we say that \(u \nerodeCongruence_{\obsTable[]} v\) if and only if \(\forall s \in \separatorsCounters, \obsMapping[](us) = \obsMapping[](vs)\).
Notice that \(\nerodeCongruence_{\obsTable[]}\) is coarser than \(\nerodeCongruence\). The \LStar algorithm uses membership and equivalence queries to refine \(\nerodeCongruence_{\obsTable[]}\) until it coincides with \(\nerodeCongruence\).

In order to be able to construct a DFA \(\automaton_{\obsTable[]}\) from \(\nerodeCongruence_{\obsTable[]}\), the table must satisfy two requirements: 
it must be 
\begin{enumerate}
    \item \emph{closed}: \(\forall u \in \representatives \Sigma, \exists v \in \representatives\), \(u \nerodeCongruence_{\obsTable[]} v\),
    \item \emph{\consistent}: \(\forall u, v \in \representatives, \forall a \in \Sigma\), \(u \nerodeCongruence_{\obsTable[]} v \implies ua \nerodeCongruence_{\obsTable[]} va\).
\end{enumerate}

\noindent
If \(\obsTable[]\) is not closed,
adding \(u\) to \(\representatives\) is enough to resolve the problem. If \(\obsTable[]\) is not \consistent,
this means that there exist \(a \in \Sigma\) and $s \in \separatorsCounters$ such that $\obsMapping[](uas) = 1 \iff \obsMapping[](vas) = 0$. Therefore we add \(as\) to \(\separatorsCounters\) to place $u,v$ in different equivalence classes.
The final part is how to handle the counterexample \(w\) provided by the teacher after a negative equivalence query.
The learner simply adds all the prefixes of \(w\) in \(\representatives\) and updates the table.

Angluin~\cite{DBLP:journals/iandc/Angluin87} showed that this learning process terminates and it requires a polynomial number of membership and equivalence queries in the size of the minimal DFA accepting \(L\), and in the length of the longest counterexample returned by the teacher.

\subsection{Visibly one-counter automata}\label{sec:visibly}
Since we are interested in learning ROCAs, we will modify the \LStar algorithm to match our needs. But, first, we introduce visibly one-counter automata and how to learn them~\cite{neider2010learning}, as our algorithm uses this as a sub-routine.

A visibly one-counter automaton is defined over a \emph{pushdown alphabet} \(\Sigma = \callAlphabet \cup \returnAlphabet \cup \internalAlphabet\) which is a union of three disjoint alphabets such that \(\callAlphabet\) is the set of \emph{calls} where every call increments the counter by one, \(\returnAlphabet\) is the set of \emph{returns} where every return decrements the counter by one, and \(\internalAlphabet\) is the set of \emph{internal actions} where internal action does not change the counter. We also define the \emph{sign} \(\signOf{a}\) of a symbol \(a \in \Sigma\), i.e., the counter operation it induces, as $\signOf{a} = 1$ if \(a \in \callAlphabet\), $\signOf{a} = -1$ if \(a \in \returnAlphabet\), and $\signOf{a} = 0$ if \(a \in \internalAlphabet\).

\begin{definition}[Adapted from~\cite{neider2010learning}]
    A \emph{visibly one-counter automaton (VCA)} is an ROCA $\automaton = (Q, \Sigma, \deltaZero, \deltaNotZero, q_0, F)$ such that \(\Sigma\) is a pushdown alphabet and if \((q, n) \transition\limits_{\automaton}^a (p, m)\) then 
    $m-n = \signOf{a}$.
\end{definition}

Unlike ROCAs, VCAs have the particularity that the counter operations are solely dictated by the input symbols. It is thus natural to see the transition functions\footnote{In~\cite{neider2010learning}, a VCA is defined with \(m + 1\) transition functions, with \(m\) a natural parameter. In this work, we fix \(m = 1\).} of a VCA as being of the form:
\begin{align*}
	\deltaZero : {} & Q \times \Sigma \setminus \returnAlphabet \to Q, \\
	\deltaNotZero : {} & Q \times \Sigma \to Q.
\end{align*}
We then define the \emph{counter value} of a word \(w = a_1 \dots a_n \in \Sigma^*\) as $\countervalue{w} = \sum_{i=1}^n \signOf{a_i}$. In particular, \(\countervalue{\emptyword} = 0\). The \emph{height} of $w$ is the maximal counter value of any of its prefixes, that is, $\height{w} = \max_{u \in \prefixes{w}} \countervalue{u}$.

If a language \(L \subseteq \Sigma^*\)
is accepted by some VCA, we say that \(L\) is a \emph{visibly one-counter language (VCL)}.

\begin{example}\label{example:vca}
    \begin{figure}
        \centering
        \begin{tikzpicture}[
    automaton,
    node distance=120pt,
]
    \node [state, initial]                  (q0)    {\(q_0\)};
    \node [state, accepting, right=of q0]   (q1)    {\(q_1\)};
    \node [state, accepting, right=of q1]   (q2)     {\(q_2\)};
    \path[deltaZero]
        (q0)    edge [loop above]   node {\(a_c, =0\)}    ()
        (q0)    edge[bend right]                node {\(b_{int}, =0\)} (q2)
        (q1)    edge                node [align=center] {\(a_{int}, =0\)\\\(b_{int}, =0\)} (q2)
        (q2) edge [loop below]    node [align=center] {\(a_{int}, =0\)\\\(b_{int}, =0\)}    ()
    ;
    \path [deltaNotZero]
        (q0)    edge [loop below]   node {\(a_c, \neq 0\)}    ()
                edge                node {\(b_{int}, \neq 0\)}     (q1)
        (q1)    edge [in=110,out=140,loop]   node {\(a_r, \neq 0\)}    ()
                edge [in=40,out=70,loop]   node {\(b_{int}, \neq 0\)}     ()
        (q2)    edge [loop above]    node [align=center] {\(a_{int}, \neq 0\)\\\(b_{int}, \neq 0\)}    ()
    ;
\end{tikzpicture}
        \caption{An example of a VCA.}%
        \label{fig:example:vca}
    \end{figure}

    A 3-state VCA \(\visibly{\automaton}\) over the pushdown alphabet $\visibly{\Sigma} = \{a_c\} \cup \{a_r\} \cup \{a_{int},b_{int}\}$ is given in \Cref{fig:example:vca}.
    Transitions functions give the input symbol and the condition on the counter value, in this order.
    Since counter operations are dictated by the input symbols, they are not repeated on the transitions.
    
    Notice that \(\visibly{\automaton}\) is built upon the ROCA \(\automaton\) from \Cref{fig:example:roca} such that the symbols of the transitions in \(\visibly{\automaton}\) encode the counter operations of the transitions in \(\automaton\). Notice also that the transitions functions are partial (lacking transitions to a sink state should be added). 

    By studying the sequence of configurations for \(w = a_c a_c b_{int} a_r b_{int} a_r a_{int}\), we deduce that this word is accepted by \(\visibly{\automaton}\). Moreover, \(\countervalue{w} = 0\) and \(\height{w} = 2\).
\qed\end{example}

A finite representation for the transition relation of VCLs is proposed by Neider and Löding in~\cite{neider2010learning} and used in their learning algorithm. 
To recall it, we begin by stating that two equivalent words according to the Myhill-Nerode congruence~$\nerodeCongruence$ have the same counter value, if they are in the prefix of the language.

\begin{lemma}\label{lemma:visibly:same_cv}
    Let \(L \subseteq \Sigma^*\) be a VCL and \(u, v \in \prefixes{L}\) such that \(u \nerodeCongruence v\).
    Then, \(\countervalue{u} = \countervalue{v}\).
\end{lemma}
\begin{proof}
    Let \(u, v \in \prefixes{L}\).
    Since \(u \in \prefixes{L}\), there exists \(w \in \Sigma^*\) such that \(uw \in L\) and \(vw \in L\) (as \(u \nerodeCongruence v\)).
    We have \(\countervalue{uw} = \countervalue{u} + \countervalue{w}= 0\) and \(\countervalue{vw} = \countervalue{v} + \countervalue{w} = 0\).
    We conclude that \(\countervalue{u} = \countervalue{v}\).
\qed\end{proof}

We now recall the concept of behavior graph of a VCL\@, which is the (potentially infinite) deterministic automaton classically defined from $\nerodeCongruence$.
\begin{definition}[\cite{neider2010learning}]\label{def:visibly:behavior_graph}
	Let \(L \subseteq \Sigma^*\)
	be a VCL and \(\nerodeCongruence\) be the Myhill-Nerode congruence for this language. Then, the \emph{behavior graph} of \(L\) is the tuple \(\behaviorGraph = (Q_{\nerodeCongruence}, \Sigma, \delta_{\nerodeCongruence}, q_{\nerodeCongruence}^0, F_{\nerodeCongruence})\) where:
	\begin{itemize}
		\item \(Q_{\nerodeCongruence} =
			  \{\equivalenceClass{w}_{\nerodeCongruence} \mid w \in \prefixes{L}\}\) is the set of states,
	    \item \(q_{\nerodeCongruence}^0 = \equivalenceClass{\emptyword}_{\nerodeCongruence}\) is the initial
		      state,
        \item \(F_{\nerodeCongruence} = \{\equivalenceClass{w}_{\nerodeCongruence} \mid w \in L\}\) is the set of final states,
		\item \(\delta_{\nerodeCongruence} : Q_{\nerodeCongruence} \times \Sigma \to Q_{\nerodeCongruence}\) is the partial transition function defined by $\delta(\equivalenceClass{w}_{\nerodeCongruence}, a) = \equivalenceClass{wa}_{\nerodeCongruence}$ for all $\equivalenceClass{w}_{\nerodeCongruence}, \equivalenceClass{wa}_{\nerodeCongruence} \in Q_{\nerodeCongruence}$ and $a \in \Sigma$.
	\end{itemize}
\end{definition}

Note that $Q_{\nerodeCongruence}$ is not empty as we assume the languages we consider are non-empty. Note also that by definition, the states of $\behaviorGraph$ are all reachable from the initial state and co-reachable from some final state\footnote{In the sequel, we simply speak about reachable and co-reachable states and we say that the automaton is \emph{trim}.}, hence why the transition function is partial.\footnote{The definition of behavior graph from~\cite{neider2010learning} is different since non-co-reachable states are allowed. However the two coincide on reachable and co-reachable states.} We have $\languageOf{\behaviorGraph} = L$.

It is proved in~\cite{neider2010learning} that the behavior graph \(\behaviorGraph\) of \(L\) always has a finite representation, even though the graph itself is potentially infinite.
This finite representation relies on the fact that \(\behaviorGraph\) has an ultimately periodic structure, i.e., it has an \enquote{initial part} that is followed by a \enquote{repeating part} repeated ad infinitum.
We give here a short overview of the idea (details can be found in~\cite{neider2010learning}).

As all words in the same equivalence class in \(\prefixes{L}\) have the same counter value (by \Cref{lemma:visibly:same_cv}), we define the \emph{level} \(\ell\) of \(\behaviorGraph\) as the set of states with counter value \(\ell\), i.e., $\{\equivalenceClass{w}_{\nerodeCongruence} \in Q_{\nerodeCongruence} \mid \countervalue{w} = \ell\}$.  The first observation made in~\cite{neider2010learning} is that the number of states in each level of \(\behaviorGraph\) is bounded by a constant \(K \in \N\), in particular by the number of states of any VCA accepting $L$. The minimal value of \(K\) is called the \emph{width} of \(\behaviorGraph\). This observation allows to enumerate the states in level \(\ell\) using a mapping \(\nu_{\ell} : \{\equivalenceClass{w}_{\nerodeCongruence} \in Q_{\nerodeCongruence} \mid \countervalue{w} = \ell\} \to \{1, \dotsc, K\}\).
Using these enumerations $\nu_{\ell}$, $\ell \in \N$, we can encode the transitions of \(\behaviorGraph\) as a sequence of (partial) mappings \(\tau_{\ell} : \{1, \dots, K\} \times \Sigma \to \{1, \dots, K\}\) (with \(\ell \in \N\)).
For all \(\ell \in \N, i \in \{1, \dots, K\}\), and \(a \in \Sigma\), the mapping \(\tau_{\ell}\) is defined as
\[
    \tau_{\ell}(i, a) = \begin{cases}
        j  & \parbox[t]{.6\textwidth}{if there exist \(\equivalenceClass{u}_{\nerodeCongruence}, \equivalenceClass{ua}_{\nerodeCongruence} \in Q_{\nerodeCongruence}\) such that \(\countervalue{u} = \ell, \nu_{\ell}(\equivalenceClass{u}_{\nerodeCongruence}) = i\), and \(\nu_{\ell + \signOf{a}}(\equivalenceClass{ua}_{\nerodeCongruence}) = j\),}\\
        \text{undefined} & \text{otherwise.}
    \end{cases}
\]
Thus, if we fix the enumerations \(\nu_{\ell}\), the behavior graph can be encoded as the sequence of mappings  \(\alpha = \tau_0 \tau_1 \tau_2 \dotso\), called a \emph{description} of \(\behaviorGraph\).
The following theorem states that there always is such a description which is \emph{periodic}.

\begin{theorem}[{\cite[Theorem 1]{neider2010learning}}]\label{thm:visibly:behavior_graph:periodicity}
    Let \(L \subseteq \Sigma^*\)
    be a VCL, \(\behaviorGraph\) be the behavior graph of \(L\), and \(K\) be the width of \(\behaviorGraph\).
    Then, there exist enumerations \(\nu_{\ell} : \{\equivalenceClass{u}_{\nerodeCongruence} \in Q_{\nerodeCongruence} \mid \countervalue{u} = \ell\} \to \{1, \dotsc, K\}\), $\ell \in \N$, such that the corresponding description \(\alpha\) of \(\behaviorGraph\) is an ultimately periodic word with offset \(m > 0\) and period \(k \geq 0\), i.e., \(\alpha = \tau_0 \dotso \tau_{m-1} {(\tau_m \dotso \tau_{m+k-1})}^{\omega}\).
\end{theorem}

Conversely, from a periodic description of \(\behaviorGraph\), it is possible to construct a VCA accepting \(L\)~\cite[Lemma 1]{neider2010learning}.

\begin{example}\label{example:vca:behavior_graph}
    \begin{figure}
        \centering
        \begin{tikzpicture}[
    automaton,
    node distance=52pt and 72pt,
    state/.append style = {
        minimum size = 1cm,
    },
    pin distance=0pt,
    pin position=above right,
    every pin edge/.style = {
        draw=none,
    }
]
    \node [state, initial, pin=1]                   (eps)   {\(\emptyword\)};
    \node [state, right=of eps, pin=1]              (a)     {\(a_c\)};
    \node [state, right=of a, pin=1]                (aa)    {\(a_ca_c\)};
    \node [state, right=of aa, pin=1]               (aaa)   {\(a_ca_ca_c\)};
    \node [right=of aaa]                            (aaaa)  {\(\ldots\)};
    \node [state, accepting, below=of eps, pin=2]   (b)     {\(b_{int}\)};
    \node [state, right=of b, pin=2]                (ab)    {\(a_cb_{int}\)};
    \node [state, right=of ab, pin=2]               (aab)   {\(a_ca_cb_{int}\)};
    \node [state, right=of aab, pin=2]              (aaab)  {\(a_ca_ca_cb_{int}\)};
    \node [right=of aaab]                           (aaaab) {\(\ldots\)};
    \path
        foreach \s/\l/\t in {eps/a_c/a, a/a_c/aa, aa/a_c/aaa, aaa/a_c/aaaa, ab/a_r/b, aab/a_r/ab, aaab/a_r/aab, aaaab/a_r/aaab} {
            (\s)    edge                node {\(\l\)}   (\t)
        }
        foreach \s/\l/\t in {eps/b_{int}/b, a/b_{int}/ab, aa/b_{int}/aab, aaa/b_{int}/aaab} {
            (\s)    edge                node [near start] {\(\l\)} (\t)
        }
        (b)         edge [loop below]   node {\(a_{int}, b_{int}\)} (b)
        foreach \s/\l in {ab/b_{int}, aab/b_{int}, aaab/b_{int}} {
            (\s)    edge [loop below]   node {\(\l\)}   ()
        }
    ;
    \begin{pgfonlayer}{background}
        \filldraw [join=round, black!30]
            ($(eps.north west)+(-0.8, 0.6)$) rectangle ($(b.south east)+(0.5, -1.2)$)
        ;
        \filldraw [join=round, black!10]
            foreach \s/\t in {a/ab, aa/aab, aaa/aaab} {
                ($(\s.north west)+(-0.6, 0.6)$) rectangle ($(\t.south east)+(0.5, -1.2)$)
            }
        ;
    \end{pgfonlayer}
    \draw [-,thick]
        ($(eps.north west)+(-1.3, 1)$)
            node [above right] {Initial part}
            -| ($(eps.north)+(0, 0.43)$)
        ($(a.north east)+(3, 1)$)
            -| ($(a.north)+(0, 0.43)$)
        ($(aa.north east)+(3, 1)$)
            node [above left] {Repeating part}
            -| ($(aa.north)+(0, 0.43)$)
        ($(aaa.north east)+(2, 1)$)
            -| ($(aaa.north)+(0, 0.4)$)
    ;
\end{tikzpicture}
        \caption{Behavior graph of the language accepted by the VCA from \Cref{fig:example:vca}.}%
        \label{fig:example:vca:behavior_graph}
    \end{figure}

    Let $\visibly{L}$ be the language accepted by the VCA $\visibly{\automaton}$ of \Cref{fig:example:vca}. Its behavior graph \(\behaviorGraph[\visibly{L}]\) is given in \Cref{fig:example:vca:behavior_graph}.
    For clarity sake, each state $\equivalenceClass{w}_{\nerodeCongruence}$ is represented by one of its words, and all states in the same level are aligned in the same column.
    For instance, states $\emptyword$ and $b_{int}$ constitute level $0$ as \(\countervalue{\emptyword} = \countervalue{b_{int}} = 0\).
    We can see that \(\behaviorGraph[\visibly{L}]\) has a width \(K = 2\).

    Let us now explain why \(\behaviorGraph[\visibly{L}]\) has a finite representation. 
    We define the following enumerations \(\nu_{\ell}\), $\ell \in \N$:
    \begin{align*}
        \nu_0(\equivalenceClass{\emptyword}) &= 1   &   \nu_0(\equivalenceClass{b_{int}}) &= 2\\
        \nu_1(\equivalenceClass{a_c}) &= 1          &   \nu_1(\equivalenceClass{a_c b_{int}}) &= 2\\
        \nu_2(\equivalenceClass{a_c a_c}) &= 1      &   \nu_2(\equivalenceClass{a_c a_c b_{int}}) &= 2\\
        &\cdots                                     &   & \cdots
    \end{align*}
    In the figure, each state has its associated number next to it.
    From these enumerations, we construct the following \(\tau_{\ell}\) mappings (the missing values are undefined):
    \begin{align*}
        \tau_0(1, a_c) &= 1 & \tau_0(1, b_{int}) &= 2   & \tau_0(2, a_{int}) &= 2   & \tau_0(2, b_{int}) &= 2\\
        \tau_1(1, a_c) &= 1 & \tau_1(1, b_{int}) &= 2   & \tau_1(2, a_r) &= 2       & \tau_1(2, b_{int}) &= 2\\
        \tau_2(1, a_c) &= 1 & \tau_2(1, b_{int}) &= 2   & \tau_2(2, a_r) &= 2       & \tau_2(2, b_{int}) &= 2\\
        &\cdots             & &\cdots                   & &\cdots                   & &\cdots
    \end{align*}
    In this way we get a description \(\alpha = \tau_0 \tau_1 \tau_2 \dotso\) of \(\behaviorGraph[\visibly{L}]\).
    As \(\tau_1 = \tau_2 = \tau_3 = \dotso\), we have that $\alpha = \tau_0 {(\tau_1)}^\omega$ is a periodic description of \(\behaviorGraph[\visibly{L}]\).
\qed\end{example}

\subsection{Learning VCAs}\label{subsec:LearningVCAs}

We now discuss the \LStar-inspired learning algorithm for VCAs, introduced in~\cite{neider2010learning}.
Let \(L\) be some VCL\@. 
The idea is to learn an initial fragment of the behavior graph of \(L\) up to a fixed counter limit \(\ell\), to extract every possible periodic description of the fragment, and to construct a VCA from each of these descriptions.
If we find one VCA accepting \(L\), we are done. Otherwise, we increase \(\ell\) and repeat the process. 

Formally, the initial fragment of \(\behaviorGraph\) up to \(\ell\) is a subgraph of \(\behaviorGraph\).
That is, it is the DFA \(\limitedBehaviorGraph = (Q, \Sigma, \delta, q_0, F)\) with \(Q = \{\equivalenceClass{w}_{\nerodeCongruence} \in Q_{\nerodeCongruence} \mid \height{w} \leq \ell\}\).
The values of the (partial) function \(\delta\), the initial state, and the final states are naturally defined over the subgraph. This DFA is called the \emph{limited behavior graph} and \(L_{\ell}\) denotes the language accepted by \(\limitedBehaviorGraph\).

The data structure of the learner is again an observation table.
Taking into account a counter-value limit of \(\ell\), we define the table as \(\obsTable = (\representatives, \separatorsCounters, \obsMapping)\) with:\footnote{In~\cite{neider2010learning}, a so-called \emph{stratified} observation table is used instead. We present a simple table to match the rest of the paper.}
\begin{itemize}
    \item \(\representatives \subseteq \Sigma^*\) a finite prefix-closed set of representatives,
    \item \(\separatorsCounters \subseteq \Sigma^*\) a finite suffix-closed set of separators,
    \item \(\obsMapping : (\representatives \cup \representatives \Sigma) \separatorsCounters \to \{0, 1\}\) such that \(\forall w \in (\representatives \cup \representatives\Sigma) \separatorsCounters, \obsMapping(w) = 1 \iff w \in L_{\ell}\).
\end{itemize}
Intuitively the aim of this table is to approximate the equivalence classes of $\nerodeCongruence_{L_{\ell}}$ to learn the limited behavior graph $\limitedBehaviorGraph$. 
On top of the membership queries (does $w$ belong to $L$?) and equivalence queries (does the learned ROCA accept $L$?), we add \emph{partial equivalence queries}: given \(\ell\), does the DFA \(\hypothesis\) accept \(L_{\ell}\)?
This allows us to check whether the initial fragment is correctly learned. 

Suppose that the limited behavior graph \(\limitedBehaviorGraph\) has been learned. 
Extracting every possible periodic description of \(\limitedBehaviorGraph\) can be performed by identifying an isomorphism between two consecutive subgraphs of \(\limitedBehaviorGraph\).
That is, we fix values for the offset \(m\) and period \(k\) and see if the subgraphs induced by the levels \(m\) to \(m + k - 1\), and by the levels \(m + k\) to \(m + 2k - 1\) are isomorphic. Note that this means we need to consider all pairs of \(m\) and \(k\) such that \(m + 2k - 1 \leq \ell\).
This can be done in polynomial time by executing two depth-first searches in parallel~\cite{neider2010learning}.
Note that multiple periodic descriptions may be found, due to the finite knowledge of the learner.

This procedure yields potentially many different VCAs (one per description).
We ask an equivalence query for each of them.
If one of the VCAs accepts the target language $L$, we are done.
Otherwise, we need to refine the table and increase the counter limit $\ell$.
However, not all of the counterexamples given by the teacher can be used.
Indeed, it may happen that the teacher returns a word that was incorrectly accepted or rejected by a VCA, but for which the information is already present in the table (due to a description that covered only the first levels in \(\limitedBehaviorGraph\), for instance).
Therefore, we only consider counterexamples \(w\) with a height \(\height{w} > \ell\).
If we cannot find such a counterexample, we use \(\limitedBehaviorGraph\) directly as a VCA (where only \(\deltaZero\) is used).
Since we know that \(\languageOf{\limitedBehaviorGraph} = L_{\ell}\), we are guaranteed to obtain a useful counterexample.

In~\cite[Theorem 2]{neider2010learning}, it is shown that this algorithm (see \Cref{alg:lstar_vca}) has polynomial time and space complexity in the width of the behavior graph, the offset and period of a periodic description of the behavior graph.

\begin{algorithm} 
    \caption{Learning a VCA~\cite{neider2010learning}}%
    \label{alg:lstar_vca}
    \begin{algorithmic}[1]
        \Require The target VCL \(L\)
        \Ensure A VCA accepting \(L\) is returned
        \Statex
        \State Initialize the observation table \(\obsTable\) with \(\ell = 0, \representatives = \separatorsCounters = \{\emptyword\}\)
        \While{true}
            \State Make $\obsTable$ closed and \consistent
            \State Construct the DFA \(\autTable\) from \(\obsTable\)
            \State Ask a partial equivalence query over \(\autTable\)
            \If{the answer is negative}
                \State Update \(\obsTable\) with the provided counterexample \Comment{\(\ell\) is not modified}
            \Else
                \State \(W \gets \emptyset\)
                \ForAll{periodic descriptions \(\alpha\) of \(\autTable\)}
                    \State Construct the VCA \(\automaton_{\alpha}\) from \(\alpha\)
                    \State Ask an equivalence query over \(\automaton_{\alpha}\)
                    \If{the answer is positive}
                        \State \Return \(\automaton_{\alpha}\)
                    \ElsIf{\(\height{w} > \ell\) for the provided counterexample $w$}
                        \State \(W \gets W \cup \{w\}\)
                    \EndIf
                \EndFor
                \If{\(W\) is empty}
                    \State Ask an equivalence query over \(\autTable\) seen as a VCA
                    \If{the answer is positive} \Comment{\(L\) is regular}
                        \State \Return \(\autTable\)
                    \Else \State \(W \gets W \cup \{w\}\) with $w$ the provided counterexample
                    \EndIf
                \EndIf
                \State Update \(\observationTable\) with an arbitrary element from \(W\) \Comment{\(\ell\) is increased}
            \EndIf
        \EndWhile
    \end{algorithmic}
\end{algorithm}

\subsection{Behavior graph of an ROCA}\label{subsec:BG_ROCA}

In this section, we introduce the concept of behavior graph of an ROCA $\automaton$ and we present two interesting properties of this graph. In particular the behavior graph of $\automaton$ will have an ultimately periodic description that will be useful for the learning of ROCAs.

We first introduce a congruence relation from which the behavior graph of an ROCA will be defined.

\begin{definition}
    Let \(\automaton\) be an ROCA accepting the language $L = \languageOf{\automaton}$ over $\Sigma$.
    We define the congruence relation $\equivBGROCA$ over $\Sigma^*$ such that $u \equivBGROCA v$ if and only if:
    \begin{eqnarray*}
		\forall w \in \Sigma^* &:& uw \in L \iff vw \in L,\\
		\forall w \in \Sigma^* &:& uw,vw \in \prefixes{L} \implies \counterAut{uw} = \counterAut{vw}.
	\end{eqnarray*}
	The equivalence class of $u$ is denoted by $\equivalenceClass{u}_{\equivBGROCA}$.
\end{definition}

The congruence relation $\equivBGROCA$ is a refinement of the Myhill-Nerode congruence. It is easy to check that $\equivBGROCA$ is indeed a congruence relation. Finally, it is noteworthy that the second condition depends on the ROCA \(\automaton\) and is limited to \(\prefixes{L}\) because even if \(\automaton\) has different counter values for words not in \(\prefixes{L}\), we still require all those words to be equivalent.
From this relation, one can define the following (possibly infinite) deterministic automaton.

\begin{definition}\label{def:behavior_graph}
    Let $\automaton = (Q, \Sigma, \deltaZero, \deltaNotZero, q_0, F)$ be an ROCA accepting $L = \languageOf{\automaton}$ over $\Sigma$. The \emph{behavior graph of} $\automaton$ is $\behaviorGraphAut = (Q_\equivBGROCA,\Sigma,\delta_\equivBGROCA,q^0_\equivBGROCA,F_\equivBGROCA)$ where:
    	\begin{itemize}
		\item $Q_\equivBGROCA = \{\equivalenceClass{u}_\equivBGROCA \mid u \in \prefixes{L}\}$ is the set of states,
		\item $q^0_\equivBGROCA = \equivalenceClass{\emptyword}_\equivBGROCA$ is the initial state,
		\item $F_\equivBGROCA = \{\equivalenceClass{u}_\equivBGROCA \mid u \in L\}$ is the set of final states,
		\item $\delta_\equivBGROCA : Q_\equivBGROCA \times \Sigma \to Q_\equivBGROCA$ is the partial transition function defined by: $\delta_\equivBGROCA(\equivalenceClass{u}_\equivBGROCA,a) = \equivalenceClass{ua}_\equivBGROCA$, for all $\equivalenceClass{u}_\equivBGROCA, \equivalenceClass{ua}_\equivBGROCA \in Q_\equivBGROCA$ and $a \in \Sigma$.
		\end{itemize}
\end{definition}

Note that $Q_{\equivBGROCA}$ is not empty by \Cref{rem:nonempty}. A straightforward induction shows that $\behaviorGraphAut$ accepts the same language $L$ as $\automaton$. Note also that by definition, $\behaviorGraph$ is trim which implies that the transition function is partial. 

\begin{example}\label{example:behavior_graph}
    \begin{figure}
        \centering
        \begin{tikzpicture}[
    automaton,
    node distance=2cm,
    state/.append style = {
        minimum size = 1cm,
    },
]
    \node [state, initial]                  (eps)   {\(\emptyword\)};
    \node [state, right=of eps]             (a)     {\(a\)};
    \node [state, right=of a]               (aa)    {\(aa\)};
    \node [state, right=of aa]              (aaa)   {\(aaa\)};
    \node [right=of aaa]                    (aaaa)  {\(\ldots\)};
    \node [state, accepting, below=of eps]  (b)     {\(b\)};
    \node [state, right=of b]               (ab)    {\(ab\)};
    \node [state, right=of ab]              (aab)   {\(aab\)};
    \node [state, right=of aab]             (aaab)  {\(aaab\)};
    \node [right=of aaab]                   (aaaab) {\(\ldots\)};
    \path
        foreach \s/\l/\t in {eps/a/a, eps/b/b, a/a/aa, a/b/ab, aa/a/aaa, aa/b/aab, aaa/a/aaaa, aaa/b/aaab, ab/a/b, aab/a/ab, aaab/a/aab, aaaab/a/aaab} {
            (\s)    edge                node {\(\l\)}   (\t)
        }
        (b)         edge [loop below]   node {\(a, b\)} (b)
        foreach \s/\l in {ab/b, aab/b, aaab/b} {
            (\s)    edge [loop below]   node {\(\l\)}   ()
        }
    ;
    \begin{pgfonlayer}{background}
        \filldraw [join=round, black!30]
            ($(eps.north west)+(-0.8, 0.6)$) rectangle ($(b.south east)+(0.5, -1.2)$)
        ;
        \filldraw [join=round, black!10]
            foreach \s/\t in {a/ab, aa/aab, aaa/aaab} {
                ($(\s.north west)+(-0.5, 0.6)$) rectangle ($(\t.south east)+(0.5, -1.2)$)
            }
        ;
    \end{pgfonlayer}
\end{tikzpicture}
        \caption{Behavior graph of the ROCA from \Cref{fig:example:roca}.}%
        \label{fig:example:behavior_graph}
    \end{figure}

    Let us come back to the ROCA $\automaton$ from \Cref{fig:example:roca}. Recall the language accepted by $\automaton$ is equal to $L = \{ u \in {\{a,b\}}^* \mid \exists n \geq 0, \exists k_1, k_2, \ldots, k_n \geq 0, \exists w \in {\{a,b\}}^*, u = a^n b (b^{k_1} a b^{k_2} a \cdots b^{k_n} a) w \}$.
    We can check that \(b \equivBGROCA abba\). Indeed \(\forall w \in \Sigma^*, b w \in L \iff abba w \in L\).
    Moreover, $\forall w \in \Sigma^*$ such that $bw, abbaw \in \prefixes{L}$, we have \(\counterAut{bw} = \counterAut{abbaw}\). 
    However, \(ab \not\equivBGROCA aab\) since \(ab\) and \(aab\) are both in \(\prefixes{L}\) but \(\counterAut{ab} = 1 \neq \counterAut{aab} = 2\).
    
    The behavior graph \(\behaviorGraphAut\) of \(\automaton\) is given in \Cref{fig:example:behavior_graph}. As in \Cref{example:vca:behavior_graph}, this behavior graph is ultimately periodic and has thus a finite representation.
    We now proceed on stating that this observation is true for the behavior graph of any ROCA\@.
\qed\end{example}

We use the same terminology of width $K$ for the behavior of an ROCA as introduced for VCAs.
Moreover, it holds that \(K \leq \sizeOfSet{Q}\) (where \(Q\) is the set of states of \(\automaton\)).
\begin{lemma}\label{lem:sizeLimBGAut}
    The width $K$ of \(\behaviorGraphAut\) is bounded by $\sizeOfSet{Q}$.
\end{lemma}

\begin{proof}
    For all $w_1,w_2$ such that $(q_0,0) \transition\limits^{w_1}_{\automaton} (q, n)$ and $(q_0,0) \transition\limits^{w_2}_{\automaton} (q, \ell)$ for the same configuration $(q,\ell)$, we have $w_1 \equivBGROCA w_2$. It follows that $K \leq \sizeOfSet{Q}$.
\qed\end{proof}

To get a description of \(\behaviorGraphAut\), we proceed as in \Cref{sec:visibly} by defining an enumeration $\nu_{\ell}$ of each level $\ell$ of \(\behaviorGraphAut\), and then encoding the transitions of \(\behaviorGraphAut\) as a sequence of mappings $\tau_{\ell}$, $\ell \in \N$. However, since we cannot deduce the counter operation from the symbol, we need to explicitly store the target counter value for each \(\tau_{\ell}\).
That is, for each $\ell \in \N$, the mapping \(\tau_{\ell} : \{1, \dotsc, K\} \times \Sigma \to \{1, \dotsc, K\} \times \{-1, 0, +1\}\) is defined as: for all \(i \in \{1, \dots, K\}\), \(a \in \Sigma\),
\[
    \tau_{\ell}(i, a) = \begin{cases}
        (j,c)  & \parbox[t]{.6\textwidth}{if there exist \(\equivalenceClass{u}_{\equivBGROCA}, \equivalenceClass{ua}_{\equivBGROCA} \in Q_{\equivBGROCA}\) such that \\ \(\counterAut{u} = \ell, ~\counterAut{ua} = \ell + c, \\ \nu_{\ell}(\equivalenceClass{u}_{\equivBGROCA}) = i\),~\(\nu_{\ell + c}(\equivalenceClass{ua}_{\equivBGROCA}) = j\),}\\
        \text{undefined} & \text{otherwise.}
    \end{cases}
\]
In this way, we obtain a description $\alpha = \tau_0 \tau_1 \tau_2 \dotso$ of $\behaviorGraphAut$.

The next theorem is the counterpart of \Cref{thm:visibly:behavior_graph:periodicity} for ROCAs. It follows from \Cref{thm:visibly:behavior_graph:periodicity} and an isomorphism we establish between the behavior graph of an ROCA $\mathcal{A}$ and that of a suitable VCA constructed from $\mathcal{A}$.
This isomorphism is detailed in \Cref{sec:isomorphism} (see \Cref{thm:behavior_graph:isomorphism_to_visibly} below).

\begin{theorem}\label{thm:behavior_graph:periodicity}
    Let \(\automaton\) be an ROCA, \(\behaviorGraphAut = (Q_{\equivBGROCA}, \Sigma, \delta_{\equivBGROCA}, q_{\equivBGROCA}^0, F_{\equivBGROCA})\) be the behavior graph of \(\automaton\), and \(K\) be the width of \(\behaviorGraphAut\).
    Then, there exists a sequence of enumerations \(\nu_{\ell} : \{\equivalenceClass{u}_{\equivBGROCA} \in Q_{\equivBGROCA} \mid \counterAut{u} = \ell\} \to \{1, \dots, K\}\) such that the corresponding description \(\alpha\) of \(\behaviorGraphAut\) is an ultimately periodic word with offset \(m > 0\) and period \(k \geq 0\), i..e, \(\alpha = \tau_0 \dotso \tau_{m-1} {(\tau_m \dotso \tau_{m+k-1})}^\omega\).
\end{theorem}

Given a periodic description of a behavior graph $\behaviorGraphAut$
allows us to construct an ROCA that accepts the same language. The proof is also delayed in \Cref{sec:isomorphism}.

\begin{proposition}\label{prop:DescriptionAut}
    Let \(\automaton\) be an ROCA accepting a language \(L \subseteq \Sigma^*\), \(\behaviorGraphAut\) be its behavior graph of width \(K\), \(\alpha = \tau_0 \dotso \tau_{m-1} {(\tau_m \dotso \tau_{m+k-1})}^\omega\) be a periodic description of \(\behaviorGraphAut\) with offset \(m\) and period \(k\).
    Then, from $\alpha$, it is possible to construct an ROCA \(\automaton_{\alpha}\) accepting \(L\) such that the size of \(\automaton_{\alpha}\) is polynomial in \(m, k\) and \(K\).
\end{proposition}

\section{Learning ROCAs}\label{sec:learningROCAs}

In this section, let us fix a language \(L \subseteq \Sigma^*\) and an ROCA \(\automaton\) such that \(\languageOf{\automaton} = L\).
We are going to explain how a learner will learn $L$ by querying a teacher. More precisely, the learner will learn a periodic description of the behavior graph $\behaviorGraphAut$ of $\automaton$ from which he will derive an ROCA accepting $L$ (by \Cref{prop:DescriptionAut}). The teacher has to answer four different types of queries:

\begin{enumerate}
    \item \emph{Membership} query. Let \(w \in \Sigma^*\), is \(w \in L\)?
    \item \emph{Counter value} query (when the word is in $\prefixes{L}$). Let \(w \in \prefixes{L}\), what is \(\counterAut{w}\)?
    \item \emph{Partial equivalence} query. Let $\automaton[H]$ be a DFA over \(\Sigma\) and \(\ell \in \N\), is it true that \(\languageOf{\automaton[H]} = L_{\ell}\)?
    \item \emph{Equivalence} query. Let $\automaton[B]$ be an ROCA over \(\Sigma\), is it true that \(\languageOf{\automaton[B]} = L\)?
\end{enumerate}

Note that counter-value queries were not used in the VCA-learning algorithm (\Cref{alg:lstar_vca}). Indeed, words over a pushdown alphabet can be assigned a counter value without an automaton. For general alphabets, the counter value cannot be inferred from a word and we must then ask the teacher for this information. As in the VCA-learning algorithm, we will use partial-equivalence queries to find the basis of a periodic description for the target automaton.

Our main result is the following theorem.

\begin{theorem}\label{thm:main}
Let $\automaton$ be an ROCA accepting a language \(L \subseteq \Sigma^*\). Given a teacher for $L$, which answers membership, counter value, partial equivalence and equivalence queries, an ROCA accepting $L$ can be computed in time and space exponential in $|Q|,|\Sigma|$ and $\lengthCe$, where $Q$ is the set of states of $\automaton$ and $\lengthCe$ is the length of the longest counterexample returned by the teacher on (partial) equivalence queries. The learner asks $\complexity(\lengthCe^3)$ partial equivalence queries, $\complexity(|Q| \lengthCe^2)$ equivalence queries and a number of membership (resp.\ counter value) queries which is exponential in $|Q|,|\Sigma|$ and $\lengthCe$.   
\end{theorem}

In the following subsections, we present the main steps of our learning algorithm. Leveraging the isomorphism stated in \Cref{thm:behavior_graph:isomorphism_to_visibly}, we give an algorithm with the same structure as the one presented for VCAs in \Cref{subsec:LearningVCAs} (see \Cref{alg:lstar_vca}). Given an ROCA $\automaton$ accepting $L$, we learn an initial fragment of the behavior graph $\behaviorGraphAut$ up to a fixed counter limit \(\ell\). When $\ell$ is large enough, it will be possible to extract a periodic description for \(\behaviorGraphAut\) from this fragment and then to construct an ROCA accepting $L$.

The initial fragment up to $\ell$ is called the \emph{limited behavior graph} $\limBehaviorGraphAut$. This is the subgraph of \(\behaviorGraphAut\) whose set of states is \(\{\equivalenceClass{w}_{\equivBGROCA} \in Q_{\equivBGROCA} \mid \heightAut{w} \leq \ell\}\). This DFA accepts the language \(L_{\ell} = \{w \in L \mid \forall x \in \prefixes{w}, 0 \leq \counterAut{x} \leq \ell\}\). Notice that $\limBehaviorGraphAut$ is composed of the first $\ell + 1$ levels of $\behaviorGraphAut$ (from $0$ to $\ell$) such that each level is restricted to states $\equivalenceClass{w}_{\equivBGROCA}$ with $\heightAut{w} \leq \ell$. In this way all its states are reachable.

We have the next lemma.

\begin{lemma}\label{lem:sizeEquiv}
The limited behavior graph $\limBehaviorGraphAut$ has at most $(\ell + 1)|Q|$ states. The number of equivalence classes $\equivalenceClass{w}_{\equivBGROCA}$ with $\heightAut{w} \leq \ell$ is bounded by $(\ell + 1)|Q| + 1$. 
\end{lemma}
\begin{proof}
By \Cref{lem:sizeLimBGAut}, each level of $\behaviorGraphAut$ has at most $|Q|$ states. Therefore $\limBehaviorGraphAut$ has at most $(\ell+1)|Q|$ states. We have also that the number of equivalence classes of
$\equivBGROCA$ up to counter limit $\ell$ is bounded by
$(\ell+1)|Q| + 1$ (all the words in $\Sigma^* \setminus \prefixes{L}$ are in one additional class).
\qed\end{proof}

In what follows we describe all the elements our algorithm requires to learn $\limBehaviorGraphAut$.

\subsection{Observation table and approximation sets}

As for learning VCAs, we use an observation table to store the data gathered during the learning process. This table aims at approximating the equivalence classes of $\equiv$ and therefore stores information about both membership to $L$ and counter values (for words known to be in $\prefixes{L}$). It depends on a counter limit $\ell \in \N$ since we first want to learn $\limBehaviorGraphAut$. We highlight the fact that our table uses two sets of suffixes: $\separatorsMapping$ and $\separatorsCounters$. Intuitively, we use the former to store membership information and the latter for counter value information (a formal explanation and an example follow). 

\begin{definition}\label{def:observation_table}
    Let $\ell \in \mathbb{N}$ be a counter limit and $L_\ell$ be the language accepted by the limited behavior graph $\limBehaviorGraphAut$ of an ROCA $\automaton$. An \emph{observation table} $\obsTable$ up to $\ell$ is a tuple $(\representatives, \separatorsCounters, \separatorsMapping, \obsMapping, \obsCounters)$ with:
    \begin{itemize}
          \item a finite prefix-closed set \(\representatives \subseteq \Sigma^*\) of \emph{representatives},
          \item two finite suffix-closed sets $\separatorsCounters, \separatorsMapping $ of \emph{separators} such that \(\separatorsCounters \subseteq \separatorsMapping \subseteq \Sigma^*\),
          \item a function \(\obsMapping : (\representatives \cup \representatives \Sigma) \separatorsMapping \to \{0, 1\}\),
          \item a function \(\obsCounters : (\representatives \cup \representatives \Sigma) \separatorsCounters \to \{\bot, 0, \dotsc, \ell\}\).
    \end{itemize}
    Let $\prefTable$ be the set $\{w \in \prefixes{us} \mid u \in \representatives \cup \representatives\Sigma, s \in \separatorsMapping, \obsMapping(us) = 1\}$. Then for all $u \in \representatives \cup \representatives \Sigma$ the following holds:
    \begin{itemize}
        \item for all $s \in \separatorsMapping$, $\obsMapping(us)$ is $1$ if $us \in L_\ell$ and $0$ otherwise,
        \item for all $s \in \separatorsCounters$,
        $\obsCounters(us)$ is $\counterAut{us}$ if $us \in \prefTable$ and $\bot$ otherwise.
    \end{itemize}
\end{definition}

In the previous definition, notice that the domains of $\obsMapping$ and $\obsCounters$ are different. Notice also that $\prefTable \subseteq \prefixes{L_\ell}$. 
We again highlight the fact that our data structure is different from the one used in~\cite{neider2010learning}.
The difference lies in our use of two sets of separators, due to the fact that we can not compute the counter value of a word directly from its symbols, but we need to ask a query to the teacher.

\begin{example}\label{example:observation_table}
    \begin{figure}
        \centering
        \begin{tabular}{l|c|cc} 
                    & \(\emptyword\)    & \(a\)         & \(ba\)\\
    \hline 
    \(\emptyword\)  & \(0, 0\)          & \(0\)         & \(1\)\\
    \(a\)           & \(0, 1\)          & \(0\)         & \(1\)\\
    \(ab\)          & \(0, 1\)          & \(1\)         & \(1\)\\
    \(aba\)         & \(1, 0\)          & \(1\)         & \(1\)\\
    \(aa\)          & \(0, \bot\)       & \(0\)         & \(0\)\\
    \hline 
    \(b\)           & \(1, 0\)          & \(1\)         & \(1\)\\
    \(abb\)         & \(0, 1\)          & \(1\)         & \(1\)\\
    \(abaa\)        & \(1, 0\)          & \(1\)         & \(1\)\\
    \(abab\)        & \(1, 0\)          & \(1\)         & \(1\)\\
    \(aaa\)         & \(0, \bot\)       & \(0\)         & \(0\)\\
    \(aab\)         & \(0, \bot\)       & \(0\)         & \(0\)
\end{tabular}
        \caption{Example of an observation table.}%
        \label{fig:example:observation_table}
    \end{figure}

In this example, we give an observation table $\obsTable$ for the ROCA $\automaton$ from \Cref{fig:example:roca} and the counter limit $\ell = 1$, see \Cref{fig:example:observation_table}. Hence we want to learn $\limBehaviorGraphAut$ whose set of states is given by the first two levels from \Cref{fig:example:behavior_graph}. 

The first column of $\obsTable$ contains the elements of $\representatives \cup \representatives \Sigma$ such that the upper part is constituted by $\representatives = \{\emptyword,a,ab,aba,aa\}$ and the lower part by $\representatives\Sigma \setminus \representatives$. The first row contains the elements of $\separatorsMapping$ such that the left part is constituted by $\separatorsCounters = \{\emptyword\}$ and the right part by $\separatorsMapping \setminus \separatorsCounters$. For each element $us \in (\representatives \cup \representatives\Sigma)\separatorsCounters$, we store the two values $\obsMapping(us)$ and $\obsCounters(us)$ in the intersection of row $u$ and column $s$. For instance, these values are equal to $0,\bot$ for $u = aa$ and $s = \emptyword$. For each element $us \in (\representatives \cup \representatives\Sigma)(\separatorsMapping \setminus \separatorsCounters)$, we have only one value $\obsMapping(us)$ to store.

Notice that \(\prefTable\) is a proper subset of \(\prefixes{L_{\ell}}\).
For instance, \(aababaa \not\in \prefTable\) and \(aababaa \in \prefixes{L_{\ell}}\).
\qed\end{example}

To compute the values of the table $\obsTable$, the learner proceeds by asking the following queries to the teacher. To compute $\obsMapping(us)$ for all $u \in \representatives \cup \representatives \Sigma$ and $s \in \separatorsMapping$, he needs to check whether $us \in L_\ell$. He first asks a membership query for $us$. If the answer is yes (that is, $us \in L$), he additionally asks a counter value query for all prefixes $x$ of $us$ (to check whether $us \in L_{\ell}$). Notice that such a query is allowed since $x \in \prefixes{L}$. The learner thus knows what is the set $\prefTable$. In order to compute $\counterAut{us}$, for all $u \in \representatives \cup \representatives \Sigma$ and $s \in \separatorsCounters$, he asks a counter value query for $us$ whenever $us \in \prefTable$. These queries are also allowed because $\prefTable \subseteq \prefixes{L_\ell}$. In summary, computing the value $\obsMapping(us)$ requires one membership query and at most $|us|$ counter value queries, and computing the value $\counterAut{us}$ requires at most one counter value query. As $\representatives \cup \representatives\Sigma$ is prefix-closed and $\separatorsMapping$ is suffix-closed, we get the next lemma.

\begin{lemma}\label{lem:nbrQueries}
    Computing $\obsMapping(us)$ and $\obsCounters(us)$ requires a polynomial number of membership and counter value queries in the sizes of $\representatives \cup \representatives\Sigma$ and $\separatorsMapping$.
\qed\end{lemma}

In the next example, we explain why it is necessary to use the additional set $\separatorsMapping$ in \Cref{def:observation_table}.

\begin{example}\label{example:table:infinite_loop}
    \begin{figure}
        \begin{subfigure}[t]{.3\textwidth}
            \centering
            \begin{tabular}{l|c} 
                    & \(\emptyword\)\\
    \hline 
    \(\emptyword\)  & \(0, 0\)\\
    \(a\)           & \(0, 1\)\\
    \(ab\)          & \(0, 1\)\\
    \(aba\)         & \(1, 0\)\\
    \hline 
    \(b\)           & \(1, 0\)\\
    \(aa\)          & \(0, \bot\)\\
    \(abb\)         & \(0, \bot\)\\
    \(abaa\)        & \(1, 0\)\\
    \(abab\)        & \(1, 0\)\\
\end{tabular}
        \end{subfigure}
        ~ 
        \begin{subfigure}[t]{.3\textwidth}
            \centering
            \begin{tabular}{l|c} 
                    & \(\emptyword\)\\
    \hline 
    \(\emptyword\)  & \(0, 0\)\\
    \(a\)           & \(0, 1\)\\
    \(ab\)          & \(0, 1\)\\
    \(aba\)         & \(1, 0\)\\
    \(abb\)         & \(0, 1\)\\
    \hline 
    \(b\)           & \(1, 0\)\\
    \(aa\)          & \(0, \bot\)\\
    \(abaa\)        & \(1, 0\)\\
    \(abab\)        & \(1, 0\)\\
    \(abba\)        & \(1, 0\)\\
    \(abbb\)        & \(0, \bot\)
\end{tabular}
        \end{subfigure}
        ~ 
        \begin{subfigure}[t]{.3\textwidth}
            \centering
            \begin{tabular}{l|c} 
                    & \(\emptyword\)\\
    \hline 
    \(\emptyword\)  & \(0, 0\)\\
    \(a\)           & \(0, 1\)\\
    \(ab\)          & \(0, 1\)\\
    \(aba\)         & \(1, 0\)\\
    \(abb\)         & \(0, 1\)\\
    \(abbb\)        & \(0, 1\)\\
    \hline 
    \(b\)           & \(1, 0\)\\
    \(aa\)          & \(0, \bot\)\\
    \(abaa\)        & \(1, 0\)\\
    \(abab\)        & \(1, 0\)\\
    \(abba\)        & \(1, 0\)\\
    \(abbba\)       & \(1, 0\)\\
    \(abbbb\)       & \(0, \bot\)\\
\end{tabular}
        \end{subfigure}
        \caption{Observation tables exposing an infinite loop when using the \LStar algorithm.}%
        \label{fig:example:table:infinite_loop}
    \end{figure}

    Assume that we only use the set $\separatorsCounters$ in \Cref{def:observation_table} and that the current observation table is the leftmost table $\obsTable$, with $\ell = 1$, given in \Cref{fig:example:table:infinite_loop} for the ROCA from \Cref{fig:example:roca}. 
    
    On top of that, assume we are using the classical $\LStar$ algorithm (see \Cref{subsec:Lstar}). As we can see, the table is not closed since the stored information for \(abb\), that is, $0,\bot$, does not appear in the upper part of the table for any $u \in \representatives$. So, we add \(abb\) in this upper part and \(abba\) and \(abbb\) in the lower part, to obtain the second table of \Cref{fig:example:table:infinite_loop}. Notice that this shift of $abb$ has changed its stored information, which is now equal to $0,1$. Indeed the set $\prefixes{\observationTable}$ now contains $abb$ as a prefix of $abba \in L_{\ell}$. 
    Again, the new table is not closed because of \(abbb\). After shifting
    this word in the upper part of the table, we obtain the third table
    of \Cref{fig:example:table:infinite_loop}. It is still not closed due to \(abbbb\), etc.
    This process will continue ad infinitum.
\qed\end{example}

To avoid an infinite loop when making the table closed, as described in the
previous example, we modify both the concept of table and how to derive an
equivalence relation from that table. Our solution is to introduce the set
$\separatorsMapping$, as already explained, but also the concept of
approximation set to approximate $\equivBGROCA$.

\begin{definition}\label{def:approx}
    Let $\obsTable  = (\representatives, \separatorsCounters, \separatorsMapping, \obsMapping, \obsCounters)$ be an observation table up to $\ell$. Let \(u, v \in \representatives \cup \representatives \Sigma\).
    Then, \(u \in \Approx(v)\) if and only if:
    \begin{itemize}
        \item $\forall s \in \separatorsCounters$, $\obsMapping(us) = \obsMapping(vs)$ and
        \item $\forall s \in \separatorsCounters$, $\obsCounters(us) \neq \bot \text{ and } \obsCounters(vs) \neq \bot \implies \obsCounters(us) = \obsCounters(vs)$.
    \end{itemize}
    The set $\Approx(v)$ is called an \emph{\approxSet}.
\end{definition}

Notice that in this definition, we consider \(\bot\) values as wildcards and
that we focus on words with suffixes from $\separatorsCounters$ only (and not from
$\separatorsMapping \setminus \separatorsCounters$). Interestingly, such wildcard entries in observation tables also feature in the work of Leucker and Neider on learning from an ``inexperienced'' teacher~\cite{DBLP:conf/isola/LeuckerN12}. Just like in that work, a crucial part of our learning algorithm concerns how to obtain an equivalence relation from an observation table with wildcards. Indeed, note that
\(\Approx\) does not define an equivalence relation as it is not transitive,
i.e., it is not true in general that if \(u \in \Approx(v)\) and \(v \in
\Approx(w)\), then \(u \in \Approx(w)\) due to the \(\bot\) values.  However,
\(\Approx\) is reflexive (i.e., \(u \in \Approx(u)\)) and symmetric (i.e., \(u
\in \Approx(v) \iff v \in \Approx(u)\)).

\begin{example}
    Let \(\observationTable\) be the table from \Cref{fig:example:observation_table} and let us compute \(\Approx(\emptyword)\).
    Notice that \(\separatorsCounters = \{\emptyword\}\) and we do not take into account \(\separatorsMapping = \{a, ba\}\) to compute the approximation sets.
    We can see that \(aba \notin \Approx(\emptyword)\) as \(\obsMapping(aba) = 1\) and \(\obsMapping(\emptyword) = 0\).
    Moreover, \(a \notin \Approx(\emptyword)\) since \(\obsCounters(a) \neq \bot, \obsCounters(\emptyword) \neq \bot\), and \(\obsCounters(a) \neq \obsCounters(\emptyword)\)
    With the same arguments, we also discard \(ab, b, abb, abaa, abab\).
    Thus, \(\Approx(\emptyword) = \{\emptyword, aa, aaa, aab\}\).
    On the other hand, \(\Approx(aa) = \{\emptyword, a, ab, aa, abb, aaa, aab\}\) due to the fact that \(\obsCounters(aa) = \bot\).
\qed\end{example}

The following notation will be convenient later. Let \(\observationTable\) be an observation table up to \(\ell \in \N\) and \(u \in \representatives \cup \representatives \Sigma\). If \(\obsCounters(u) = \bot\) (which means that $u \not\in \prefTable$), then we say that \(u\) is a \emph{\binWord}. Let us denote by \(\withoutBinWords{\representatives}\) the set \(\representatives\) from which the \binWords have been removed. We define \(\withoutBinWords{\representatives \cup \representatives \Sigma}\) in a similar way.

We can now formalize the relation between $\equivBGROCA$ and $\Approx$. 

\begin{proposition}\label{long:prop:observation_table:u_equiv_v_implies_u_in_approx_v}
    Let \(\observationTable \) be an observation table up to \(\ell \in \N\). 
    \begin{itemize}
        \item Let \(u, v \in \withoutBinWords{\representatives \cup \representatives \Sigma}\). Then, \(u \equivBGROCA v \implies u \in \Approx(v)\).
        \item All the words \(u \in \representatives \cup \representatives \Sigma \setminus \withoutBinWords{\representatives \cup \representatives \Sigma}\) are in the same \approxSet.
    \end{itemize}
    
\end{proposition}
\begin{proof}
    Let \(u, v \in \withoutBinWords{\representatives \cup \representatives\Sigma}\) be such that \(u \equivBGROCA v\).
    Our goal is to prove that \(u \in \Approx(v)\).
    That is,
    \begin{enumerate}
        \item \(\forall s \in \separatorsCounters, \obsMapping(us) = \obsMapping(vs)\),
        \item \(\forall s \in \separatorsCounters, \obsCounters(us) \neq \bot \land \obsCounters(vs) \neq \bot \implies \obsCounters(us) = \obsCounters(vs)\).
    \end{enumerate}
    
    Let us prove Claim 1. Let \(u, v \in \withoutBinWords{\representatives
    \cup \representatives\Sigma}\), and \(s \in \separatorsCounters\) be such
    that \(us \in L_{\ell}\).   Let us show that \(vs \in L_{\ell}\). Since
    \(u \equivBGROCA v\) and $us \in L_{\ell} \subseteq L$, it follows that
    $vs \in L$. In order to conclude \(vs \in L_{\ell}\), we need to prove
    that \(\forall x \in \prefixes{vs}\), we have \(\counterAut{x} \leq
    \ell\). Let \(x \in \prefixes{vs}\).
    \begin{itemize}
        \item Assume \(x \in \prefixes{v}\). By hypothesis, \(v \in \withoutBinWords{\representatives \cup \representatives\Sigma}\). Thus, \(v \in \prefTable \subseteq \prefixes{L_{\ell}}\). Therefore, \(\counterAut{x} \leq \ell\).
        \item Assume \(x = vy\) with \(y \in \prefixes{s}\). The fact that \(u \equivBGROCA v\) and \(uy, vy \in \prefixes{L}\) (since \(us, vs \in L\)) implies that \(\counterAut{uy} = \counterAut{vy}\). As $us \in L_{\ell}$ by hypothesis, we know that \(\counterAut{uy} \leq \ell\). Therefore, \(\counterAut{x} = \counterAut{vy} \leq \ell\).
    \end{itemize}
    We have proved \(us \in L_{\ell} \implies vs \in L_{\ell}\). The other implication is proved similarly.
    
    Let us prove Claim 2. Let \(s \in \separatorsCounters\) be such that \(\obsCounters(us) \neq \bot\) and \(\obsCounters(vs) \neq \bot\). By hypothesis, we have \(us, vs \in \prefTable \subseteq \prefixes{L_\ell}\) and, thus, \(us, vs \in \prefixes{L}\).
    Since \(u \equivBGROCA v\), it holds \(\counterAut{us} = \counterAut{vs} \leq \ell\).
    Moreover, since \(us, vs \in \prefTable\), we have \(\obsCounters(us) = \counterAut{us}\) and \(\obsCounters(vs) = \counterAut{vs}\). Therefore, we conclude that \(\obsCounters(us) = \obsCounters(vs)\).
    
    We now proceed to the proof of the second part of the proposition. Any word $u \in \representatives \cup \representatives \Sigma \setminus \withoutBinWords{\representatives \cup \representatives \Sigma}$ is a \binWord. It follows that for all $s \in \separatorsCounters$, $\obsMapping(us) = 0$ and $\obsCounters(us) = \bot$. Therefore all the \binWords are in the same \approxSet.
\qed\end{proof}

\subsection{Closed and consistent observation table}\label{sec:closedConsistant}

As for the $\LStar$ algorithm~\cite{DBLP:journals/iandc/Angluin87}, we need to define constraints that the table must respect in order to obtain a congruence relation from \(\Approx\). This is more complex than for $\LStar$, namely, the table must be closed, \consistent, and \prefixed. The first two constraints are close to the ones already imposed by $\LStar$. The last one is new. Crucially, it implies that $\Approx$ is transitive.

\begin{definition}\label{def:closedConsistent}
    Let \(\observationTable\) be an observation table up to \(\ell \in \N\).
    We say the table is:
    \begin{itemize}
        \item \emph{closed} if \(\forall u \in \representatives \Sigma, \Approx(u) \cap \representatives \neq \emptyset\),
        \item \emph{\consistent} if \(\forall u \in \representatives\), $\forall a \in \Sigma$,
            \[
                ua \in \bigcap_{v \in \Approx(u) \cap \representatives} \Approx(va),
            \]
        \item \emph{\prefixed} if \(\forall u, v \in \representatives \cup \representatives \Sigma\) such that \(u \in \Approx(v)\),
            \[
                \forall s \in \separatorsCounters,~ \obsCounters(us) \neq \bot \iff \obsCounters(vs) \neq \bot.
            \]
    \end{itemize}
\end{definition}

\begin{example} 
    Let \(\observationTable\) be the table from \Cref{fig:example:observation_table}. We have $\Approx(b) \cap \representatives \neq \emptyset$ because $aba \in \Approx(b)$. More generally one can check that \(\observationTable\) is closed. However, \(\observationTable\) is not \consistent.  Indeed, \(\emptyword b \notin \bigcap_{v \in \Approx(\emptyword) \cap \representatives} \Approx(v b)\) since $\Approx(\emptyword) \cap \representatives = \{\emptyword, aa\}$ and $\emptyword b \notin \Approx(aa b)$. Finally, \(\observationTable\) is also not \prefixed since \(aa \in \Approx(\emptyword)\) but \(\obsCounters(aa) = \bot\) and \(\obsCounters(\emptyword) = 0\).
\qed\end{example}

When $\observationTable$ is closed, $\Sigma$- and \prefixed, we define the following relation over $\representatives \cup \representatives \Sigma$: two words \(u, v \in \representatives \cup \representatives \Sigma\) are \(\equivTable\)-equivalent,
\[
u \equivTable v \iff u \in \Approx(v).
\]
In this case, $\equivTable$ is a congruence over $\representatives$ as stated in the next proposition.

\begin{proposition}\label{lem:rightcongruence}
    Let \(\observationTable\) be a closed, $\Sigma$- and \prefixed observation
    table up to \(\ell \in \N\).  Then, \(\equivTable\) is an equivalence
    relation over $\representatives \cup \representatives \Sigma$ that is a
    congruence over \(\representatives\).
\end{proposition}
\begin{proof}
    First, we have to prove that \(\equivTable\) is an equivalence relation over $\representatives \cup \representatives \Sigma$.
    It is easy to see that \(\equivTable\) is reflexive and symmetric. Let us show that it is transitive. Let \(u, v, w \in \representatives \cup \representatives\Sigma\) be such that \(u \in \Approx(v)\) and \(v \in \Approx(w)\).
    We want to show that \(u \in \Approx(w)\), i.e., \(\forall s \in \separatorsCounters\), \(\obsMapping(us) = \obsMapping(ws)\) and \( \obsCounters(us) \neq \bot \land \obsCounters(ws) \neq \bot \implies \obsCounters(us) = \obsCounters(ws)\).
    Let \(s \in \separatorsCounters\). By hypothesis, we have \(\obsMapping(us) = \obsMapping(vs)\) and \(\obsMapping(vs) = \obsMapping(ws)\).
    So, \(\obsMapping(us) = \obsMapping(ws)\). Now, assume \(\obsCounters(us) \neq \bot \land \obsCounters(ws) \neq \bot\).
    Since the table is \prefixed and \(u \in \Approx(v)\), it must hold that \(\obsCounters(vs) \neq \bot\).
    Therefore, \(\obsCounters(us) = \obsCounters(vs)\).
    Likewise, we deduce that \(\obsCounters(vs) = \obsCounters(ws)\) since \(v \in \Approx(w)\).
    We conclude that \(\obsCounters(us) = \obsCounters(ws)\).

    Second, we prove that \(\equivTable\) is a congruence over \(\representatives\).
    Let \(u, v \in \representatives\) be such that \(u \equivTable v\) and let \(a \in \Sigma\).
    Since the table is closed, there exist \(u', v' \in \representatives\) such that \(u' \equivTable ua\) and \(v' \equivTable va\).
    Since the table is \consistent, we have $ua \in \bigcap_{w \in \Approx(u) \cap \representatives} \Approx(wa)$.
    Thus as  \(v \in \Approx(u) \cap \representatives\), we have \(ua \in \Approx(va)\), i.e., \(ua \equivTable va\).
    It follows that \(u' \equivTable v'\).
\qed\end{proof}

We will explain in \Cref{sec:correctness} below how to proceed to make
a table closed, $\Sigma$- and \prefixed. This will require to \emph{extend}
the table with some new words added to $\representatives \cup \representatives\Sigma$, $\separatorsCounters$, or $\separatorsMapping$ and to \emph{update} already known values of the table (since $\prefTable$  may have changed). This will be feasible in a finite amount of time. 

\begin{proposition}\label{prop:finiteProcess}
    Given an observation table \(\observationTable\) up to \(\ell \in \N\), there exists an algorithm that makes it closed, $\Sigma$- and \prefixed in a finite amount of time.
\end{proposition}

The proof of this proposition is deferred in \Cref{subsec:growth}.
Presently, we argue that
from a closed, $\Sigma$- and \prefixed observation table and its equivalence
relation $\equivTable$, it is possible to construct a DFA\@.
\begin{definition}\label{def:observation_table:automaton}
    Let \(\observationTable\) be a closed, $\Sigma$- and \prefixed observation table up to \(\ell \in \N\).
    From \(\equivTable\), we define the DFA \(\autTable =
    (Q_{\observationTable}, \Sigma, \delta_{\observationTable},
    q^0_{\observationTable}, F_{\observationTable})\) with:
    \begin{itemize}
        \item \(Q_{\observationTable} = \{\equivalenceClass{u}_{\equivTable} \mid u \in \representatives\}\), 
        \item \(q^{0}_{\observationTable} = \equivalenceClass{\emptyword}_{\equivTable}\),
        \item \(F_{\observationTable} = \{\equivalenceClass{u}_{\equivTable} \mid \obsMapping(u) = 1\}\), and
        \item the (total) transition function $\delta_{\observationTable}$ is defined by $\delta_{\observationTable}(\equivalenceClass{u}_{\equivTable}, a) = \equivalenceClass{ua}_{\equivTable}$, for all $\equivalenceClass{u}_{\equivTable} \in Q_{\observationTable}$ and $a \in \Sigma$.
    \end{itemize}
\end{definition}

The next lemma states that $\autTable$ is well-defined and its definition makes sense with regards to the information stored in the table $\observationTable$.

\begin{lemma}\label{lemma:observation_table:automaton:correct_representatives}
    Let \(\observationTable\) be a closed, $\Sigma$- and \prefixed observation
    table up to \(\ell \in \N\), and \(\autTable\) be the related DFA\@. Then,
    \begin{itemize}
        \item $\autTable$ is well-defined and
        \item for all \(u \in \representatives \cup \representatives\Sigma\),
          \(u \in \languageOf{\autTable} \iff u \in L_{\ell}\).
    \end{itemize}
\end{lemma}

\begin{proof}
Let us first explain why $\autTable$ is well-defined. The initial state is
well defined since $\emptyword \in R$. The set of final states is well defined
because $u \equivTable v \wedge \obsMapping(u) = 1 \implies \obsMapping(v) =
1$. The transition function is well defined because the table
$\observationTable$ is closed and $\equivTable$ is a congruence.

Let us prove the second statement of the lemma. An easy induction shows that \(\forall u \in \representatives \cup \representatives\Sigma, \delta_{\observationTable}(q_{\observationTable}^0, u) = \equivalenceClass{u}_{\equivTable}\). 
Therefore, for all such $u$, we get \(u \in \languageOf{\autTable} \iff u \in L_{\ell}\).
\qed\end{proof}

\subsection{Handling counterexamples to partial equivalence queries}\label{sec:cex}

Given an observation table \(\observationTable\) that is closed, $\Sigma$- and
$\bot$- consistent, let \(\autTable\) be the DFA constructed from
\(\observationTable\) like described
in \Cref{def:observation_table:automaton}.
If the teacher's answer to a partial equivalence query over \(\autTable\) is
positive, the learned DFA $\autTable$ exactly accepts \(L_{\ell}\).  Otherwise,
the teacher returns a counterexample, that is, a word \(w \in \Sigma^*\) such
that \(w \in L_{\ell} \iff w \notin \languageOf{\autTable}\).  We extend and
update the table $\observationTable$ to obtain a new observation table
\(\observationTable' = (\representatives', \separatorsCounters',
\separatorsMapping', \obsMapping', \obsCounters')\) such that:
\[
  \representatives' = \representatives \cup \prefixes{w}.  
\]
We then compute \(\obsMapping'\) and \(\obsCounters'\) using membership
and counter value queries. We finally make \(\observationTable'\) closed, $\Sigma$- and \prefixed.

We show that the new equivalence relation $\equivTable[\obsTable']$ is a refinement of $\equivTable$ with strictly more equivalence classes.

\begin{proposition}\label{prop:betterApproximation}
    Let \(\observationTable = (\representatives, \separatorsCounters, \separatorsMapping, \obsMapping, \obsCounters)\) be a closed, $\Sigma$- and \prefixed observation table up to \(\ell \in \N\), \(w \in \Sigma^*\) be the counterexample provided by the teacher to a partial equivalence query, and \(\observationTable' = (\representatives', \separatorsCounters', \separatorsMapping', \obsMapping', \obsCounters')\) be the closed, $\Sigma$- and \prefixed observation table obtained after processing \(w\).
    Then,
    \[
        \forall u, v \in \representatives \cup \representatives \Sigma, u \equivTable[\obsTable'] v \implies u \equivTable v.
    \]
    Furthermore, the index of \(\equivTable[\obsTable']\) is strictly greater than the index of \(\equivTable\).
\end{proposition}

The proof of this proposition uses the next lemma stating that the \approxSet s cannot increase when an observation table is extended with some new words. In it, we denote by $\Approx'$ the \approxSet s with respect to the table $\observationTable'$.  
\begin{lemma}\label{lem:observation_table:u_in_approx_prime_v_implies_u_in_approx_v}
    Let \(\observationTable = (\representatives, \separatorsCounters, \separatorsMapping, \obsMapping, \obsCounters)\) and \(\observationTable' = (\representatives', \separatorsCounters', \separatorsMapping', \obsMapping', \obsCounters')\) be two observation tables up to the same counter limit \(\ell \in \N\) such that \(\representatives \cup \representatives \Sigma \subseteq \representatives' \cup \representatives' \Sigma\), \(\separatorsCounters \subseteq \separatorsCounters'\), and \(\separatorsMapping \subseteq \separatorsMapping'\).
    Then, \(\forall u, v \in \representatives \cup \representatives \Sigma\),
    \[
        u \in \Approx'(v) \implies u \in \Approx(v).
    \]
\end{lemma}
\begin{proof}
    By hypothesis, since both tables use the same counter limit $\ell$, the same language \(L_{\ell}\) is considered.
    Thus it holds that \(\obsMapping(us) = \obsMapping'(us)\) for all $u \in \representatives \cup \representatives \Sigma$ and $s \in \separatorsCounters$.
    Moreover, we have $\prefTable \subseteq \prefTable[\obsTable']$.
    This means that we also have \(\forall us \in (\representatives \cup \representatives \Sigma) \separatorsCounters\), \(\obsCounters(us) \neq \bot \implies \obsCounters'(us) = \obsCounters(us)\), and some \(\obsCounters(us) = \bot\) can possibly be replaced by \(\obsCounters'(us) \in \{0, \dotsc, \ell\}\).
    
    Let \(u, v \in \representatives \cup \representatives\Sigma\) be such that \(u \in \Approx'(v)\), i.e., \(\forall s \in \separatorsCounters', \obsMapping'(us) = \obsMapping'(vs)\) and \(\obsCounters'(us) \neq \bot \land \obsCounters'(vs) \neq \bot \implies \obsCounters'(us) = \obsCounters'(vs)\).
    We want to show that \(u \in \Approx(v)\).

    Let \(s \in \separatorsCounters\).
    We start by proving that \(\obsMapping(us) = \obsMapping(vs)\).
    Since \(u \in \Approx'(v)\), we have \(\obsMapping'(us) = \obsMapping'(vs)\).
    Therefore \(\obsMapping(us) = \obsMapping'(us) = \obsMapping'(vs) = \obsMapping(vs)\). We now show that \(\obsCounters(us) \neq \bot \land \obsCounters(vs) \neq \bot \implies \obsCounters(us) = \obsCounters(vs)\).
    From \(\obsCounters(us) \neq \bot \land \obsCounters(vs) \neq \bot\), it follows that \(\obsCounters'(us) \neq \bot \land \obsCounters'(vs) \neq \bot\). Since \(u \in \Approx'(v)\), we have \(\obsCounters(us) = \obsCounters'(us) = \obsCounters'(vs) = \obsCounters(vs)\).
\qed\end{proof}

\begin{proof}[of \Cref{prop:betterApproximation}]
    The first part of the claim is an immediate consequence of \Cref{lem:observation_table:u_in_approx_prime_v_implies_u_in_approx_v}.
We focus on showing that the index of \(\equivTable[\obsTable']\) is strictly greater than the index of \(\equivTable\).
    By contradiction, let us assume this is false.
    Let \(\autTable\) (resp.\ \(\autTable[\obsTable']\)) be the DFA constructed from \(\equivTable\) (resp.\ \(\equivTable[\obsTable']\)). By definition of the automata and the first part of the proposition, it holds that \(\autTable\) and \(\autTable[\obsTable']\) are isomorphic.
    However, \(w\) is a counterexample for \(\autTable\), i.e., \(w \notin L_{\ell} \iff w \in \languageOf{\autTable}\). By \Cref{lemma:observation_table:automaton:correct_representatives} and since \(w \in \representatives'\), we have \(w \in L_{\ell} \iff w \in \languageOf{\autTable[\obsTable']} \).
    This is a contradiction with \(\autTable\) and \(\autTable[\obsTable']\) being isomorphic.
\qed\end{proof}

Using this proposition and \Cref{long:prop:observation_table:u_equiv_v_implies_u_in_approx_v}, we deduce that after a finite number of steps, we will obtain an observation table \(\observationTable\) and its corresponding DFA \(\autTable\) such that \(\languageOf{\autTable} = L_{\ell}\).

\begin{corollary}
Let $\ell \in \mathbb{N}$ be a counter limit. With the described learning process, a DFA accepting $L_\ell$ is eventually learned.    
\end{corollary}
\begin{proof}
On one hand, given an observation table $\observationTable$ up to $\ell$, by \Cref{long:prop:observation_table:u_equiv_v_implies_u_in_approx_v}, we have that for all \(u, v \in \withoutBinWords{\representatives \cup \representatives \Sigma}\), \(u \equivBGROCA v \implies u \equivTable v\), and all the words in $\representatives \cup \representatives \Sigma \setminus \withoutBinWords{\representatives \cup \representatives \Sigma}$ are in the same equivalence class of $\equivTable$. On the other hand, by \Cref{lem:sizeEquiv}, the number of equivalence classes of
$\equivBGROCA$ up to counter limit $\ell$ is bounded by $(\ell + 1)|Q| + 1$. With \Cref{prop:betterApproximation}, it follows that the index of $\equivTable$ eventually stabilizes. This means that the teacher stops giving any counterexample and the final DFA $\autTable$ accepts
$L_\ell$.
\qed\end{proof}

\subsection{Learning algorithm}\label{subsec:learning_algo}

We now have every piece needed to give the learning algorithm for ROCAs. It
has the same structure as the algorithm for VCAs (see \Cref{alg:lstar_vca})
that was detailed in \Cref{subsec:LearningVCAs}. Here, we only recall the main steps of this algorithm and what are the differences (without explicitly giving the algorithm).

We start by initializing the observation table \(\observationTable\) with \(\ell = 0, \representatives = \separatorsCounters = \separatorsMapping = \{\emptyword\}\).
Then, we make the table closed, \(\Sigma\)-, and \prefixed (this will be explained in \Cref{subsec:makingClosed} below), construct the DFA $\autTable$, and ask a partial equivalence query with $\autTable$.
If the teacher answers positively, we have learned a DFA accepting \(L_{\ell}\).
Otherwise, we use the provided counterexample to update the table without increasing \(\ell\) (as explained in \Cref{sec:cex}). 

Once the learned DFA $\autTable$ accepts the language $L_{\ell}$, the next proposition states that the initial fragments (up to a certain counter limit) of both $\autTable$ and $\behaviorGraphAut$ are isomorphic.
This means that, once we have learned a long enough initial fragment, we can extract a periodic description from \(\autTable\) that is valid for \(\behaviorGraphAut\).

\begin{proposition}\label{prop:isoPrefix}
    Let \(\behaviorGraphAut\) be the behavior graph of an ROCA $\automaton$, $K$ be its width, and $m, k$ be the offset and the period of a periodic description of \(\behaviorGraphAut\).
    Let \(s = m + {(K \cdot k)}^4\). Let \(\observationTable\) be a closed, $\Sigma$- and \prefixed observation table up to \(\ell > s\) such that \(\languageOf{\autTable} = L_{\ell}\).
    Then, the trim parts of the subautomata of \(\behaviorGraphAut\) and \(\autTable\) restricted to the levels in \(\{0, \dotsc, \ell - s\}\) are isomorphic. 
\end{proposition}

The proof of this proposition is given in \Cref{subsec:properties}.

Hence we extract all possible periodic descriptions \(\alpha\) from \(\autTable\), exactly as for VCAs. By \Cref{prop:DescriptionAut}, each of these descriptions \(\alpha\) yields an ROCA \(\automaton_{\alpha}\) on which we ask an equivalence query.
If the teacher answers positively, we have learned an ROCA accepting \(L\) and we are done.
Otherwise, the teacher returns a counterexample \(w\).

It may happen that \(w \notin \languageOf{\automaton_{\alpha}} \iff w \in \languageOf{\autTable}\) if the description \(\alpha\) only considers the first few levels of the DFA $\autTable$ and, by doing so, discards important knowledge appearing on the further levels.
Thus, such a counterexample must be discarded as $\autTable$ already correctly
classifies it.

If we do not have any ROCAs or every counterexample was discarded, we use
\(\autTable\) directly as an ROCA (this can be done simply by constructing an
ROCA that simulates \(\autTable\) without ever modifying the counter).
Since we know that \(\languageOf{\autTable} = L_{\ell}\), we are sure the returned counterexample is usable.

Finally, let \(w\) be an arbitrary counterexample among the ones that were not discarded.
The new counter limit is the height of \(w\), i.e., \(\ell = \heightAut{w}\).
We then extend and update the table by adding $\prefixes{w}$ to $\representatives$. Notice that this operation may change some values of \(\obsMapping\) since a word that was rejected may now be accepted due to the higher counter limit.

\section{Properties of the behavior graph of an ROCA}\label{sec:isomorphism}

In this section, we provide the proofs of \Cref{thm:behavior_graph:periodicity} and \Cref{prop:DescriptionAut} from \Cref{subsec:BG_ROCA} stating two important properties of the behavior graph of an ROCA\@. We also prove \Cref{prop:isoPrefix} from \Cref{subsec:learning_algo} useful in our learning algorithm. The latter proof requires a preliminary lemma.

\subsection{Relating the behavior graphs of VCAs and ROCAs}
To this end, from an ROCA \(\automaton\) accepting a language \(L\), we first explain how to construct a VCL \(\visibly{L}\) such that \(\visibly{L}\) encodes the counter operations defined in the transitions of \(\automaton\).
For instance, each time we see a transition increasing the counter value, the corresponding symbol in \(\visibly{L}\) will be a call. This will allow us to derive an isomorphism between the behavior graph of $\automaton$ as defined in \Cref{def:behavior_graph} and the behavior graph of \(\visibly{L}\) as defined in \Cref{def:visibly:behavior_graph}.

Let \(\automaton = (Q, \Sigma, \deltaZero, \deltaNotZero, q_0, F)\) be an ROCA and let \(L = \languageOf{\automaton}\). First we define a pushdown alphabet \(\pushdownAlphabet = \callAlphabet \cup \returnAlphabet \cup \internalAlphabet\) with \(\callAlphabet = \{a_c \mid a \in \Sigma\}\), \(\returnAlphabet = \{a_r \mid a \in \Sigma\}\), and \(\internalAlphabet = \{a_{int} \mid a \in \Sigma\}\). That is, for each symbol in \(\Sigma\), we create three new symbols: one call, one return, and one internal action. 

Then, we define a VCA $\visibly{\automaton}$ over $\pushdownAlphabet$ and a mapping $\lambda : \Sigma^* \to \pushdownAlphabet^*$ so that for all \(u \in \Sigma^*\) we have
\(\counterAut{u} = \countervalue{\toVisibly(u)}\). The automaton \(\visibly{\automaton} = (Q, \pushdownAlphabet, \visiblyDeltaZero, \visiblyDeltaNotZero, q_0, F)\) has the same set of states, initial state, and set of final states, but two new transition functions that are partial and mimic $\deltaZero$ and $\deltaNotZero$ as follows: 
	\begin{align*}
		\deltaZero(q,a) = (p,0) &\implies \visiblyDeltaZero(q,a_{int}) = p,\\
		\deltaZero(q,a) = (p,+1) &\implies \visiblyDeltaZero(q,a_c) = p,\\
		\deltaNotZero(q,a) = (p,0) &\implies \visiblyDeltaNotZero(q,a_{int}) = p,\\
		\deltaNotZero(q,a) = (p,+1) &\implies \visiblyDeltaNotZero(q,a_c) = p,\\
		\deltaNotZero(q,a) = (p,-1) &\implies \visiblyDeltaNotZero(q,a_r) = p.
	\end{align*}
The language accepted by $\visibly{\automaton}$ is denoted $\visibly{L}$.

Finally, given $u = a_0 \dots a_k \in \Sigma^*$, we define $\toVisibly(u) \in \visibly{\Sigma}^*$ as follows. 
The word $\toVisibly(u) = \visibly{a_0} \dots \visibly{a_k}$ is constructed using $\automaton$ by setting each $\visibly{a_i}$ to $b_c$ if $\counterAut{a_0 \dots a_i} - \counterAut{a_0 \dots a_{i-1}}$ is positive and $a_i$ is $b$; to $b_r$ if the difference is negative and $a_i$ is $b$; and to $b_{int}$ if the difference is zero and $a_i$ is $b$.
Note that given $\toVisibly(u) \in \visibly{\Sigma}^*$, it is easy to come back to $u$: simply discard the index in $\{c,r,int\}$ of each symbol of $\toVisibly(u)$. We denote by $\sigma$ this mapping from $\pushdownAlphabet$ to $\Sigma$, i.e. $\sigma(a_x) = a$ for $a \in \Sigma$ and $x \in \{c, r, int\}$.

\begin{example}
    We consider again the ROCA \(\automaton\) given in \Cref{fig:example:roca}.
    The VCA \(\visibly{\automaton}\) constructed from \(\automaton\) is given in \Cref{fig:example:vca}. Let \(u = aababaa\) and let us study \(\toVisibly(u)\).
    Recall from \Cref{example:roca} that we have the following run for \(u\) in \(\automaton\):
    \[
        (q_0, 0)
            \transition^a_{\automaton} (q_0, 1)
            \transition^a_{\automaton} (q_0, 2)
            \transition^b_{\automaton} (q_1, 2)
            \transition^a_{\automaton} (q_1, 1)
            \transition^b_{\automaton} (q_1, 1)
            \transition^a_{\automaton} (q_1, 0)
            \transition^a_{\automaton} (q_2, 0).
    \]
    Thus, by encoding the counter operations into the symbols, we obtain the word \(\toVisibly(u) = \visibly{u} = a_c a_c b_{int} a_r b_{int} a_r a_{int}\) and the corresponding run in \(\visibly{\automaton}\):
    \[
        (q_0, 0)
            \transition^{a_c}_{\visibly{\automaton}} (q_0, 1)
            \transition^{a_c}_{\visibly{\automaton}} (q_0, 2)
            \transition^{b_{int}}_{\visibly{\automaton}} (q_1, 2)
            \transition^{a_r}_{\visibly{\automaton}} (q_1, 1)
            \transition^{b_{int}}_{\visibly{\automaton}} (q_1, 1)
            \transition^{a_r}_{\visibly{\automaton}} (q_1, 0)
            \transition^{a_{int}}_{\visibly{\automaton}} (q_2, 0).
    \]
\qed\end{example}

The next lemma is easily proved.

\begin{lemma}\label{lem:lambda(u)}
    For all \(u \in \Sigma^*\), we have $\counterAut{u} = \countervalue{\toVisibly(u)}$. Moreover, $\visibly{L} = \{\toVisibly(u) \mid u \in L\}$ and $\prefixes{\visibly{L}} = \{\toVisibly(u) \mid u \in \prefixes{L}\}$.
\qed\end{lemma}

The announced isomorphism is stated in the following theorem. 

\begin{theorem}\label{thm:behavior_graph:isomorphism_to_visibly}
    Let \(\automaton = (Q, \Sigma, \deltaZero, \deltaNotZero, q_0, F)\) be an ROCA accepting $L = \languageOf{\automaton}$ and let $\visibly{L}$ be the corresponding VCL\@. 
    Let \(\behaviorGraphAut = (Q_{\equivBGROCA}, \Sigma, \delta_{\equivBGROCA}, q^0_{\equivBGROCA}, F_{\equivBGROCA})\) be the behavior graph of \(\automaton\) and \(\behaviorGraph[\visibly{L}] = (Q_{\nerodeCongruence}, \pushdownAlphabet, \delta_{\nerodeCongruence}, q^0_{\nerodeCongruence}, F_{\nerodeCongruence})\) be the behavior graph of \(\visibly{L}\).
    Then, \(\behaviorGraphAut\) and \(\behaviorGraph[\visibly{L}]\) are isomorphic up to \(\sigma\). Moreover this isomorphism respects the counter values (i.e., level membership) and both offset and period of periodic descriptions.
\end{theorem}

\begin{proof}
    We start by showing that \(\forall u, v \in \prefixes{L}\), we have 
    \[
        u \equivBGROCA v \iff \toVisibly(u) \nerodeCongruence \toVisibly(v).
    \] 
    As $\prefixes{\visibly{L}} = \{\toVisibly(u) \mid u \in \prefixes{L}\}$ by \Cref{lem:lambda(u)}, we will get the required one-to-one correspondence between the states $\equivalenceClass{u}_\equivBGROCA$ of $\behaviorGraphAut$ and the states $\equivalenceClass{\toVisibly(u)}_{\nerodeCongruence}$ of $\behaviorGraph[\visibly{L}]$. Notice that this correspondence respects the counter values of the states since all words in state $\equivalenceClass{u}_\equivBGROCA$ have the same counter value $\counterAut{u}$ which is equal (by \Cref{lem:lambda(u)}) to the counter value $\countervalue{\toVisibly(u)}$ of all words in state $\equivalenceClass{\toVisibly(u)}_{\nerodeCongruence}$. In other words, both $\equivalenceClass{\toVisibly(u)}_{\nerodeCongruence}$ and $\equivalenceClass{u}_\equivBGROCA$ are in level $\countervalue{\toVisibly(u)}$ of their corresponding behavior graphs.
    
    Let \(u, v \in \prefixes{L}\) be such that \(u \equivBGROCA v\), i.e., \(\forall w \in \Sigma^*, uw \in L \iff vw \in L\) and \(uw, vw \in \prefixes{L} \implies \counterAut{uw} = \counterAut{vw}\).
    We show that \(\forall \visibly{w} \in \pushdownAlphabet^*, \toVisibly(u)\visibly{w} \in \visibly{L} \iff \toVisibly(v)\visibly{w} \in \visibly{L}\).
    Let \(\visibly{w} \in \pushdownAlphabet^*\) and assume \(\toVisibly(u)\visibly{w} \in \visibly{L}\). By definition of $\visibly{\automaton}$, there exists $x \in \Sigma^*$ such that $ux \in L$ and $\toVisibly(ux) = \toVisibly(u)\visibly{w}$.  
    Since \(u \equivBGROCA v\), we get \(vx \in L\) and $\counterAut{uy} = \counterAut{vy}$ for all $y \in \prefixes{x}$.
    By definition of $\visibly{\automaton}$, it follows that \(\toVisibly(v) \visibly{w} \in \visibly{L}\), meaning that \(\toVisibly(u) \visibly{w} \in \visibly{L} \implies \toVisibly(v) \visibly{w} \in \visibly{L}\). Using a similar argument, we prove that \(\toVisibly(v) \visibly{w} \in \visibly{L} \implies \toVisibly(u) \visibly{w} \in \visibly{L}\), meaning that \(\toVisibly(u) \nerodeCongruence \toVisibly(v)\).

    We now show that \(\forall \toVisibly(u), \toVisibly(v) \in \prefixes{\visibly{L}}\) such that \(\toVisibly(u) \nerodeCongruence \toVisibly(v)\), we have $u \equivBGROCA v$. 
    First, let us show that \(\forall w \in \Sigma^*, uw \in L \iff vw \in L\). Assume that $uw \in L$. Then $\toVisibly(uw) = \toVisibly(u)\visibly{w} \in \visibly{L}$ with $\visibly{w} \in \visibly{\Sigma}^*$. As \(\toVisibly(u) \nerodeCongruence \toVisibly(v)\), we get $\toVisibly(v)\visibly{w} \in \visibly{L}$. Hence by discarding with $\sigma$ the index in $\{c,r,int\}$ of each symbol of $\toVisibly(v)\visibly{w}$, we obtain that $vw \in L$. The other implication is proved similarly and we have thus proved that \(uw \in L \iff vw \in L\). 
    Second, let us show that \(\forall w \in \Sigma^*, uw, vw \in \prefixes{L} \implies \counterAut{uw} = \counterAut{vw}\). Let \(w \in \Sigma^*\) and $\visibly{w} \in \visibly{\Sigma}^*$ be such that $\toVisibly(uw) = \toVisibly(u)\visibly{w} \in \prefixes{\visibly{L}}$.
    It follows that \(\toVisibly(u)\visibly{w} \nerodeCongruence \toVisibly(v)\visibly{w}\) and \(\toVisibly(v)\visibly{w} \in \prefixes{\visibly{L}}\) since $\toVisibly(u) \nerodeCongruence \toVisibly(v)$. From \(\toVisibly(u)\visibly{w} \nerodeCongruence \toVisibly(v)\visibly{w}\) and \(\toVisibly(v)\visibly{w} ,\toVisibly(v)\visibly{w} \in \prefixes{\visibly{L}}\), we deduce \(\countervalue{\toVisibly(u) \visibly{w}} = \countervalue{\toVisibly(v) \visibly{w}}\) by \Cref{lemma:visibly:same_cv}. We also have that $\toVisibly(v)\visibly{w} = \toVisibly(vw)$.
    By \Cref{lem:lambda(u)}, we get $\counterAut{uw} = \counterAut{vw}$.
    
    The one-to-one correspondence between the states of $\behaviorGraphAut$ and the states of $\behaviorGraph[\visibly{L}]$ is established. Notice that it puts in correspondence the initial (resp.\ final) states of both behavior graphs. We have also a correspondence with respect to the transitions (up to $\sigma$). Indeed, this follows from how both transition functions $\delta_{\equivBGROCA}$ and $\delta_{\nerodeCongruence}$ are defined.
    
    Finally, if $\alpha$ is a periodic description of $\behaviorGraphAut$ with offset $m$ and period $k$, then by $\lambda$ we get a periodic description of $\behaviorGraph[\visibly{L}]$ with the same offset and period. The converse is also true by using $\sigma$.
\qed\end{proof}

\subsection{Properties}\label{subsec:properties}

The isomorphism stated in \Cref{thm:behavior_graph:isomorphism_to_visibly} allows us to re-use some of the properties described in \Cref{sec:visibly} for VCAs and to transfer them to ROCAs. In particular, we have that the behavior graph $\behaviorGraphAut$ of an ROCA \(\automaton\) has a periodic description (\Cref{thm:behavior_graph:periodicity}). The proof of this result directly follows from \Cref{thm:visibly:behavior_graph:periodicity,thm:behavior_graph:isomorphism_to_visibly}.

Let us now prove our \Cref{prop:DescriptionAut} regarding the construction of an ROCA from the periodic description of its behavior graph. An illustrating example is given after the proof.

\begin{proof}[of \Cref{prop:DescriptionAut}]
    Let $\behaviorGraphAut$ be the behavior graph of an ROCA $\automaton$ accepting $L$. Let \(\alpha = \tau_0 \dotso \tau_{m-1} {(\tau_m \dotso \tau_{m+k-1})}^\omega\) be a periodic description of \(\behaviorGraphAut\) with offset \(m\) and period \(k\). Recall that the mappings $\tau_{\ell}$ used in this description are defined as \(\tau_{\ell} : \{1, \dotsc, K\} \times \Sigma \to \{1, \dotsc, K\} \times \{-1, 0, +1\}\) with \(K\) the width of the behavior graph.
    
    Let us explain how to construct an ROCA \(\automaton_{\alpha}\) accepting $L$.
    The proof is split into two cases, depending on the period $k$. It uses the enumerations $\nu_{\ell}$ and mappings \(\tau_{\ell}\) defined for the periodic description $\alpha$. 
    In the following constructions, it is supposed that all transitions which are not explicitly defined go to a bin state noted \(\bot\). That is, the transition functions are always total.

    First, assume \(k = 0\), i.e., $\alpha = \tau_0 \dotso \tau_{m-1}$. This means that \(\behaviorGraphAut\) is finite and \(L\) is therefore a regular language.
    The idea is simply to construct an ROCA that does not modify its counter value.
    That is, only \(\deltaZero\) is actually used. The states of this ROCA consist of two components. The first component stores the number of the current equivalence class (given by the enumeration $\nu_{\ell}$). The second component keeps track of the current level $\ell$ (given by the index $\ell$ of $\tau_{\ell}$). The transitions are defined according to the description $\alpha$.
    Let \(\automaton_{\alpha} = (Q, \Sigma, \deltaZero, \deltaNotZero, q_0, F)\) with:
    \begin{itemize}
        \item \(Q = \{1, \dotsc, K\} \times \{0, \dotsc, m - 1\} \uplus \{\bot\}\),
        \item \(q_0 = (\nu_0(\equivalenceClass{\emptyword}_{\equivBGROCA}), 0)\), 
        \item \(F = \{(\nu_0(\equivalenceClass{u}_{\equivBGROCA}), 0) \mid u \in L\}\),
        \item \(\forall q \in \{1, \dotsc, K\}, \forall \ell \in \{0, \dotsc, m - 1\}, \forall a \in \Sigma\), let \((p, c) = \tau_{\ell}(q, a)\).
            We define \(\deltaZero((q, \ell), a) = ((p, \ell + c), 0)\).
    \end{itemize}
    It is easy to see that \(\languageOf{\automaton_{\alpha}} = L\).

    Now assume \(k \neq 0\).
    The idea is the following: the counter value is always equal to zero as
    long as we remain in the initial part of the description $\alpha$, we
    modify the counter value when we are in the periodic part of the
    description.  For clarity, we encode the current level (from $0$ to
    $m+k-1$) into the states of the ROCA, in a way similar to the previous
    case, with the use of a modulo-\(k\) (\(+ m\)) counter in the repeating part.
    Let \(\automaton_{\alpha} = (Q, \Sigma, \deltaZero, \deltaNotZero, q_0, F)\) with: 
    \begin{itemize}
        \item \(Q = \{1, \dotsc, K\} \times \{0, \dotsc, m + k - 1\} \uplus \{\bot\}\), 
        \item \(q_0 = (\nu_0(\equivalenceClass{\emptyword}_{\equivBGROCA}), 0)\),
        \item \(F = \{(\nu_0(\equivalenceClass{u}_{\equivBGROCA}), 0) \mid u \in L\}\),
        \item The transition functions are defined as follows, \(\forall q \in \{1, \dotsc, K\}\), \(\forall a \in \Sigma\):
            \begin{itemize}
                \item \(\forall \ell \in \{0, \dotsc, m - 1\}\), let \((p, c) = \tau_{\ell}(q, a)\) and \(\deltaZero((q, \ell), a) = ((p, \ell + c), 0)\).
                \item For $\ell = m$, let \((p, c) = \tau_{\ell}(q, a)\). We need two functions $\deltaZero$ and $\deltaNotZero$ depending on whether we come from the initial part with counter $0$ or we are in the repeating part with counter $>0$.
                \begin{itemize}
                    \item If \(c = 0\), let \(\deltaZero((q, \ell), a) = \deltaNotZero((q, \ell), a) = ((p, \ell), 0)\).
                    \item If \(c = +1\), let \(\deltaZero((q, \ell), a) = \deltaNotZero((q, \ell), a) = ((p, \ell + (1 \bmod k)), +1)\).
                    (The \(1 \bmod k\) part is only useful when \(k = 1\) to guarantee we remain between \(0\) and \(m + k - 1 = m\).)
                    \item  If \(c = -1\), let \(\deltaZero((q, \ell), a) = ((p, \ell - 1), 0)\) and \(\deltaNotZero((q, \ell), a) = ((p, m + k - 1), -1)\) (We need to use a modulo-$k$ for the second component of the state in the periodic part). 
                \end{itemize}
                
                \item \(\forall \ell \in \{m + 1, \dotsc, m + k - 2\}\), let \((p, c) = \tau_{\ell}(q, a)\).
                    Let \(\deltaNotZero((q, \ell), a) = ((p, \ell + c), c)\).
                \item For $\ell = m+k-1$ and $k > 1$, let \((p, c) = \tau_{\ell}(q, a)\).
                \begin{itemize}
                    \item If \(c = 0\), let \(\deltaNotZero((q, \ell), a) = ((p, \ell), 0)\).
                    \item If \(c = +1\), let \(\deltaNotZero((q, \ell), a) = ((p, m), +1)\) (We need again to use a modulo-$k$).
                    \item If \(c = -1\), let \(\deltaNotZero((q, \ell), a) = ((p, \ell - 1), -1)\).
                \end{itemize}
            \end{itemize}
    \end{itemize}
This concludes our description of the algorithm to construct the ROCA\@. Clearly, the size of \(\automaton_{\alpha}\) is polynomial in \(m, k\) and \(K\). For completeness, its correctness is proved in \Cref{sec:correct-period2roca}.
\qed\end{proof}

\begin{example}\label{example:construction_roca_from_description}
    \begin{figure}
        \centering
        \begin{tikzpicture}[
    automaton,
    node distance = 70pt and 120pt,
]
    \node [state, initial]                  (1 0)   {\((1, 0)\)};
    \node [state, right=of 1 0]             (1 1)   {\((1, 1)\)};
    \node [state, right=of 1 1]             (1 2)   {\((1, 2)\)};
    \node [state, accepting, below=of 1 0]  (2 0)   {\((2, 0)\)};
    \node [state, right=of 2 0]             (2 1)   {\((2, 1)\)};
    \node [state, right=of 2 1]             (2 2)   {\((2, 2)\)};

    \path[deltaZero]
        foreach \s/\l/\t in {1 0/a/1 1, 1 0/b/2 0, 2 1/a/2 0} {
            (\s) edge node {\(\l, =0, 0\)} (\t)
        }
        (1 1)   edge [bend right]   node [']            {\(b, =0, 0\)}                  (2 1)
                edge                node                {\(a, =0, +1\)}                 (1 2)
        (2 0)   edge [loop left]    node[align=center]  {\(a, =0, 0\)\\\(b, =0, 0\)}    ()
        (2 1)   edge [in=-140, out=-110, loop]   node                {\(b, =0, 0\)}                  ()
    ;
    
    \path[deltaNotZero]
        (1 1)   edge [bend left]    node                {\(b, \neq0, 0\)}               (2 1)
                edge [bend left]   node                {\(a, \neq0, +1\)}              (1 2)
        (1 2)   edge [bend left=15]   node             {\(a, \neq0, +1\)}              (1 1)
                edge                node                {\(b, \neq0, 0\)}               (2 2)
        (2 1)   edge [in=-70, out=-40, loop]   node   {\(b, \neq0, 0\)}               ()
                edge [bend left=10]    node                {\(a, \neq0, -1\)}              (2 2)
        (2 2)   edge [bend left=10]    node                {\(a, \neq0, -1\)}              (2 1)
    ;

\end{tikzpicture}
        \caption{The ROCA \(\automaton_{\alpha}\) constructed from the periodic description \(\alpha = \tau_0 {(\tau_1 \tau_2)}^\omega\) of \(\behaviorGraphAut\) from \Cref{fig:example:behavior_graph}.}%
        \label{fig:example:construction_roca_from_description}
    \end{figure}

    Let \(\automaton\) be the ROCA from \Cref{fig:example:roca} and \(\behaviorGraphAut\) be its behavior graph from \Cref{fig:example:behavior_graph}.
    Let us assume that \(\nu_0(\equivalenceClass{\emptyword}_{\equivBGROCA}) = \nu_1(\equivalenceClass{a}_{\equivBGROCA}) = \dotso = 1\) and \(\nu_0(\equivalenceClass{b}_{\equivBGROCA}) = \nu_1(\equivalenceClass{ab}_{\equivBGROCA}) = \dotso = 2\).
    We then have the following \(\tau_i\) mappings:
    \begin{align*}
        \tau_0(1, a) &= (1, +1) & \tau_0(1, b) &= (2, 0)    & \tau_0(2, a) &= (2, 0)    & \tau_0(2, b) &= (2, 0)\\
        \tau_1(1, a) &= (1, +1) & \tau_1(1, b) &= (2, 0)    & \tau_1(2, a) &= (2, -1)   & \tau_1(2, b) &= (2, 0)\\
        \tau_2(1, a) &= (1, +1) & \tau_2(1, b) &= (2, 0)    & \tau_2(2, a) &= (2, -1)   & \tau_2(2, b) &= (2, 0)\\
        &\dotso                 & &\dotso                   & &\dotso                   & &\dotso
    \end{align*}
    Let \(\alpha = \tau_0 {(\tau_1 \tau_2)}^\omega\) be a periodic description of \(\behaviorGraph\).
    The resulting ROCA \(\automaton_{\alpha}\) is given in \Cref{fig:example:construction_roca_from_description}.
    We can clearly see that \(\automaton_{\alpha}\) mimics the periodic description as the states \((1, 0)\) and \((2, 0)\) form the initial part that is followed by a repeating part formed by the states \((1, 1), (2, 1), (1, 2)\) and \((2, 2)\).
\qed\end{example}

We conclude this section with the proof of \Cref{prop:isoPrefix}. This requires a preliminary lemma stating that one can bound the height $\heightAut{w}$ of a witness $w$ for non-equivalence with respect to $\equivBGROCA$. We recall that $\heightAut{w} = \max_{x \in \prefixes{w}} \counterAut{x}$. Such a property is proved in~\cite[Lemma 3]{neider2010learning} for VCAs and immediately transfers to ROCAs by \Cref{thm:behavior_graph:isomorphism_to_visibly}.

\begin{lemma}[{\cite[Lemma 3]{neider2010learning}}]\label{prop:behavior_graph:bounded_height}
    Let \(\automaton\) be an ROCA accepting a language $L \subseteq \Sigma^*$, $\behaviorGraphAut$ be its behavior graph, \(\alpha\) be a periodic description of $\behaviorGraphAut$ with offset \(m\) and period \(k\), and \(s = m + {(K \cdot k)}^4\).
    Let $\equivalenceClass{u}_\equivBGROCA$ and $\equivalenceClass{v}_\equivBGROCA$ be two distinct states of $\behaviorGraphAut$ such that $\counterAut{u} = \counterAut{v}$. Then there exists a word $w \in \Sigma^*$ such that $uw \in L \Leftrightarrow vw \not\in L$ and \(\heightAut{uw}, \heightAut{vw} \leq s + \counterAut{u}\).
\qed\end{lemma}

\begin{proof}[of \Cref{prop:isoPrefix}]
    Let \(\behaviorGraphAut = (Q_{\equivBGROCA}, \Sigma, \delta_{\equivBGROCA}, q_{\equivBGROCA}^0, F_{\equivBGROCA})\) be the behavior graph of \(\automaton\) and let \(\autTable = (Q_{\observationTable}, \Sigma, \delta_{\observationTable}, q_{\observationTable}^0, F_{\observationTable})\) be the DFA constructed from the table \(\observationTable  = (\representatives, \separatorsCounters, \separatorsMapping, \obsMapping, \obsCounters)\).
    Recall that by definition of \(\withoutBinWords{\representatives}\), \(u \in \withoutBinWords{\representatives} \iff \obsCounters(u) \neq \bot\), that is, $u \in \prefixes{\obsTable}$. Recall that $\prefixes{\obsTable}\subseteq \prefixes{L}$.
    We denote by $\withoutBinWords{\representatives}_{\ell - s}$ (resp. $\withoutBinWords{\representatives \cup \representatives\Sigma}_{\ell - s}$) the set of elements $u$ of $\withoutBinWords{\representatives}$ (resp. $\withoutBinWords{\representatives \cup \representatives\Sigma}$) such that $\heightAut{u} \leq \ell-s$.
    
    We consider the two subautomata: $(1)$ $\cal B$ equal to $\autTable$
    restricted to the states $\equivalenceClass{u}_{\equivTable}$ with $u \in
    \withoutBinWords{\representatives}_{\ell - s}$ and $(2)$ the limited
    behavior graph $\limBehaviorGraphAut[\ell - s]$. Both subautomata accept $L_{\ell-s}$ and all
    their states are reachable. We define $\isomorphism$ such that:
    \[
        \isomorphism(\equivalenceClass{u}_{\equivTable}) = \equivalenceClass{u}_{\equivBGROCA}, \mbox{ for all } u \in \withoutBinWords{R}_{\ell - s},
    \]
    with the aim to show that $\isomorphism$ embeds $\cal B$ to
    $\limBehaviorGraphAut[\ell - s]$. Therefore $\cal B$ and
    $\isomorphism({\cal B})$ will be isomorphic subautomata both accepting
    $L_{\ell - s}$. It will follow that the trim parts of $\cal B$ and
    $\limBehaviorGraphAut[\ell - s]$ will also be isomorphic as stated
    in \Cref{prop:isoPrefix}.
    
    First notice that by definition of $\isomorphism$, $\equivalenceClass{u}_{\equivBGROCA}$ is a state of $\limBehaviorGraphAut[\ell - s]$ because $u \in \prefixes{L}$ and $\heightAut{u} \leq \ell-s$, for all \(u \in \withoutBinWords{R}_{\ell - s}\). 
    
    Then let us show that:
    \begin{align}
        \forall u, v \in \withoutBinWords{\representatives \cup \representatives\Sigma}_{\ell - s}, u \equivTable v \iff u \equivBGROCA v, \label{eq:respectClasses}
    \end{align}
    in order to deduce that \(\isomorphism\) respects the equivalence classes.
    Let \(u, v \in \withoutBinWords{\representatives \cup \representatives\Sigma}_{\ell - s}\). By \Cref{long:prop:observation_table:u_equiv_v_implies_u_in_approx_v}, we already know that \(u \equivBGROCA v \implies u \equivTable v\).
    Assume \(u \equivTable v\) and, towards a contradiction, assume \(u \not\equivBGROCA v\).
    If \(\counterAut{u} \neq \counterAut{v}\), it holds that \(\obsCounters(u) \neq \obsCounters(v)\) (since \(u, v \in \withoutBinWords{\representatives \cup \representatives\Sigma}\)).
    It follows that \(u \notin \Approx(v)\), that is, \(u \not\equivTable v\), which is a contradiction.
    So, we get \(\counterAut{u} = \counterAut{v}\). 
    Since \(u \not\equivBGROCA v\), we know by \Cref{prop:behavior_graph:bounded_height} that there exists a witness \(w \in \Sigma^*\) such that \(uw \in L \iff vw \notin L\) and \(\heightAut{uw}, \heightAut{vw} \leq \ell\) (since \(s + \counterAut{u} \leq s + \ell - s = \ell\)).
    Since \(\languageOf{\autTable} = L_{\ell}\), it is impossible that \(\delta_{\observationTable}(q_{\observationTable}^0, u) = \delta_{\observationTable}(q_{\observationTable}^0, v)\).
    As $u,v \in \representatives \cup \representatives\Sigma$, by definition of $\autTable$, it follows that \(u \not\equivTable v\) which is again a contradiction.
    We have proved that \(u \equivTable v \iff u \equivBGROCA v\).
    
    Note that $\isomorphism$ puts the initial states of both subautomata, $\cal B$ and $\limBehaviorGraphAut[\ell - s]$, in correspondence: \(\isomorphism(q_{\equivTable}^0) = \isomorphism(\equivalenceClass{\emptyword}_{\equivTable}) = \equivalenceClass{\emptyword}_{\equivBGROCA} = q_{\equivBGROCA}^0\). Let us show that it is also the case for the final states. Let $u \in \withoutBinWords{R}_{\ell - s}$. If $\equivalenceClass{u}_{\equivTable} \in F_{\observationTable}$, then by definition, $\obsMapping(u) = 1$, that is, $u \in L_{\ell} \subseteq L$. It follows that $\equivalenceClass{u}_{\equivBGROCA} \in F_{\equivBGROCA}$. Conversely, if $\equivalenceClass{u}_{\equivBGROCA} \in F_{\equivBGROCA}$, then by definition $u \in L$ and moreover $\heightAut{u} \leq \ell - s \leq \ell$. Thus $u \in L_{\ell}$ and $\equivalenceClass{u}_{\equivTable} \in F_{\observationTable}$.   
    
    It remains to prove that $\isomorphism$ respects the transitions. Let $\equivalenceClass{u}_{\equivTable}, \equivalenceClass{u'}_{\equivTable} \in \withoutBinWords{R}_{\ell - s}$ such that $u' \equivTable ua$ with $a \in \Sigma$. By~\eqref{eq:respectClasses}, we have $u' \equivBGROCA ua$. By definition of both transition functions, we get $\delta_{\observationTable}(\equivalenceClass{u}_{\equivTable},a) = \equivalenceClass{u'}_{\equivTable}$ and $\delta_{\equivBGROCA}(\equivalenceClass{u}_{\equivBGROCA},a) = \equivalenceClass{u'}_{\equivBGROCA}$.

    From all the above points, we conclude that $\isomorphism$ embeds $\cal B$ in $\limBehaviorGraphAut[\ell - s]$ and that it is an isomorphism between their trim parts.
\qed\end{proof}

\section{Complexity of the learning algorithm}\label{sec:correctness}

This section is devoted to the proof of our main theorem (\Cref{thm:main}) that establishes the complexity of our learning algorithm for ROCAs as well as the number of queries for each kind. In particular, we explain how to make an observation table closed and consistent and prove that this is feasible in a finite amount of time. We also explain how the observation table grows during the execution of the learning algorithm.

Hence in this section, we suppose that we have an ROCA $\automaton$ that accepts a language $L \subseteq \Sigma^*$. We denote by $Q$ the set of states of $\automaton$ and by $\lengthCe$ the length of the longest counterexample returned by the teacher on (partial) equivalence queries. The learning algorithm for $L$ is similar to the one given for VCAs in \Cref{subsec:LearningVCAs} (see \Cref{alg:lstar_vca}). In this algorithm, an iteration in the while loop is called a \emph{round}: it consists in making the current table closed and consistent and then handling the counterexample provided either by a partial equivalence query or by an equivalence query. Notice that the total number of rounds performed by the learning algorithm coincides with the total number of partial equivalence queries. 

\subsection{Number of partial equivalence queries and equivalence queries}\label{subsec:NbrEquivQueries}

We begin by evaluating the number of (partial) equivalence queries of our learning algorithm. 

\begin{proposition}\label{prop:nbrEquivQueries}
In the learning algorithm,
\begin{itemize}
    \item the final counter limit \(\ell\) is bounded by \(\lengthCe\),
    \item the number of partial equivalence queries in in $\complexity(\lengthCe^3)$,
    \item the number of equivalence queries is in $\complexity(|Q| \lengthCe^2)$.
\end{itemize}
\end{proposition}
\begin{proof}
    The proof is inspired by~\cite{neider2010learning}.
    First, notice that the counter limit $\ell$ of \(\observationTable\) is increased only when a counterexample for an equivalence query is processed.
    Thus, $\ell$ is determined by the height of a counterexample, and this height cannot exceed \(\frac{\lengthCe}{2} \leq \lengthCe\). Moreover, it follows that the number of times $\ell$ is increased is in $\complexity(\lengthCe)$.

    Second, let us study the number of partial equivalence queries for a fixed counter limit $\ell$. Recall that the index of $\equivTable$ strictly increases after each provided counterexample (see \Cref{prop:betterApproximation}) and that the number of equivalence classes of \(\equivBGROCA\) up to \(\ell\) is bounded by $(\ell + 1)|Q| + 1$ (see \Cref{lem:sizeEquiv}). By \Cref{long:prop:observation_table:u_equiv_v_implies_u_in_approx_v}, it follows that we ask at most \(\complexity(\sizeOfSet{Q} \lengthCe)\) partial equivalence queries by fixed \(\ell\).
    Therefore, the total number of partial equivalence queries is \(\complexity(\sizeOfSet{Q} \lengthCe^2)\).

    Finally, let us study the number of equivalence queries when a DFA $\autTable$ accepting $L_{\ell}$ has been learned. We generate at most \(\ell^2\) periodic descriptions (as the number of pairs of offset and period is bounded by \(\ell^2\)). We thus ask an equivalence query for the ROCA constructed from each such description, that is, at most \(\ell^2 = \complexity(\lengthCe^2)\) queries. This leads to a total number of equivalence queries in \(\complexity(\lengthCe^3)\).
\qed\end{proof}

\subsection{Making a table closed and consistent and handling counterexamples}\label{subsec:makingClosed}

In this section, we suppose we have an observation table $\observationTable = (\representatives, \separatorsCounters, \separatorsMapping, \obsMapping, \obsCounters)$ up to counter limit $\ell$. We denote by $\sizeOfSet{\obsTable}$ the size of $\obsTable$ which is given by the size of both $\representatives \cup \representatives\Sigma$ and $\separatorsMapping$.

In a first step, we focus on making $\obsTable$ closed, $\Sigma$- and \prefixed.
This requires extending $\observationTable$ by adding new words to $\representatives \cup \representatives\Sigma$,  $\separatorsCounters$, or $\separatorsMapping$ and updating the known values of $\observationTable$ since $\prefTable$ may change. In the sequel, the new table is denoted by $\observationTable' = (\representatives', \separatorsCounters', \separatorsMapping', \obsMapping', \obsCounters')$ up to the same level $\ell$. An observation table that is not closed (resp.\ not \consistent, not \prefixed) is called \emph{open} (resp.\ \emph{\inconsistent},  \emph{\notPrefixed}). We explain below how to \emph{resolve} an openness (resp.\ \inconsistency, \inconsistencyBot).

In a second step, we focus on how to handle a counterexample provided by the teacher after a partial equivalence query or an equivalence query. This also requires to extend and update $\obsTable$ into a new table $\observationTable[\ell']'$ such that $\ell' = \ell$ in case of partial equivalence query and $\ell' > \ell$ in case of equivalence query. 

\subsubsection{Making a table closed.}

Assume \(\observationTable\) is open, i.e., 
    \[
    \exists u \in \representatives\Sigma \mbox{ such that } \Approx(u) \cap \representatives = \emptyset.
    \]
We say that we have a \emph{$u$-openness}. We want to extend and update $\observationTable$ to resolve it. Notice that we have $u \not\in R$ since $u \in \Approx(u)$ and $\Approx(u) \cap \representatives = \emptyset$.
We thus define the new table \(\observationTable' = (\representatives', \separatorsCounters, \separatorsMapping, \obsMapping', \obsCounters')\) such that \(\representatives' = \representatives \cup \{u\}\) ($\separatorsCounters, \separatorsMapping$ are left unchanged). 
Then, \(u\) is no longer a word making the table open, i.e., the previous $u$-openness is resolved.
Moreover, we have:

\begin{lemma}\label{lem:resolveOpenness1}
    Let \(\observationTable\) be an observation table with a $u$-openness, \(u \in \representatives\Sigma\). Let \(\observationTable'\) be the new table as described above.
    Then $u \not \in \representatives$ and $\Approx'(u) \cap \representatives' = \{u\}$.
\end{lemma}
\begin{proof}
    We already explained that $u \not\in \representatives$. 
    Since \(\representatives' = \representatives \cup \{u\}\) and $u \in \Approx'(u)$, it suffices to show that \(\Approx'(u) \cap \representatives = \emptyset\) to get $\Approx'(u) \cap \representatives' = \{u\}$.
    By contradiction, assume \(\exists v \in \representatives, v \in \Approx'(u)\).
    By \Cref{lem:observation_table:u_in_approx_prime_v_implies_u_in_approx_v}, since \(v \in \representatives\) and \(u \in \representatives \Sigma\), we have \(v \in \Approx(u)\).
    However, due to the $u$-openness, we know that \(\Approx(u) \cap \representatives = \emptyset\).
    We thus have a contradiction and $\Approx'(u) \cap \representatives' = \{u\}$.
\qed\end{proof}

The next lemma will be useful in \Cref{subsec:growth} when we will study how the observation table grows during the learning process and what the total number of membership and counter value queries are.

\begin{lemma}\label{lem:resolveOpenness2}
Let \(\observationTable\) be an observation table with a $u$-openness, \(u \in \representatives\Sigma\). Let \(\observationTable'\) be the new table as described above. Then,
\begin{itemize}
    \item $|\representatives'| = |\representatives| + 1$, $|\separatorsCounters'| = |\separatorsCounters|$, and $|\separatorsMapping'| = |\separatorsMapping|$,
    \item the number of membership queries and the number of counter value queries are both bounded by a polynomial in $\sizeOfSet{\obsTable}$. 
\end{itemize}
\end{lemma}
\begin{proof}
    The first item is easily proved. Let us prove the second one. To extend
    and update $\obsTable$, the learner asks queries to the teacher in order to compute the new value $\obsMapping'(uas)$ (resp. $\obsCounters'(uas)$) for each $a \in \Sigma$ and each $s \in \separatorsMapping$ (resp. $s \in \separatorsCounters$). Moreover, as $\representatives' \neq \representatives$, the set $\prefTable[\observationTable']$ may strictly include $\prefTable$, and thus all the values $\obsCounters'(u's)$, $u' \in \representatives \cup \representatives\Sigma$, $s \in \separatorsCounters$ must be recomputed. By \Cref{lem:nbrQueries}, this requires a polynomial number of membership and counter value queries in $\sizeOfSet{\obsTable}$.
\qed\end{proof}

\subsubsection{Making a table \texorpdfstring{\consistent}{Σ-consistent}.}

Assume \(\observationTable\) is \inconsistent, i.e., \(\exists ua \in \representatives\Sigma\) such that \(ua \notin \bigcap_{v \in \Approx(u) \cap \representatives} \Approx(va)\).
In other words, 
    \[ 
    \exists ua \in \representatives\Sigma, \exists v \in \representatives \mbox{ such that } 
    v \in \Approx(u) \mbox{ and } ua \notin \Approx(va).
    \]
We say that we have a \emph{$(u,v,a)$-\inconsistency}. Let us explain how to resolve it. 
We have two cases: there exists $s \in \separatorsCounters$ such that either \(\obsMapping(uas) \neq \obsMapping(vas)\), or \(\obsCounters(uas) \neq \bot \land \obsCounters(vas) \neq \bot \land \obsCounters(uas) \neq \obsCounters(vas)\).
In both cases, let $\observationTable' = (\representatives, \separatorsCounters', \separatorsMapping', \obsMapping', \obsCounters')$ be the new table such that $\separatorsCounters' = \separatorsCounters \cup \{as\}$ and $\separatorsMapping' = \separatorsMapping \cup \{as\}$ ($\representatives$ is left unchanged).
Notice that we add \(as\) to \(\separatorsMapping\) in order to maintain \(\separatorsCounters' \subseteq \separatorsMapping'\).
The next lemma indicates that the $(u,v,a)$-\inconsistency is resolved since it states that $v$ no longer belongs to $\Approx'(u)$. Moreover it states that this \approxSet\ gets smaller.

\begin{lemma}\label{lem:resolveSigmaInc1}
    Let \(\observationTable\) be an observation table with a $(u,v,a)$-\inconsistency, $ua \in \representatives\Sigma$, $v \in \representatives$. Let \(\observationTable'\) be the new table as described above. 
    Then, $v \notin \Approx'(u)$ and \(\sizeOfSet{\Approx(u)} > \sizeOfSet{\Approx'(u)}\).
\end{lemma}
\begin{proof}
    Let us first explain why \(v \notin \Approx'(u)\). We go back to the two previous cases by considering the suffix $as \in \separatorsCounters'$. Either we have \(\obsMapping(uas) \neq \obsMapping(vas)\). Since we do not change \(\ell\), it follows that we have the same inequality \(\obsMapping'(uas) \neq \obsMapping'(vas)\). Thus, by definition, \(v \notin \Approx'(u)\). Or we have \(\obsCounters(uas) \neq \bot \land \obsCounters(vas) \neq \bot \land \obsCounters(uas) \neq \obsCounters(vas)\). Since \(\obsCounters(uas) \neq \bot\), we have \(\obsCounters'(uas) =  \obsCounters(uas) \neq \bot\) (likewise for \(vas\)). Therefore, it holds that \(\obsCounters'(uas) \neq \bot \land \obsCounters'(vas) \neq \bot \land \obsCounters'(uas) \neq \obsCounters'(vas)\). Thus again we get \(v \notin \Approx'(u)\).
    
    Since \(v \in \Approx(u)\), \(v \notin \Approx'(u)\) and $\representatives = \representatives'$, we get \(\sizeOfSet{\Approx(u)} > \sizeOfSet{\Approx'(u)}\) by \Cref{lem:observation_table:u_in_approx_prime_v_implies_u_in_approx_v}.
\qed\end{proof}

\begin{lemma}\label{lem:resolveSigmaInc2}
Let \(\observationTable\) be an observation table with a $(u,v,a)$-\inconsistency, $ua \in \representatives\Sigma$, $v \in \representatives$. Let \(\observationTable'\) be the new table as described above. Then,
\begin{itemize}
    \item $|\representatives'| = |\representatives|$, $|\separatorsCounters'| = |\separatorsCounters| + 1$, and $|\separatorsMapping'| = |\separatorsMapping| + 1$,
    \item the number of membership queries and the number of counter value queries are both bounded by a polynomial in $\sizeOfSet{\obsTable}$. 
\end{itemize}
\end{lemma}
\begin{proof}
    The first item is immediate. To extend and update $\obsTable$, the learner
    asks queries to the teacher so as to compute the new value $\obsMapping'(u'as)$ (resp. $\obsCounters'(u'as)$) for each $u' \in \representatives \cup \representatives\Sigma$. Moreover, as $\prefTable[\observationTable']$ may have changed, all the values $\obsCounters'(u's')$, $u' \in \representatives \cup \representatives\Sigma$, $s' \in \separatorsCounters$ must be recomputed. By \Cref{lem:nbrQueries}, this requires a polynomial number of membership and counter value queries in $\sizeOfSet{\obsTable}$.
\qed\end{proof}

\subsubsection{Making a table \texorpdfstring{\prefixed}{⊥-consistent}.}

Assume \(\observationTable\) is \notPrefixed, i.e., 
    \[ 
    \exists u, v \in \representatives \cup \representatives \Sigma, \exists s \in \separatorsCounters \mbox{ such that } u \in \Approx(v) \mbox{ and } \obsCounters(us) \neq \bot \Leftrightarrow \obsCounters(vs) = \bot.
    \]

We say that we have a \emph{$(u,v,s)$-\inconsistencyBot} and we call \emph{\mismatch} the disequality $\obsCounters(us) \neq \bot \Leftrightarrow \obsCounters(vs) = \bot$. We now explain how to resolve this \inconsistencyBot.
Let us assume, without loss of generality, that \(\obsCounters(us) \neq \bot\) and \(\obsCounters(vs) = \bot\).
So, \(us \in \prefTable\), i.e., there exist \(u' \in \representatives \cup \representatives\Sigma\) and \(s' \in \separatorsMapping\) such that \(us \in \prefixes{u's'}\) and \(\obsMapping(u's') = 1\).
We denote by $s''$ the word such that $us'' = u's'$. Notice that \(s\) is a prefix of \(s''\). We have two cases according to whether $u'$ is a prefix of $u$ (see \Cref{fig:observation_table:towards_prefixed:u_prime_prefix_u}) or $u$ is a proper prefix of $u'$ (see \Cref{fig:observation_table:towards_prefixed:u_prefix_u_prime}).

\begin{figure}
    \begin{subfigure}{.45\textwidth}
        \begin{tikzpicture}[
    auto
]
    \path
        coordinate (eps)    at (0, 0)
        coordinate (u')     at (1.5, 0)
        coordinate (u)      at (2.2, 0)
        coordinate (u s)    at (5, 0)
        coordinate (u s'')  at (5.4, 0)
    ;
    \draw
        (eps) -- (u s'')
    ;
    \foreach \start/\end/\lab in {eps/u/u, u/u s/s} {
        \draw
            let
                \p1 = ($(\start)!.5!(\end)$),
                \p{mid} = (\x1, 0.7)
            in
            (\start) to [out=45, in=180] node [at end] {\(\lab\)}   (\p{mid})
                    to [out=0, in=135]                              (\end)
        ;
    }
    \foreach \start/\end/\lab in {eps/u'/u', u'/u s''/s'} {
        \draw
            let
                \p1 = ($(\start)!.5!(\end)$),
                \p{mid} = (\x1, -0.7)
            in
            (\start) to [out=-45, in=-180] node [at end, below] {\(\lab\)} (\p{mid})
                    to [out=0, in=-135]                             (\end)
        ;
    }
    \draw
        let
            \p1 = ($(u)!.5!(u s'')$),
            \p{mid} = (\x1, 1.2)
        in
        (u) to [out=45, in=180] node [at end] {\(s''\)} (\p{mid})
            to [out=0, in=135]                          (u s'')
    ;
\end{tikzpicture}
        \caption{\(u'\) is a prefix of \(u\).}%
        \label{fig:observation_table:towards_prefixed:u_prime_prefix_u}
    \end{subfigure}
    \hfill 
    \begin{subfigure}{.45\textwidth}
        \begin{tikzpicture}[
    auto
]
    \path
        coordinate (eps)    at (0, 0)
        coordinate (u)      at (2.2, 0)
        coordinate (u')     at (3.7, 0)
        coordinate (u s)    at (5, 0)
        coordinate (u s'')  at (5.4, 0)
    ;
    \draw
        (eps) -- (u s'')
    ;
    \foreach \start/\end/\lab in {eps/u/u, u/u s/s} {
        \draw
            let
                \p1 = ($(\start)!.5!(\end)$),
                \p{mid} = (\x1, 0.7)
            in
            (\start)    to [out=45, in=180] node [at end] {\(\lab\)}    (\p{mid})
                        to [out=0, in=135]                              (\end)
        ;
    }
    \foreach \start/\end/\lab in {eps/u'/u', u'/u s''/s'} {
        \draw
            let
                \p1 = ($(\start)!.5!(\end)$),
                \p{mid} = (\x1, -0.7)
            in
            (\start)    to [out=-45, in=-180] node [at end, below] {\(\lab\)}   (\p{mid})
                        to [out=0, in=-135]                                     (\end)
        ;
    }
    \draw
        let
            \p1 = ($(u)!.5!(u s'')$),
            \p{mid} = (\x1, 1.2)
        in
        (u) to [out=45, in=180] node [at end] {\(s''\)} (\p{mid})
            to [out=0, in=135]                          (u s'')
    ;
\end{tikzpicture}
        \caption{\(u\) is a proper prefix of \(u'\).}%
        \label{fig:observation_table:towards_prefixed:u_prefix_u_prime}
    \end{subfigure}
    \caption{The different words used to resolve a $(u,v,s)$-\inconsistencyBot.}%
    \label{fig:observation_table:towards_prefixed}
\end{figure}

\begin{itemize}
\item We first suppose that $u'$ is a prefix of $u$. 

Let us show that $s'' \in \separatorsMapping \setminus \separatorsCounters$. As \(\separatorsMapping\) is suffix-closed, and $s''$ is a suffix of $s' \in \separatorsMapping$, it follows that $s''$ belongs to $\separatorsMapping$. By contradiction, assume $s'' \in \separatorsCounters$. Since \(u \in \Approx(v)\) and \(\obsMapping(us'') = \obsMapping(u's') = 1\), we have \(\obsMapping(vs'') = \obsMapping(us'') = 1\). Moreover, since \(s\) is a prefix of \(s''\), it holds that \(vs \in \prefTable\). This means that \(\obsCounters(vs) \neq \bot\) which is a contradiction.

To resolve the $(u,v,s)$-\inconsistencyBot, we consider the new table $\observationTable' = (\representatives, \separatorsCounters', \separatorsMapping, \obsMapping', \obsCounters')$ such that $\separatorsCounters' = \separatorsCounters \cup \suffixes{s''}$ ($\representatives$ and $\separatorsMapping$ are left unchanged).

\item We then suppose that \(u\) is a proper prefix of \(u'\). 

Then, to resolve the $(u,v,s)$-\inconsistencyBot, we have two cases: 
\begin{itemize}
    \item Either \(vs'' \in L_{\ell}\). Let $\observationTable' = (\representatives, \separatorsCounters, \separatorsMapping', \obsMapping', \obsCounters')$ be the updated table such that $\separatorsMapping' = \separatorsMapping \cup \suffixes{s''}$ ($\representatives$ and $\separatorsCounters$ are left unchanged).
    \item Or \(vs'' \notin L_{\ell}\). 
        This time, we consider the new table $\observationTable' = (\representatives, \separatorsCounters', \separatorsMapping', \obsMapping', \obsCounters')$ such that $\separatorsCounters' = \separatorsCounters \cup \suffixes{s''}$ and $\separatorsMapping' = \separatorsMapping \cup \suffixes{s''}$ ($\representatives$ is left unchanged).
\end{itemize}
\end{itemize}

The next lemma indicates that the $(u,v,s)$-\inconsistencyBot is indeed resolved, i.e., $u \not\in \Approx'(v)$ or the \mismatch $\obsCounters(us) \neq \bot \Leftrightarrow \obsCounters(vs) = \bot$ is eliminated.

\begin{lemma}\label{lem:resolveBotInc1}
Let \(\observationTable\) be an observation table with a $(u,v,s)$-\inconsistencyBot, $u, v \in \representatives \cup \representatives \Sigma$, $s \in \separatorsCounters$. Let \(\observationTable'\) be the new table as described above. Then,
\begin{itemize}
        \item If $u'$ is a prefix of $u$, then $u \notin \Approx'(v)$.
        \item If $u$ is a proper prefix of $u'$, then,
        \begin{itemize}
        \item if $vs'' \in L_\ell$, then $u \in \Approx'(v) \implies \obsCounters'(us) = \obsCounters'(vs)$.
        \item if $vs'' \notin L_\ell$, then $u \notin \Approx'(v)$.
        \end{itemize}
    \end{itemize}
In particular, either \(\sizeOfSet{\Approx(v)} > \sizeOfSet{\Approx'(v)}\) or the mismatch $\obsCounters(us) \neq \bot \Leftrightarrow \obsCounters(vs) = \bot$ is eliminated.
\end{lemma}
\begin{proof}
    Suppose $u'$ is a prefix of $u$, and consider $s'' \in \separatorsMapping$. We have \(\obsMapping(vs'') = 0\) since otherwise, as explained above, \(vs \in \prefTable\) in contradiction with \(\obsCounters(vs) = \bot\). Therefore, $\obsMapping'(vs'') = \obsMapping(vs'') = 0$. We also have \(\obsMapping'(us'') = \obsMapping'(u's') = \obsMapping(u's') = 1\) (by definition of $u's'$). Since $s'' \in \separatorsCounters'$, it follows that $u \notin \Approx'(v)$.

    Suppose that $u$ is a proper prefix of $u'$. The subcase $vs'' \notin L_\ell$ is solved similarly. Consider again $s'' \in \separatorsCounters'$. By hypothesis we have \(\obsMapping'(vs'') = 0\). Recall that we have \(\obsMapping'(us'') = 1\). Therefore $u \notin \Approx'(v)$.
    
    In the other subcase ($vs'' \in L_\ell$), we have $vs \in \prefixes{\obsTable'}$ since $s'' \in \separatorsMapping'$, $\obsMapping'(vs'') = 1$ and $vs$ is a prefix of $vs''$. As $us \in \prefixes{\obsTable}$ by hypothesis, we also have $us \in \prefixes{\obsTable'}$. Hence $\obsCounters'(us) \neq \bot \land \obsCounters'(vs) \neq \bot$. If $v \in \Approx'(u)$, it follows that $\obsCounters'(us) = \obsCounters'(vs)$.
    
    Therefore in all cases, either the size of $\Approx(v)$ has decreased (\Cref{lem:observation_table:u_in_approx_prime_v_implies_u_in_approx_v}) or the mismatch $\obsCounters(us) \neq \bot \Leftrightarrow \obsCounters(vs) = \bot$ is eliminated.
\qed\end{proof}

\begin{lemma}\label{lem:resolveBotInc2}
Let \(\observationTable\) be an observation table with a $(u,v,s)$-\inconsistencyBot, $u, v \in \representatives \cup \representatives \Sigma$, $s \in \separatorsCounters$. Let \(\observationTable'\) be the new table as described above. Then,
\begin{itemize}
    \item $|\representatives'| = |\representatives|$, $|\separatorsCounters'| \leq |\representatives \cup \representatives\Sigma| \cup |\separatorsMapping|$, and $|\separatorsMapping'| \leq |\representatives \cup \representatives\Sigma| \cup |\separatorsMapping|$,
    \item if the mismatch $\obsCounters(us) \neq \bot \Leftrightarrow \obsCounters(vs) = \bot$ is eliminated, we have $|\separatorsCounters'| = |\separatorsCounters|$,
    \item the number of membership queries and the number of counter value queries are both bounded by a polynomial in $\sizeOfSet{\obsTable}$. 
\end{itemize}
\end{lemma}

\begin{proof}
    We consider the different cases as above, which all together will lead to the statements of this lemma. 
    
    Suppose first that $u'$ is a prefix of $u$. Then by construction of $\obsTable'$, we have $\representatives' = \representatives$, $\separatorsCounters' = \separatorsCounters \cup \suffixes{s''}$, and $\separatorsMapping' = \separatorsMapping$. As $\separatorsCounters' \subseteq \separatorsMapping'$, we get $|\separatorsCounters'| \leq |\separatorsMapping|$. Moreover, as $\representatives$, $\separatorsMapping$ are unchanged and $\suffixes{s''}$ is added to $\separatorsCounters$, we only need to add the new values $\obsCounters'(yz)$ to $\observationTable'$, for all $y \in \representatives \cup \representatives\Sigma$ and $z \in \suffixes{s''}$. Therefore there is no membership query and a polynomial number of counter value queries in $\sizeOfSet{\obsTable}$ (by \Cref{lem:nbrQueries}).
    
    Second, suppose that $u$ is a proper prefix of $u'$. We have two subcases according to whether $vs'' \in L_\ell$. To determine which subcase applies, by \Cref{lem:nbrQueries} and the preceding argument, this requires one membership query and at most $|vs''|$ counter value queries. The latter number is polynomial in the size of $\obsTable$. Indeed we have $v \in \representatives \cup \representatives\Sigma$, $s'' = xs'$ with $s' \in \separatorsMapping$ and $x$ prefix of $u' \in \representatives \cup \representatives\Sigma$ (see \Cref{fig:observation_table:towards_prefixed:u_prefix_u_prime}), and thus $|vs''| \leq |v| + |u'| + |s'|$.
    
    Consider the first subcase $vs'' \in L_\ell$. Then by construction of $\obsTable'$, we have $\representatives' = \representatives$, $\separatorsCounters' = \separatorsCounters$, and $\separatorsMapping' = \separatorsMapping \cup  \suffixes{s''}$. Recall that as $s' \in \separatorsMapping$ is a suffix of $s''$, there are only $|x|$ suffixes to add to $\separatorsMapping$ with $|x| \leq |u'|$. Thus $|\separatorsMapping'| \leq |\separatorsMapping| + |\representatives \cup \representatives\Sigma|$. Notice also that this is the only subcase where the mismatch $\obsCounters(us) \neq \bot \Leftrightarrow \obsCounters(vs) = \bot$ may be eliminated and for which we have $|S'| = |S|$. For each new suffix $z$ to add to $\separatorsMapping$ and each $y \in \representatives \cup \representatives\Sigma$, we have to add the new value $\obsMapping'(yz)$ to $\obsTable'$, and this requires a polynomial number of membership and counter value queries in $\sizeOfSet{\obsTable}$ (by \Cref{lem:nbrQueries}). Moreover, as $\prefixes{\obsTable'}$ may have changed, the values $\obsCounters'(yz)$ have to be recomputed for all $y \in \representatives \cup \representatives\Sigma$, $z \in \separatorsCounters$. This requires again a polynomial number of counter value queries in $\sizeOfSet{\obsTable}$.
    
    Finally consider the second subcase $vs'' \not\in L_\ell$. Then by construction of $\obsTable'$, we have $\representatives' = \representatives$, $\separatorsCounters' = \separatorsCounters \cup \suffixes{s''}$, and $\separatorsMapping' = \separatorsMapping \cup  \suffixes{s''}$. Therefore this case is similar to the previous subcase with the exception that $\separatorsCounters' = \separatorsCounters \cup \suffixes{s''}$. The same arguments can be repeated for $\separatorsMapping$. As $\separatorsCounters' \subseteq \separatorsMapping'$, it follows that $|\separatorsCounters'| \leq |\separatorsMapping| + |\representatives \cup \representatives\Sigma|$. The way that $\obsTable'$ is extended and updated can also be repeated, which leads again to a polynomial number of membership and counter value queries. 
\qed\end{proof}

\subsubsection{Handling counterexamples.}
Suppose now that a counterexample $w \in \Sigma^*$ is provided by the teacher after a partial equivalence query or an equivalence query. This also requires to extend and update $\obsTable$ into a new table $\observationTable[\ell']' = (\representatives', \separatorsCounters', \separatorsMapping', \obsMapping[\ell']', \obsCounters[\ell']')$ such that $\ell' = \ell$ in case of partial equivalence query and $\ell' > \ell$ in case of equivalence query. In both cases, the new table is defined such that $\representatives' = \representatives \cup \prefixes{w}$, $\separatorsCounters' = \separatorsCounters$, and $\separatorsMapping' = \separatorsMapping$. We have the following lemma. 

\begin{lemma}\label{lem:resolveCe}
    Let \(\observationTable\) be an observation table and \(\observationTable[\ell']'\) be the new table as described above. Then,
    \begin{itemize}
        \item $|\representatives'| \leq |\representatives| + \lengthCe$, $|\separatorsCounters'| = |\separatorsCounters|$, and $|\separatorsMapping'| = |\separatorsMapping|$,
        \item the number of membership queries and the number of counter value queries are both bounded by a polynomial in $\lengthCe$ and $\sizeOfSet{\obsTable}$. 
    \end{itemize}
\end{lemma}
\begin{proof}
    The first statement is easily proved. The worst case with respect to the number of queries happens when $\ell' > \ell$. In this case, we have to compute all the values $\obsMapping[\ell']'(us)$, $u \in \representatives' \cup \representatives'\Sigma$, $s \in \separatorsMapping$ and all the values $\obsCounters[\ell']'(us)$, $u \in \representatives' \cup \representatives'\Sigma$, $s \in \separatorsCounters$. By \Cref{lem:nbrQueries}, this requires a polynomial number of membership and counter value queries in $\lengthCe$ and $\sizeOfSet{\obsTable}$.
\qed\end{proof}

\subsection{Growth of the table and number of membership and counter value queries}\label{subsec:growth}

In this section, we study the size of the observation table at the end of the execution of our learning algorithm. We respectively denote by $\Rf, \SHatf$ the final sets $\representatives, \separatorsMapping$ at the end of the execution. With this study, we will be able to show that making an observation table closed and consistent can be done in a finite amount of time (\Cref{prop:finiteProcess}). We will also be able to count the total number of membership and counter value queries performed during the execution of the algorithm. 

We begin with the study of the size of $\Rf \cup \Rf\Sigma$.

\begin{proposition}\label{prop:nbrOpen-sizeR}
\begin{itemize}
    \item The size of $\Rf \cup \Rf\Sigma$ is in $\complexity(|Q| |\Sigma| \lengthCe^4)$, thus polynomial in $|Q|, |\Sigma|$ and $\lengthCe$.
    \item During the process of making a table closed, $\Sigma$- and \prefixed, the number of resolved openness cases is in $\complexity(|Q|\lengthCe)$. 
\end{itemize}
\end{proposition}
\begin{proof}
    We first study the growth of $\representatives$ and the number of resolved openness cases after one round of the learning algorithm, that is, after first making the current table $\obsTable$ closed and consistent, and then handling a counterexample to a (partial) equivalence query.
    
    During the process of making $\obsTable$ closed and consistent, notice that $\representatives$ increases only when an openness is resolved (by \Cref{lem:resolveOpenness2,lem:resolveSigmaInc2,lem:resolveBotInc2}). By \Cref{lem:resolveOpenness1}, when a $u$-openness is resolved, $\representatives$ is increased with $u$ which is the only representative in its new \approxSet. By \Cref{long:prop:observation_table:u_equiv_v_implies_u_in_approx_v}, it follows that the number of resolved openness cases is bounded by the index of $\equiv$ up to counter limit $\ell$. This index is bounded by $(\ell + 1) |Q| + 1$ (see \Cref{lem:sizeEquiv}).
    As $\ell \leq \lengthCe$ by \Cref{prop:nbrEquivQueries}, the number of resolved openness cases is in $\complexity(|Q|t)$ and $\withoutBinWords{\representatives}$ may increase by at most $(\lengthCe + 1) |Q| + 1$ words. By adding the potential \binWord of $\representatives \setminus \withoutBinWords{\representatives}$, we get that $\representatives$ may increase by at most $(\lengthCe + 1) |Q| + 2$ words.

    After that, to handle the counterexample provided by the teacher, $R$ may still increase by at most $\lengthCe$ words (by \Cref{lem:resolveCe}).
    
    Let us now study the size of $\Rf$. We denote by $r(i)$ the size of $\representatives$ respectively at the initialization of the learning algorithm when $i = 0$ and after each round $i$, when $i \geq 1$. We have 
    \begin{align*}
    r(0) &=  1,\\
    r(i) &\leq r(i-1) + (t + 1) |Q| + 2 + \lengthCe, \forall i > 0. 
    \end{align*}
    Recall that the number of rounds is bounded by the total number of partial equivalence queries, which is in $\complexity(\lengthCe^3)$ by \Cref{prop:nbrEquivQueries}. Hence, the size of $\Rf$ is in $\complexity(|Q|t^4)$ and the size of $\Rf \cup \Rf\Sigma$ is in $\complexity(|Q||\Sigma|t^4)$.
\qed\end{proof}

In the next proposition, we study how the sizes of $\separatorsCounters$ and $\separatorsMapping$ increase when making an observation table closed and consistent. Recall that both $\separatorsCounters$ and  $\separatorsMapping$ do not change when resolving openness cases (see \Cref{lem:resolveOpenness1}). Moreover during the process of making the table consistent (see \Cref{lem:resolveSigmaInc1,lem:resolveBotInc1}), either an \approxSet\ decreases in size or a mismatch is eliminated. Below, we study how many times each of these two cases may happen.

\begin{proposition}\label{prop:NbrInc-sizeS}
Let $\obsTable$ be an observation table and let $\obsTable'$ be the resulting table after making $\obsTable$ closed, $\Sigma$- and \prefixed. Then, with $\RfSize$ being the size of $\Rf \cup \Rf\Sigma$,
\begin{itemize}
    \item $\sizeOfSet{\separatorsCounters'}$ is in $\complexity(\RfSize^{4 \RfSize^2 + 2} |\separatorsMapping|)$ and $\sizeOfSet{\separatorsMapping'}$ is in $\complexity(\RfSize^{4 \RfSize^2 + 4} |\separatorsMapping|)$.
    \item the number of times that an \approxSet\ decreases is bounded by $\RfSize^2$ and the number of times a mismatch is eliminated is in $\complexity(\RfSize^{4 \RfSize^2 + 3} |\separatorsMapping|)$. 
\end{itemize}
\end{proposition}
\begin{proof}
    Resolving $\Sigma$- and $\bot$-inconsistencies during the process of making $\obsTable$ closed and consistent either decreases \approxSet s or eliminates mismatches, and the latter case only occurs when resolving a \inconsistencyBot (see \Cref{lem:resolveSigmaInc1,lem:resolveBotInc1}). 
    
    Let us first study the number of times that an \approxSet\ may decrease. By definition, each processed \approxSet\ is a subset of $\representatives' \cup \representatives'\Sigma$ and there are at most $|\representatives' \cup \representatives'\Sigma|$ such sets. Therefore the number of times an \approxSet\ decreases is bounded by $|\representatives' \cup \representatives'\Sigma|^2$, and thus by 
    \begin{align}
           \RfSize^2. \label{eq:alphaCarre} 
    \end{align}

    Let us now study the growth of both $\separatorsCounters$ and $\separatorsMapping$. By \Cref{lem:resolveOpenness2,lem:resolveSigmaInc2,lem:resolveBotInc2}, $\separatorsCounters$ and $\separatorsMapping$ may increase when resolving $\Sigma$- and $\bot$-inconsistencies only. Let us denote by $\hat{s}$ the initial size of $\separatorsMapping$ and by $\hat{s}(i,j)$ the size of the current set $\separatorsMapping$ after $i+j$ steps composed of $i$ eliminations of mismatches and $j$ decreases of \approxSet s. During such steps, by \Cref{lem:resolveSigmaInc2,lem:resolveBotInc2}, the size of $\separatorsMapping$ increases by at most $|\representatives' \cup \representatives'\Sigma| \leq \RfSize$ words. Therefore we get:
    \begin{align}
        \hat{s}(0,0) &=  \hat{s},\\
        \hat{s}(i,j) &\leq (i+j) \RfSize + \hat{s}, \forall i+j > 0.  \label{eq:sij}
    \end{align}
    
    Let us introduce another notation: $s(j)$ is the size of the current set $\separatorsCounters$ after $j$ decreases of \approxSet s and a certain number of eliminations of mismatches such that the last step is one decrease of an \approxSet. Thus $s(0) \leq \hat{s}$ (since $\separatorsCounters \subseteq \separatorsMapping$) and from $s(j-1)$ to $s(j)$, a certain number of mismatches is resolved followed by one decrease of an \approxSet. Let us recall that when a mismatch is resolved, the current $\separatorsCounters$ is unchanged (see \Cref{lem:resolveBotInc2}). It follows that from $s(j-1)$ to $s(j)$, at most:
    \begin{align}
        s(j-1) \RfSize \mbox{ mismatches} \label{eq:mismatch} 
    \end{align}
    are resolved (indeed for a fixed $s(j-1)$, at most $\sizeOfSet{\representatives \cup \representatives\Sigma}$ mismatches can be resolved). Hence, from~\eqref{eq:mismatch} we get $s(j) \leq \hat{s}(s(j-1)\RfSize,j)$ and from~\eqref{eq:sij} we get $s(j) \leq (s(j-1)\RfSize + j) \RfSize + \hat{s}$. It follows that:
    \begin{align*}
        s(0) &\leq \hat{s}, \\
        s(j) &\leq s(j-1)\RfSize^2 j + \hat{s}.
    \end{align*}
    By \Cref{lemma:bounds_recurrence} in \Cref{sec:bounds_recurrence}, we get $s(j) \leq (j+1){(\RfSize^2j)}^j\hat{s}$. Recall that there are at most $\RfSize^2$ decreases of \approxSet s by~\eqref{eq:alphaCarre}. And, to have the table closed and consistent, the last such decrease could be followed by several eliminations of mismatches that do not change the size of the current $\separatorsCounters$. Therefore, at the end of the process, the size of the resulting set $\separatorsCounters'$ is equal to $s(\RfSize^2) \leq (\RfSize^2 + 1){(\RfSize^2 \RfSize^2)}^{\RfSize^2}\hat{s}$ which is in $\complexity(\RfSize^{4 \RfSize^2 + 2} |\separatorsMapping|)$. By~\eqref{eq:mismatch}, we also get that the total number of eliminations of mismatches is:
    \begin{align}
        s(\RfSize^2)\RfSize \label{eq:nbrMismatch}
    \end{align}
    which is in $\complexity(\RfSize^{4 \RfSize^2 + 3} |\separatorsMapping|)$. Finally we can deduce the size of the resulting set $\separatorsMapping'$ when the table is closed and consistent: it is equal to $\hat{s}(i,j)$ with $i = s(\RfSize^2)\RfSize$ by~\eqref{eq:nbrMismatch} and $j = \RfSize^2$ by~\eqref{eq:alphaCarre}. By~\eqref{eq:sij}, we get $\sizeOfSet{\separatorsMapping'} \leq (s(\RfSize^2)\RfSize + \RfSize^2)\RfSize + \hat{s}$ which is in $\complexity(\RfSize^{4 \RfSize^2 + 4} |\separatorsMapping|)$.
\qed\end{proof}

As a corollary of \Cref{prop:nbrOpen-sizeR,prop:NbrInc-sizeS}, we obtain that a table can be made closed and consistent in a finite amount of time (\Cref{prop:finiteProcess}).
In view of these two results, we decided to handle a counterexample $w$ to a (partial) equivalence query by adding $\prefixes{w}$ to the current set $\representatives$. Indeed, an alternative could have been to  add $\suffixes{w}$ to $\separatorsMapping$. It should however be clear that this is not such a good idea because of the exponential growth of $\separatorsMapping$ as established in \Cref{prop:NbrInc-sizeS}. In contrast, the size of $\representatives$ at the end of the learning process is only polynomial (see \Cref{prop:nbrOpen-sizeR}).

We now study the size of $\SHatf$ which turns out to be exponential.

\begin{proposition}\label{prop:sizeSHatf}
    The size of $\SHatf$ is exponential in $|Q|, |\Sigma|$ and $\lengthCe$.
\end{proposition}
\begin{proof}
    We first study the growth of $\separatorsMapping$ after one round of the learning algorithm. We first make the current table $\obsTable$ closed and consistent, and then handle a counterexample to a (partial) equivalence query. For the resulting table $\obsTable[\ell']'$, we have by \Cref{lem:resolveCe,prop:NbrInc-sizeS} that there is some constant $c$ such that:
    \begin{align}
        \sizeOfSet{\separatorsMapping'} \leq c\RfSize^{4 \RfSize^2 + 4} |\separatorsMapping|. \label{eq:Shat}
    \end{align}
    
    Recall that the number of rounds is bounded by the total number of partial
    equivalence queries, which is in $\complexity(\lengthCe^3)$ by \Cref{prop:nbrEquivQueries}. Since initially
    $\separatorsMapping = \{\emptyword\}$, we get by~\eqref{eq:Shat} that
    $\sizeOfSet{\SHatf}$ is bounded by ${(c\RfSize^{4 \RfSize^2 +
    4})}^{\complexity(\lengthCe^3)}$. Since $\RfSize = \sizeOfSet{\Rf \cup \Rf\Sigma}$ is polynomial
  in $|Q|, |\Sigma|$ and $\lengthCe$ by \Cref{prop:nbrOpen-sizeR},
  it follows that $\sizeOfSet{\SHatf}$ is exponential in $|Q|, |\Sigma|$ and
  $\lengthCe$.
\qed\end{proof}

We get the next corollary about the total number of membership and counter value queries asked by the learner during the execution of the learning algorithm.

\begin{corollary}\label{cor:nbrQueries}
    In the learning algorithm, the number of membership queries and counter value queries is exponential in $|Q|, |\Sigma|$ and $\lengthCe$.
\end{corollary}
\begin{proof}
    We first study the number of membership and counter value queries after one round of the learning algorithm. Recall that both numbers are polynomial in the size of the current table (and thus in $|\Rf \cup \Rf\Sigma|,|\SHatf|$) after resolving an openness, a \inconsistency, a \inconsistencyBot, or handling a counterexample (see \Cref{lem:resolveOpenness2,lem:resolveSigmaInc2,lem:resolveBotInc2,lem:resolveCe}). Recall that making the table closed and consistent requires resolving at most $\complexity(|Q|\lengthCe)$ openness cases (by \Cref{prop:nbrOpen-sizeR}), at most $\RfSize^2$ decreases of \approxSet s and at most $\complexity(\RfSize^{4 \RfSize^2 + 3} |\separatorsMapping|)$ eliminations of mismatches (by \Cref{prop:NbrInc-sizeS}), where $\RfSize = |\Rf \cup \Rf\Sigma| = \complexity(|Q| |\Sigma| \lengthCe^4)$ by \Cref{prop:nbrOpen-sizeR}. Hence, we get after one round a number of membership and counter value queries that is exponential in $|Q|, |\Sigma|$ and $\lengthCe$. 
    
    Second, as the number of rounds is in $\complexity(\lengthCe^3)$, the total number of membership and counter value queries performed during the learning algorithm is again exponential in $|Q|, |\Sigma|$ and $\lengthCe$. 
\qed\end{proof}

\subsection{Proof of \Cref{thm:main}}

We are now ready to prove our main theorem establishing the complexity of the learning algorithm (\Cref{thm:main}). We need a last lemma studying the complexity of the basic operations carried out by the learner during this algorithm. By \emph{basic operations}, given an observation table $\obsTable$, we mean:
\begin{itemize}
    \item to check whether \(u \in \Approx(v)\),
    \item to check whether \(u \in \prefTable\),
    \item to check whether $\obsTable$ is closed,
    \item to check whether $\obsTable$ is \consistent,
    \item to check whether $\obsTable$ is \prefixed,
    \item to construct the automaton $\autTable$ when $\obsTable$ is closed, $\Sigma$- and \prefixed,
    \item to compute the periodic descriptions $\alpha_1, \dotsc, \alpha_n$ in $\autTable$ and construct the ROCAs \(\automaton_{\alpha_1}, \dotsc, \automaton_{\alpha_n}\),
    \item to modify \(\autTable\) to see it as an ROCA\@.
\end{itemize}

\begin{lemma}\label{lem:basicOpLearner}
    Given an observation table $\obsTable$, each basic operation of the learner is in time polynomial in $\sizeOfSet{\obsTable}$. 
\end{lemma}
\begin{proof}
    Since the actual complexity of the basic operations greatly depends on the implementation details, we give here a very naive complexity.
    For instance, to check whether \(u \in \Approx(v)\) with \(u, v \in \representatives \cup \representatives \Sigma\), we test for every \(s \in \separatorsCounters\) if \(\obsMapping(us) = \obsMapping(vs)\) and \(\obsCounters(us) \neq \bot \land \obsCounters(vs) \neq \bot \implies \obsCounters(us) = \obsCounters(vs)\).
    This operation is therefore in \(\complexity(\lengthOf{\separatorsCounters})\) which is polynomial in $\sizeOfSet{\obsTable}$.

    Let \(u \in \Sigma^*\).
    Checking whether \(u \in \prefTable\) is equivalent to checking if there exist \(v \in \representatives \cup \representatives \Sigma\) and \(s \in \separatorsMapping\) such that \(u\) is a prefix of \(vs\) and \(\obsMapping(vs) = 1\), leading to a polynomial complexity
    
    Checking whether the table is closed (i.e., \(\forall u \in \representatives \Sigma, \Approx(u) \cap \representatives \neq \emptyset\)) can be done in time polynomial in $\sizeOfSet{\obsTable}$, as well as checking whether the table is \consistent (i.e., \(\forall ua \in \representatives \Sigma\), \(ua \in \bigcap_{v \in \Approx(u) \cap \representatives} \Approx(va)\)) and is \prefixed (i.e., \(\forall u, v \in \representatives \cup \representatives \Sigma\) such that \(u \in \Approx(v)\), \( \forall s \in \separatorsCounters,~ \obsCounters(us) \neq \bot \iff \obsCounters(vs) \neq \bot\)).

    Let \(n\) be the index of \(\equivTable\).
    To construct \(\autTable\), we need to select one representative $u$ by equivalence class of \(\equivTable\) and to define the transitions (which implies finding the class of \(ua\)).
    Since \(n\) is bounded by \(\lengthOf{\representatives}\), the construction of $\autTable$ has a complexity polynomial in $\sizeOfSet{\obsTable}$.

    By~\cite{neider2010learning}, we know that finding the periodic descriptions in \(\autTable\) is in polynomial time. The construction of an ROCA from a description is also in polynomial time, by \Cref{prop:DescriptionAut}.

    Finally, to see \(\autTable\) as an ROCA, just modify it with a function $\deltaZero$ that keeps the counter always equal to $0$.

    Thus, every basic operation for the learner is in polynomial time in the size of the table.
\qed\end{proof}

We conclude with the proof of our main theorem.

\begin{proof}[of \Cref{thm:main}]
    The number of different kinds of queries is established in \Cref{prop:nbrEquivQueries,cor:nbrQueries}. By \Cref{prop:nbrOpen-sizeR,prop:NbrInc-sizeS}, the space used by the learning algorithm is mainly the space used to store the observation table which is polynomial (resp.\ exponential) in $|Q|, |\Sigma|, \lengthCe$ for $\Rf \cup \Rf\Sigma$ (resp.\ for $\SHatf$). Finally, the algorithm runs in time exponential in $|Q|, |\Sigma|, \lengthCe$. Indeed each basic operation of the learner is in time polynomial in size of the table by \Cref{lem:basicOpLearner}, the algorithm is executed in at most $\complexity(\lengthCe^3)$ rounds by \Cref{prop:nbrEquivQueries}, and each round results in at most $\complexity(|Q|\lengthCe)$ openness cases, $\RfSize^2$ decreases of \approxSet s, $\complexity(\RfSize^{4 \RfSize^2 + 3} |\separatorsMapping|)$ eliminations of mismatches, and has to handle one counterexample by \Cref{prop:nbrOpen-sizeR,prop:NbrInc-sizeS} (where $\RfSize = \sizeOfSet{\Rf \cup \Rf\Sigma}$). 
\qed\end{proof}

\section{Complete example}\label{example}

We now give an example of an execution of our learning algorithm.
Let \(\automaton\) be the ROCA of \Cref{fig:example:roca} and assume the teacher
possesses \(\automaton\).

We initialize the observation table with \(\emptyword\) as the unique element
of both \(\representatives\) and \(\separatorsCounters = \separatorsMapping\).
Moreover, we set the counter limit to be \(\ell = 0\), i.e., we start with
\(\observationTable[0]\).
We then fill the table with membership and counter value queries.
\Cref{fig:roca:roca:learning:example:tables:zero:init} gives the resulting table.
Observe that \(\obsCounters(a) = \bot\), as \(a\) is not in
\(\prefixes{\observationTable[0]}\).

We have the following approximation sets:
\begin{align*}
  \Approx(\emptyword) = \Approx(a) &= \{\emptyword, a\}
  &\text{and}&&
  \Approx(b) &= \{b\}.
\end{align*}
Observe that this table is not closed, as
\(\Approx(b) \cap \representatives = \emptyset\).
We thus add \(b\) as a new representative, i.e.,
\(\representatives\) is now \(\{\emptyword, b\}\), and obtain the table of
\Cref{fig:roca:roca:learning:example:tables:zero:closed} with
\begin{align*}
  \Approx(\emptyword) = \Approx(a) &= \{\emptyword, a\}
  &\text{and}&&
  \Approx(b) &= \{b\}.
\end{align*}

\begin{figure}
  \centering
  \begin{subfigure}{.45\textwidth}
    \centering
    \begin{tabular}{@{} l | c @{}} 
                       & \(\emptyword\) \\
        \specialrule{\cmidrulewidth}{0pt}{0pt}
        \(\emptyword\) & \(0, 0\)     \\
        \specialrule{\cmidrulewidth}{0pt}{0pt}
        \(a\)          & \(0, \bot\)  \\
        \(b\)          & \(1, 0\)    \\
    \end{tabular}
    \caption{Initial table.}%
    \label{fig:roca:roca:learning:example:tables:zero:init}
  \end{subfigure}
  \begin{subfigure}{.45\textwidth}
    \centering
    \begin{tabular}{@{} l | c @{}} 
                       & \(\emptyword\) \\
        \specialrule{\cmidrulewidth}{0pt}{0pt}
        \(\emptyword\) & \(0, 0\)     \\
        \(b\)          & \(1, 0\)    \\
        \specialrule{\cmidrulewidth}{0pt}{0pt}
        \(a\)          & \(0, \bot\)  \\
        \(ba\)         & \(1, 0\)    \\
        \(bb\)         & \(1, 0\)    \\
    \end{tabular}
    \caption{After resolving the \(b\)-openness.}%
    \label{fig:roca:roca:learning:example:tables:zero:closed}
  \end{subfigure}

  \begin{subfigure}{.45\textwidth}
    \centering
    \begin{tabular}{@{} l | c c @{}} 
                       & \(\emptyword\) & \(b\)         \\
        \specialrule{\cmidrulewidth}{0pt}{0pt}
        \(\emptyword\) & \(0, 0\)     & \(1, 0\)   \\
        \(b\)          & \(1, 0\)    & \(1, 0\)   \\
        \specialrule{\cmidrulewidth}{0pt}{0pt}
        \(a\)          & \(0, \bot\)  & \(0, \bot\) \\
        \(ba\)         & \(1, 0\)    & \(1, 0\)   \\
        \(bb\)         & \(1, 0\)    & \(1, 0\)   \\
    \end{tabular}
    \caption{After resolving the
    \((\emptyword, a, \emptyword)\)-\inconsistencyBot.}%
    \label{fig:roca:roca:learning:example:tables:zero:prefixed}
  \end{subfigure}
  \begin{subfigure}{.45\textwidth}
    \centering
    \begin{tabular}{@{} l | c c @{}} 
                       & \(\emptyword\) & \(b\)         \\
        \specialrule{\cmidrulewidth}{0pt}{0pt}
        \(\emptyword\) & \(0, 0\)     & \(1, 0\)   \\
        \(b\)          & \(1, 0\)    & \(1, 0\)   \\
        \(a\)          & \(0, \bot\)  & \(0, \bot\) \\
        \specialrule{\cmidrulewidth}{0pt}{0pt}
        \(ba\)         & \(1, 0\)    & \(1, 0\)   \\
        \(bb\)         & \(1, 0\)    & \(1, 0\)   \\
        \(aa\)         & \(0, \bot\)  & \(0, \bot\) \\
        \(ab\)         & \(0, \bot\)  & \(0, \bot\) \\
    \end{tabular}
    \caption{After resolving the \(a\)-openness.}%
    \label{fig:roca:roca:learning:example:tables:zero:final}
  \end{subfigure}
  \caption{Observation tables up to zero for \Cref{example}.}%
  \label{fig:roca:roca:learning:example:tables:zero}
\end{figure}

The table is \notPrefixed, as \(a \in \Approx(\emptyword)\) but
\(\obsCounters(\emptyword \cdot \emptyword) = 0\) and
\(\obsCounters(a \cdot \emptyword) = \bot\).
That is, the counter values for the column \(\emptyword\) are different.
To match with the notations used above when explaining how to make a table
\prefixed, let \(u = \emptyword, v = a, s = \emptyword\),
\(u' = b\), and \(s' = \emptyword\).
We have that \(u \cdot s = \emptyword\) is a prefix of \(u' \cdot s' = b\),
and \(\obsMapping(b) = 1\).
Let \(s'' = b\) (observe that \(u \cdot s'' = b = u' \cdot s'\), as required).
Since \(u\) is a proper prefix of \(u'\), we ask a membership query over
\(v \cdot s'' = a \cdot b\), which returns true, indicating
that \(a \cdot b \in \languageOf{\automaton}\).
As we need to check whether \(a \cdot b\) is in the bounded language up to
zero, we
also ask a counter value over \(a \cdot b\), which returns 1.
Hence, \(a \cdot b\) is not in the bounded language.
As explained in \Cref{subsec:makingClosed}, we add all the suffixes
of \(b\) to both \(\separatorsCounters\) and \(\separatorsMapping\).
We obtain the table given in
\Cref{fig:roca:roca:learning:example:tables:zero:prefixed} with
\begin{align*}
  \Approx(\emptyword) &= \{\emptyword\},
  \\
  \Approx(b) = \Approx(ba) = \Approx(bb) &= \{b, ba, bb\},
  \shortintertext{and}
  \Approx(a) &= \{a\}.
\end{align*}

The table is opened, due to \(a\).
As before, we then add \(a\) as a new representative, which results in the table
of \Cref{fig:roca:roca:learning:example:tables:zero:final} with
\begin{align*}
  \Approx(\emptyword) &= \{\emptyword\},
  \\
  \Approx(b) = \Approx(ba) = \Approx(bb) &= \{b, ba, bb\},
  \shortintertext{and}
  \Approx(a) = \Approx(aa) = \Approx(ab) &= \{a, aa, ab\}.
\end{align*}
Clearly, the table is closed and \prefixed.
It is also not hard to see that it is \consistent, as the intersection of each
approximation set with \(\representatives\) is a singleton.
Hence, we can compute an equivalence relation \(\equivTable[0]\) from the table
such that:
\begin{align*}
  \equivalenceClass{\emptyword}_{\equivTable[0]} &= \{\emptyword\}
  \\
  \equivalenceClass{b}_{\equivTable[0]} &= \{b, ba, bb\}
  \\
  \equivalenceClass{a}_{\equivTable[0]} &= \{a, aa, ab\}.
\end{align*}

\begin{figure}
  \centering
  \begin{subfigure}{.45\textwidth}
    \centering
    \begin{tikzpicture}[
        automaton,
        node distance = 40pt,
        ]
        \node [state, initial]                  (eps) {\(\emptyword\)};
        \node [state, accepting, below=of eps]  (b)   {\(b\)};
        \node [state, right=of eps]             (a)   {\(a\)};
        \path
            (eps) edge              node [right] {\(b\)}    (b)
                edge              node [above] {\(a\)}    (a)
            (b)   edge [loop below] node {\(a, b\)} (b)
            (a)   edge [loop above] node {\(a, b\)} (a)
        ;
    \end{tikzpicture}
    \caption{The hypothesis DFA.}%
    \label{fig:roca:roca:learning:example:tables:zero:hypo:dfa}
  \end{subfigure}
  \hfill
  \begin{subfigure}{.45\textwidth}
    \centering
    \begin{tikzpicture}[
        automaton,
        node distance = 40pt and 60pt,
        ]
        \node [state, initial]                  (eps) {\(b\)};
        \node [state, accepting, below=of eps]  (b)   {\(b\)};
        \node [state, right=of eps]             (a)   {\(a\)};
        \path[deltaZero]
            (eps) edge              node [right] {\(b, =0, 0\)}  (b)
                edge              node [above] {\(a, =0, 0\)}  (a)
            (b)   edge [loop below] node [below, align=center] {
                                                \(a, =0, 0\)\\
                                                \(b, =0, 0\)}    (b)
            (a)   edge [loop above] node [above, align=center] {
                                                \(a, =0, 0\)\\
                                                \(b, =0, 0\)}    (a)
        ;
    \end{tikzpicture}
    \caption{The hypothesis ROCA.}%
    \label{fig:roca:roca:learning:example:tables:zero:hypo:roca}
  \end{subfigure}
  \caption{The hypotheses DFA and ROCA constructed from the table of
  \Cref{fig:roca:roca:learning:example:tables:zero:final}.}%
  \label{fig:roca:roca:learning:example:tables:zero:hypo}
\end{figure}

We then construct the hypothesis DFA \(\hypothesis_0\), which is given in
\Cref{fig:roca:roca:learning:example:tables:zero:hypo:dfa}.
By a partial equivalence query over \(\hypothesis_{0}\), we confirm that
\(\hypothesis_0\) accepts the bounded language of \(\automaton\) up to 0,
i.e., learning the bounded behavior graph up to 0 is done.
Since we cannot extract any ultimately periodic description from
\(\hypothesis_0\), we immediately convert it into an ROCA
\(\hypothesis_{\text{ROCA}}\) that never increases its counter.
This ROCA is given in \Cref{fig:roca:roca:learning:example:tables:zero:hypo:roca}.
We ask an equivalence query over the ROCA\@.
Assume that the teacher returns the word \(w = aabaa\).
By performing counter value queries on every prefix of the counterexample, we
increase the counter limit to 2, i.e., we now work with \(\observationTable[2]\).
Furthermore, we add all the prefixes of \(w\) as new representatives.

The resulting table is given in
\Cref{fig:roca:roca:learning:example:tables:two:init}.
One can check that \(aa\) is in \(\Approx(aab)\) but \(\obsCounters(aa \cdot b) = 2\)
and \(\obsCounters(aab \cdot b) = \bot\).
That is, we have a \((aa, aab, b)\)-\inconsistencyBot.
As we did earlier, let \(u = aa, v = aab, s = b, u' = aabaa\), and
\(s' = \emptyword\).
We have \(u \cdot s = aab\) is a prefix of \(u' \cdot s' = aabaa\).
Let \(s'' = baa\) (observe that \(u \cdot s'' = aabaa = u' \cdot s'\), as required).
Since \(u\) is a proper prefix of \(u'\), we ask a membership query over
\(v \cdot s'' = aab \cdot baa\), which returns true, indicating
that \(aabbaa \in \languageOf{\automaton}\).
As we need to check whether \(aabbaa\) is in the bounded language up to 2,
we perform
a counter value query on each prefix of \(aabbaa\) and observe that the height
of the word is 2.
Hence, \(aabbaa\) is in the bounded language.
We thus add all suffixes of \(s'' = baa\) in \(\separatorsMapping\).
\Cref{fig:roca:roca:learning:example:tables:two:firstPrefixed} gives the obtained
table.
Observe that adding the new columns changed the prefix of the language encoded
in the observation table and some values for \(\obsCounters\) changed.
For instance, \(\obsCounters(a \cdot b)\) is now 1, instead of \(\bot\).
Furthermore, the \inconsistencyBot is indeed resolved.

\begin{figure}
  \centering
  \begin{subfigure}{.2\textwidth}
    \centering
    \begin{tabular}{@{} l | c c @{}} 
                       & \(\emptyword\) & \(b\)         \\
        \specialrule{\cmidrulewidth}{0pt}{0pt}
        \(\emptyword\) & \(0, 0\)     & \(1, 0\)   \\
        \(b\)          & \(1, 0\)    & \(1, 0\)   \\
        \(a\)          & \(0, 1\)     & \(0, \bot\) \\
        \(aa\)         & \(0, 2\)     & \(0, 2\)    \\
        \(aab\)        & \(0, 2\)     & \(0, \bot\) \\
        \(aaba\)       & \(0, 1\)     & \(0, \bot\) \\
        \(aabaa\)      & \(1, 0\)    & \(1, 0\)   \\
        \specialrule{\cmidrulewidth}{0pt}{0pt}
        \(ba\)         & \(1, 0\)    & \(1, 0\)   \\
        \(bb\)         & \(1, 0\)    & \(1, 0\)   \\
        \(ab\)         & \(0, \bot\)  & \(0, \bot\) \\
        \(aaa\)        & \(0, \bot\)  & \(0, \bot\) \\
        \(aabb\)       & \(0, \bot\)  & \(0, \bot\) \\
        \(aabab\)      & \(0, \bot\)  & \(0, \bot\) \\
        \(aabaaa\)     & \(1, 0\)    & \(1, 0\)   \\
        \(aabaab\)     & \(1, 0\)    & \(1, 0\)   \\
    \end{tabular}
    \caption{Initial table.}%
    \label{fig:roca:roca:learning:example:tables:two:init}
  \end{subfigure}
  \hfill
  \begin{subfigure}{.35\textwidth}
    \centering
    \begin{tabular}{@{} l | c c | c c c @{}} 
                        & \(\emptyword\) & \(b\)         & \(a\) & \(aa\) & \(baa\) \\
        \specialrule{\cmidrulewidth}{0pt}{0pt}
        \(\emptyword\) & \(0, 0\)     & \(1, 0\)   & 0   & 0    & 1    \\
        \(b\)          & \(1, 0\)    & \(1, 0\)   & 1  & 1   & 1    \\
        \(a\)          & \(0, 1\)     & \(0, 1\)    & 0   & 0    & 1    \\
        \(aa\)         & \(0, 2\)     & \(0, 2\)    & 0   & 0    & 1    \\
        \(aab\)        & \(0, 2\)     & \(0, 2\)    & 0   & 1   & 1    \\
        \(aaba\)       & \(0, 1\)     & \(0, 1\)    & 1  & 1   & 1    \\
        \(aabaa\)      & \(1, 0\)    & \(1, 0\)   & 1  & 1   & 1    \\
        \specialrule{\cmidrulewidth}{0pt}{0pt}
        \(ba\)         & \(1, 0\)    & \(1, 0\)   & 1  & 1   & 1    \\
        \(bb\)         & \(1, 0\)    & \(1, 0\)   & 1  & 1   & 1    \\
        \(ab\)         & \(0, 1\)     & \(0, 1\)    & 1  & 1   & 1    \\
        \(aaa\)        & \(0, \bot\)  & \(0, \bot\) & 0   & 0    & 0     \\
        \(aabb\)       & \(0, 2\)     & \(0, 2\)    & 0   & 1   & 1    \\
        \(aabab\)      & \(0, 1\)     & \(0, 1\)    & 1  & 1   & 1    \\
        \(aabaaa\)     & \(1, 0\)    & \(1, 0\)   & 1  & 1   & 1    \\
        \(aabaab\)     & \(1, 0\)    & \(1, 0\)   & 1  & 1   & 1    \\
    \end{tabular}
    \caption{After resolving the \((aa, aab, b)\)-\inconsistencyBot.}%
    \label{fig:roca:roca:learning:example:tables:two:firstPrefixed}
  \end{subfigure}
  \hfill
  \begin{subfigure}{.4\textwidth}
    \centering
    \begin{tabular}{@{} l | c c c | c c @{}} 
                            & \(\emptyword\) & \(b\)         & \(a\)         & \(aa\) & \(baa\) \\
            \specialrule{\cmidrulewidth}{0pt}{0pt}
            \(\emptyword\) & \(0, 0\)     & \(1, 0\)   & \(0, 1\)    & 0    & 1    \\
            \(b\)          & \(1, 0\)    & \(1, 0\)   & \(1, 0\)   & 1   & 1    \\
            \(a\)          & \(0, 1\)     & \(0, 1\)    & \(0, 2\)    & 0    & 1    \\
            \(aa\)         & \(0, 2\)     & \(0, 2\)    & \(0, \bot\) & 0    & 1    \\
            \(aab\)        & \(0, 2\)     & \(0, 2\)    & \(0, 1\)    & 1   & 1    \\
            \(aaba\)       & \(0, 1\)     & \(0, 1\)    & \(1, 0\)   & 1   & 1    \\
            \(aabaa\)      & \(1, 0\)    & \(1, 0\)   & \(1, 0\)   & 1   & 1    \\
            \specialrule{\cmidrulewidth}{0pt}{0pt}
            \(ba\)         & \(1, 0\)    & \(1, 0\)   & \(1, 0\)   & 1   & 1    \\
            \(bb\)         & \(1, 0\)    & \(1, 0\)   & \(1, 0\)   & 1   & 1    \\
            \(ab\)         & \(0, 1\)     & \(0, 1\)    & \(1, 0\)   & 1   & 1    \\
            \(aaa\)        & \(0, \bot\)  & \(0, \bot\) & \(0, \bot\) & 0    & 0     \\
            \(aabb\)       & \(0, 2\)     & \(0, 2\)    & \(0, 1\)    & 1   & 1    \\
            \(aabab\)      & \(0, 1\)     & \(0, 1\)    & \(1, 0\)   & 1   & 1    \\
            \(aabaaa\)     & \(1, 0\)    & \(1, 0\)   & \(1, 0\)   & 1   & 1    \\
            \(aabaab\)     & \(1, 0\)    & \(1, 0\)   & \(1, 0\)   & 1   & 1    \\
        \end{tabular}
    \caption{After resolving the \((a, aaba, a)\)-\inconsistency.}%
    \label{fig:roca:roca:learning:example:tables:two:firstInconsistent}
  \end{subfigure}

  \begin{subfigure}{.48\textwidth}
    \centering
    \begin{tabular}{@{} l | c c c c | c @{}} 
                        & \(\emptyword\) & \(b\)         & \(a\)         & \(aa\)        & \(baa\) \\
        \specialrule{\cmidrulewidth}{0pt}{0pt}
        \(\emptyword\) & \(0, 0\)     & \(1, 0\)   & \(0, 1\)    & \(0, 2\)    & 1    \\
        \(b\)          & \(1, 0\)    & \(1, 0\)   & \(1, 0\)   & \(1, 0\)   & 1    \\
        \(a\)          & \(0, 1\)     & \(0, 1\)    & \(0, 2\)    & \(0, \bot\) & 1    \\
        \(aa\)         & \(0, 2\)     & \(0, 2\)    & \(0, \bot\) & \(0, \bot\) & 1    \\
        \(aab\)        & \(0, 2\)     & \(0, 2\)    & \(0, 1\)    & \(1, 0\)   & 1    \\
        \(aaba\)       & \(0, 1\)     & \(0, 1\)    & \(1, 0\)   & \(1, 0\)   & 1    \\
        \(aabaa\)      & \(1, 0\)    & \(1, 0\)   & \(1, 0\)   & \(1, 0\)   & 1    \\
        \specialrule{\cmidrulewidth}{0pt}{0pt}
        \(ba\)         & \(1, 0\)    & \(1, 0\)   & \(1, 0\)   & \(1, 0\)   & 1    \\
        \(bb\)         & \(1, 0\)    & \(1, 0\)   & \(1, 0\)   & \(1, 0\)   & 1    \\
        \(ab\)         & \(0, 1\)     & \(0, 1\)    & \(1, 0\)   & \(1, 0\)   & 1    \\
        \(aaa\)        & \(0, \bot\)  & \(0, \bot\) & \(0, \bot\) & \(0, \bot\) & 0     \\
        \(aabb\)       & \(0, 2\)     & \(0, 2\)    & \(0, 1\)    & \(1, 0\)   & 1    \\
        \(aabab\)      & \(0, 1\)     & \(0, 1\)    & \(1, 0\)   & \(1, 0\)   & 1    \\
        \(aabaaa\)     & \(1, 0\)    & \(1, 0\)   & \(1, 0\)   & \(1, 0\)   & 1    \\
        \(aabaab\)     & \(1, 0\)    & \(1, 0\)   & \(1, 0\)   & \(1, 0\)   & 1    \\
    \end{tabular}
    \caption{After resolving the \((aa, aab, a)\)-\inconsistency.}%
    \label{fig:roca:roca:learning:example:tables:two:secondInconsistent}
  \end{subfigure}
  \hfill
  \begin{subfigure}{.48\textwidth}
    \centering
    \begin{tabular}{@{} l | c c c c c c @{}} 
                        & \(\emptyword\) & \(b\)         & \(a\)         & \(aa\)        & \(baa\)       \\
        \specialrule{\cmidrulewidth}{0pt}{0pt}
        \(\emptyword\) & \(0, 0\)     & \(1, 0\)   & \(0, 1\)    & \(0, 2\)    & \(1, 0\)   \\
        \(b\)          & \(1, 0\)    & \(1, 0\)   & \(1, 0\)   & \(1, 0\)   & \(1, 0\)   \\
        \(a\)          & \(0, 1\)     & \(0, 1\)    & \(0, 2\)    & \(0, \bot\) & \(1, 0\)   \\
        \(aa\)         & \(0, 2\)     & \(0, 2\)    & \(0, \bot\) & \(0, \bot\) & \(1, 0\)   \\
        \(aab\)        & \(0, 2\)     & \(0, 2\)    & \(0, 1\)    & \(1, 0\)   & \(1, 0\)   \\
        \(aaba\)       & \(0, 1\)     & \(0, 1\)    & \(1, 0\)   & \(1, 0\)   & \(1, 0\)   \\
        \(aabaa\)      & \(1, 0\)    & \(1, 0\)   & \(1, 0\)   & \(1, 0\)   & \(1, 0\)   \\
        \(aaa\)        & \(0, \bot\)  & \(0, \bot\) & \(0, \bot\) & \(0, \bot\) & \(0, \bot\) \\
        \specialrule{\cmidrulewidth}{0pt}{0pt}
        \(ba\)         & \(1, 0\)    & \(1, 0\)   & \(1, 0\)   & \(1, 0\)   & \(1, 0\)   \\
        \(bb\)         & \(1, 0\)    & \(1, 0\)   & \(1, 0\)   & \(1, 0\)   & \(1, 0\)   \\
        \(ab\)         & \(0, 1\)     & \(0, 1\)    & \(1, 0\)   & \(1, 0\)   & \(1, 0\)   \\
        \(aabb\)       & \(0, 2\)     & \(0, 2\)    & \(0, 1\)    & \(1, 0\)   & \(1, 0\)   \\
        \(aabab\)      & \(0, 1\)     & \(0, 1\)    & \(1, 0\)   & \(1, 0\)   & \(1, 0\)   \\
        \(aabaaa\)     & \(1, 0\)    & \(1, 0\)   & \(1, 0\)   & \(1, 0\)   & \(1, 0\)   \\
        \(aabaab\)     & \(1, 0\)    & \(1, 0\)   & \(1, 0\)   & \(1, 0\)   & \(1 0\)   \\
        \(aaaa\)        & \(0, \bot\)  & \(0, \bot\) & \(0, \bot\) & \(0, \bot\) & \(0, \bot\) \\
        \(aaab\)        & \(0, \bot\)  & \(0, \bot\) & \(0, \bot\) & \(0, \bot\) & \(0, \bot\) \\
    \end{tabular}
    \caption{After resolving
    the \((a, aaa, \emptyword)\)-\inconsistencyBot and the \(aaa\)-openness.}%
    \label{fig:roca:roca:learning:example:tables:two:final}
  \end{subfigure}
  \caption{Observation tables up to two for
  \Cref{example}.}%
  \label{fig:roca:roca:learning:example:tables:two}
\end{figure}

Since \(aaba \in \Approx(a)\) but \(aaba \cdot a \notin \Approx(a \cdot a)\),
we have a \((a, aaba, a)\)-\inconsistency.
We also have that
\(\obsMapping(a \cdot \emptyword) \neq \obsMapping(aaba \cdot \emptyword)\).
Hence, the \inconsistency us resolved by adding all the suffixes of
\(a \cdot \emptyword\) to both \(\separatorsCounters\) and \(\separatorsMapping\).
\Cref{fig:roca:roca:learning:example:tables:two:firstInconsistent} gives the
resulting table, which has a \((aa, aab, a)\)-\inconsistency.
As \(\obsMapping(aa \cdot a) \neq \obsMapping(aab \cdot a)\), we add
all suffixes of \(a \cdot a\) to both \(\separatorsCounters\) and
\(\separatorsMapping\) and obtain the table of
\Cref{fig:roca:roca:learning:example:tables:two:secondInconsistent}.

We have a \((a, aaa, \emptyword)\)-\inconsistencyBot, which can be resolved by
adding all suffixes of \(baa\) to \(\separatorsCounters\).
Furthermore, \(\Approx(aaa) \cap \representatives = \emptyset\), i.e., we also have
a \(aaa\)-openness, which is resolved by adding \(aaa\) as a new
representative.
\Cref{fig:roca:roca:learning:example:tables:two:final} gives the
corresponding table.

One can check that this table is closed, \(\Sigma\)- and \prefixed.
Hence, we can define the following equivalence relation \(\equivTable[2]\):
\begin{align*}
  \equivalenceClass{\emptyword}_{\equivTable[2]} &= \{\emptyword\}
  \\
  \equivalenceClass{b}_{\equivTable[2]} &= \{b, aabaa, ba, bb, aabaaa, aabaab\}
  \\
  \equivalenceClass{a}_{\equivTable[2]} &= \{a\}
  \\
  \equivalenceClass{aa}_{\equivTable[2]} &= \{aa\}
  \\
  \equivalenceClass{aab}_{\equivTable[2]} &= \{aab, aabb\}
  \\
  \equivalenceClass{aaba}_{\equivTable[2]} &= \{aaba, ab, aabab\}
  \\
  \equivalenceClass{aaa}_{\equivTable[2]} &= \{aaa, aaaa, aaaab\}
\end{align*}
We can construct a 7-state DFA \(\hypothesis_{2}\), which is given in
\Cref{fig:roca:roca:learning:example:tables:two:hypo:roca}.
By a partial equivalence query over \(\hypothesis_{2}\), we confirm that
\(\hypothesis_2\) accepts the bounded language of \(\automaton\) up to 2,
i.e., learning the bounded behavior graph up to 2 is done.
From that DFA, we obtain an ultimately periodic description.%
\footnote{An algorithm to compute an ultimately periodic description is
given in~\cite{neider2010learning}.
We immediately give the computed description here.}
First, we enumerate the states:
\begin{align*}
  \nu_0(\equivalenceClass{\emptyword}_{\equivTable[2]})
  = \nu_{1}(\equivalenceClass{a}_{\equivTable[2]})
  = \nu_{2}(\equivalenceClass{aa}_{\equivTable[2]})
  = \nu_{3}(\equivalenceClass{aaa}_{\equivTable[2]})
  &= 1
  \\
  \nu_{0}(\equivalenceClass{b}_{\equivTable[2]})
  = \nu_{1}(\equivalenceClass{aaba}_{\equivTable[2]})
  = \nu_{2}(\equivalenceClass{aab}_{\equivTable[2]})
  &= 2.
\end{align*}
This allows to abstract the transitions as:
\begin{align*}
  \tau_{0}(1, a)
  = \tau_{1}(1, a)
  = \tau_{2}(1, a)
  &= (1, +1)
  \\
  \tau_{0}(1, b)
  = \tau_{1}(1, b)
  = \tau_{2}(1, b)
  &= (2, 0)
  \\
  \tau_{0}(2, a)
  = \tau_{0}(2, b)
  = \tau_{1}(2, b)
  = \tau_{2}(2, b)
  &= (2, 0)
  \\
  \tau_{1}(2, a)
  = \tau_{2}(2, a)
  &= (2, -1).
\end{align*}
Then, we construct an ROCA \(\hypothesis\), as explained in
\Cref{def:observation_table:automaton},
which is given in \Cref{fig:roca:roca:learning:example:tables:two:hypo:roca}.
An equivalence query over \(\hypothesis\) returns true, meaning that
\(\languageOf{\hypothesis} = \languageOf{\automaton}\), and we are done.

\begin{figure}
  \centering
  \begin{subfigure}{\textwidth}
    \centering
  \begin{tikzpicture}[
  automaton,
]
  \node [state, initial]                  (eps) {\(\emptyword\)};
  \node [state, accepting, below=of eps]  (b)   {\(b\)};
  \node [state, right=of eps]             (a)   {\(a\)};
  \node [state, right=of a]               (aa)  {\(aa\)};
  \node [state, right=of aa]              (aaa) {\(aaa\)};
  \node [state, right=of b]               (aaba){\(aaba\)};
  \node [state, right=of aaba]            (aab) {\(aab\)};
  \path
    (eps) edge              node [left]   {\(b\)}     (b)
          edge              node [above]  {\(a\)}     (a)
    (b)   edge [loop below] node [below]  {\(a, b\)}  (b)
    (a)   edge              node [above]  {\(a\)}     (aa)
          edge              node [left]   {\(b\)}     (aaba)
    (aa)  edge              node [above]  {\(a\)}     (aaa)
          edge              node [left]   {\(b\)}     (aab)
    (aab) edge              node [above]  {\(a\)}     (aaba)
          edge [loop below] node [below]  {\(b\)}     ()
    (aaba)edge              node [above]  {\(a\)}     (b)
          edge [loop below] node [below]  {\(b\)}     ()
  ;
\end{tikzpicture}
    \caption{The DFA.}%
    \label{fig:roca:roca:learning:example:tables:two:hypo:dfa}
  \end{subfigure}

  \begin{subfigure}{\textwidth}
    \centering
    \begin{tikzpicture}[
    automaton,
    node distance = 80pt and 80pt,
    ]
    \node [state, initial]                  (eps) {\(b\)};
    \node [state, accepting, below=of eps]  (b)   {\(b\)};
    \node [state, right=of eps]             (a)   {\(a\)};
    \node [state, right=of a]               (aa)  {\(aa\)};
    \node [state, right=of b]               (aaba){\(aaba\)};
    \node [state, right=of aaba]            (aab) {\(aab\)};
    \path[deltaZero]
        (eps) edge              node [left]   {\(b, =0, 0\)} (b)
            edge              node [above]  {\(a, =0, 0\)} (a)
        (b)   edge [loop below] node [below, align=center]  {
                                            \(a, =0, 0\)\\
                                            \(b, =0, 0\)}  (b)
        (a)   edge              node [above]  {\(a, =0, +1\)}  (aa)
            edge [bend right=15] node [left]   {\(b, =0, 0\)} (aaba)
        (aaba)edge              node [above]  {\(a, =0, 0\)} (b)
            edge [in=-110, out=-140, loop] node [below left=-1pt and -12pt] {\(b, =0, 0\)} ()
    ;
    \path[deltaNotZero]
        (a)   edge [bend left=15] node [right]  {\(b, \neq 0, 0\)} (aaba)
        (aa)  edge [bend left=15] node [below]  {\(a, \neq 0, +1\)} (a)
            edge                node [right]  {\(b, \neq 0, 0\)} (aab)
        (aab) edge [loop below]   node [below]  {\(b, \neq 0, 0\)} (aab)
            edge                node [above]  {\(a, \neq 0, -1\)} (aaba)
        (aaba)edge [in=-70, out=-40, loop] node [below right=-1pt and -12pt, align=center] {
                                                \(a, \neq 0, -1\)\\
                                                \(b, \neq 0, 0\)} ()
    ;
    \draw [deltaNotZero, rounded corners = 10pt]
        let
        \p{s} = (a.north),
        \p{t} = (aa.north),
        \p{m} = ($(\p{s}) + (0, 0.5)$),
        in
        (\p{s}) -- (\x{s}, \y{m})
                -- (\x{t}, \y{m})
                node [above, midway]  {\(a, \neq 0, +1\)}
                -- (\p{t})
    ;
    \end{tikzpicture}
    \caption{The ROCA.}%
    \label{fig:roca:roca:learning:example:tables:two:hypo:roca}
  \end{subfigure}
  \caption{The hypothesis DFA constructed from the table of
  \Cref{fig:roca:roca:learning:example:tables:two:final}
  and a hypothesis ROCA.}
\end{figure}

\section{Experiments}\label{sec:experiments}

We evaluated our learning algorithm on two types of benchmarks.
The first uses randomly generated ROCAs, while the second focuses on a new approach to learn an ROCA that can efficiently verify whether a JSON document is valid against a given JSON schema. Notice that while there already exist several algorithms that infer
a JSON schema from a collection of JSON documents (see survey~\cite{DBLP:conf/edbt/BaaziziCGS19}), none are based on
learning techniques nor do they yield an automaton-based validation algorithm.

We first discuss some optimizations and implementation details, followed by the
random benchmarks.  We then introduce more precisely the context of our
JSON-based use case and give the results.

\subsection{Implementation details}

Let \(\observationTable = (\representatives, \separatorsCounters, \separatorsMapping, \obsMapping, \obsCounters)\) be an observation table up to counter limit \(\ell\).
We present here two optimizations regarding \(\observationTable\).
The first is an efficient way to store and manipulate \(\prefTable\), while the second improves the computations of the \approxSet s. Finally, We give our experimental framework.

\subsubsection{Using a tree to store observations.}\label{sec:impl:observation_tree}

The first optimization we consider is to store \(\prefTable\) in a tree structure, called the \emph{observation tree} and denoted by $\observationTree$.
Then, when a new word \(u\) is added in \(\prefTable\), it is sufficient to only update the path in \(\observationTree\) leading to \(u\). Finally, instead of directly storing the values for \(\obsMapping\) and \(\obsCounters\) in the table $\obsTable$, each cell of $\obsTable$ stores a pointer to a node in the tree $\observationTree$. We only give the general intuition here.

Hence, our goal is to store all the values that are needed in \(\observationTable\), i.e., the values for the functions \(\obsMapping\) and \(\obsCounters\) (on their respective domains).
If we store all the information directly in \(\observationTree\), having to recompute some \(\obsCounters\) due to a change in \(\prefTable\) is more immediate.
Indeed, let \(u\) be a word added in \(\prefTable\).
As only the prefixes of \(u\) can potentially have new values for \(\obsCounters\), it is sufficient to only iterate over the ancestors of the node storing the data for \(u\).
In other words, the tree structure reduces the runtime complexity (although the worst-case scenario does not change much).

When new representatives or separators have to be added in \(\observationTable\), nodes are added in \(\observationTree\) and the table is extended and updated with pointers to the tree.
Moreover, when a word \(u\) is added in \(\observationTree\), all the prefixes \(x\) of \(u\) must also be added, even if \(x \notin (\representatives \cup \representatives\Sigma) \separatorsMapping\).
As a consequence, the number of nodes in \(\observationTree\) will be greater than \(\sizeOfSet{(\representatives \cup \representatives \Sigma) \separatorsMapping}\) since all the prefixes are stored in the tree.
In other words, the tree increases the memory consumption of the algorithm but significantly improves the run time.

\subsubsection{Efficient computation of \(\Approx\).} 

Let us now discuss how to efficiently compute the \approxSet s.
Let \(u, v \in \representatives \cup \representatives\Sigma\).
If there exists \(s \in \separatorsCounters\) such that \(\obsMapping(us) \neq \obsMapping(vs)\), it follows that \(v \notin \Approx(u)\). Thus when computing \(\Approx(u)\), we can limit the elements in \(\representatives \cup \representatives \Sigma\) to consider to those words \(v\) that agree with \(u\) on the values of \(\obsMapping\).

More precisely, let us define a relation as follows.
Let \(u, v \in \representatives \cup \representatives\Sigma\).
We say that \(u \nerodeCongruence_{\obsTable} v\) if and only if \(\forall s \in \separatorsCounters, \obsMapping(us) = \obsMapping(vs)\). The \approxSet s can thus be defined as:

\begin{definition}[Optimization of \Cref{def:approx}]
    Let \(\observationTable = (\representatives, \separatorsCounters, \separatorsMapping, \obsCounters, \obsMapping)\) be an observation table up to \(\ell\).
    Let \(u \in \representatives \cup \representatives \Sigma\) and \(v \in \equivalenceClass{u}_{\nerodeCongruence_{\obsTable}}\).
    Then, \(v \in \Approx(u)\) if \(\forall s \in \separatorsCounters\), \(\obsCounters(us) \neq \bot\) and \(\obsCounters(vs) \neq \bot \implies \obsCounters(us) = \obsCounters(vs)\).
\end{definition}
With this definition, to compute \(\Approx(u)\), it is enough to iterate over the elements \(v \in \equivalenceClass{u}_{\nerodeCongruence_{\obsTable}}\) (instead of all $v \in \representatives \cup \representatives\Sigma$).

One can efficiently store each equivalence class of $\nerodeCongruence_{\obsTable}$ using a set, in order to obtain a polynomial time complexity for lookup, addition, and removal of an element, in the size of the table.\footnote{The complexity can be reduced to a constant time, in average, if we assume the set relies on hash tables.}
Moreover, we can also store the \(\Approx\) sets instead of recomputing them from scratch each time.

Notice that updating the classes of \(\nerodeCongruence_{\obsTable}\) can be easily done. Suppose a value \(\obsMapping(us)\) (with \(u \in \representatives \cup \representatives \Sigma\) and \(s \in \separatorsCounters\)) is modified, and let \(\equivalenceClass{v}_{\nerodeCongruence_{\obsTable}}\) be the class of \(u\) before the modification.
We have to remove \(u\) from the set for \(\equivalenceClass{v}_{\nerodeCongruence_{\obsTable}}\) and then add it to the set for \(\equivalenceClass{u}_{\nerodeCongruence_{\obsTable}}\).
This can be done in polynomial time.

Finally, we need to compute the \(\Approx\) sets only when the table has to be made closed and consistent (since new elements are added in \(\representatives \cup \representatives \Sigma\), \(\separatorsCounters\), or \(\separatorsMapping\)).
Notice that not all \approxSet s are modified and thus should not be recomputed. This happens for \(\Approx(u)\) when \(\equivalenceClass{u}_{\nerodeCongruence_{\obsTable}}\) is not modified and the values \(\obsCounters(us)\) remain unchanged.

\subsubsection{Experimental framework.}

The ROCAs and the learning algorithm were implemented by extending the well-known Java libraries \textsc{AutomataLib} and \textsc{LearnLib}~\cite{AutomataLib,LearnLib}.
These modifications can be consulted on our GitHub repositories~\cite{ModifiedAutomataLib,ModifiedLearnLib}, while the code for the benchmarks is available on~\cite{BenchmarksCode}.
Implementation specific details (such as the libraries) are given alongside the code.
The server used for the computations ran Debian 10 over Linux 5.4.73-1-pve with a 4-core Intel\textregistered{} Xeon\textregistered{} Silver 4214R Processor with 16.5M cache, and 64GB of RAM\@. 
Moreover, we used OpenJDK version 11.0.12.

\subsection{Random ROCAs}

We discuss our benchmarks based on randomly generated ROCAs.
We begin by explaining how this generation works. We then explain how we check the equivalence of two ROCAs.
We finally comment our results.

\subsubsection{Random generation of ROCAs.}

An ROCA with given size $n = \sizeOfSet{Q}$ is randomly generated such that:
\begin{itemize}
    \item \(\forall q \in Q, q\) has a probability \(0.5\) of being final,
    \item \(\forall q \in Q, \forall a \in \Sigma, \deltaNotZero(q, a) = (p, c)\) with \(p\) a random state in \(Q\) and \(c\) a random counter operation in \(\{-1, 0, +1\}\).
        We define \(\deltaZero(q, a) = (p, c)\) in a similar way except that \(c \in \{0, +1\}\).
\end{itemize}
All random draws are assumed to come from a uniform distribution.

Since this generation does not guarantee the produced ROCA has \(n\) reachable states, we produce 100 ROCAs and select the ROCA with a number of reachable states that is maximal.
However, note that it is still possible the resulting ROCA does not have \(n\) (co)-reachable states.

\subsubsection{Equivalence of two ROCAs.} 

The language equivalence problem of ROCAs is known to be
decidable and NL-complete~\cite{DBLP:journals/jcss/BohmGJ14}.
Unfortunately, the algorithm described
in~\cite{DBLP:journals/jcss/BohmGJ14} is difficult to implement.
Instead, we use an ``approximate'' equivalence oracle for our experiments.\footnote{The teacher might, with some small probability, answer with false positives but never with false negatives.}

Let \(\automaton\) and \(\automaton[B]\) be two ROCAs such that \(\automaton[B]\) is the learned ROCA from a periodic description with period \(k\).
The algorithm explores the configuration space of both ROCAs in parallel.
If, at some point, it reaches a pair of configurations such that one is accepting and the other not, then we have a counterexample.
However, to have an algorithm that eventually stops, we need to bound the counter value of the configurations to explore. Our approach is to first explore up to counter value \({\sizeOfSet{\automaton \times \automaton[B]}}^2\) (in view of~\cite[Proposition 18]{DBLP:journals/jcss/BohmGJ14} about shortest accepting runs in an ROCA). If no counterexample is found, we add \(k\) to the bound and, with probability \(0.5\), a new exploration is done up to the new bound.
We repeat this whole process until we find a counterexample or until the random draw forces us to stop.
This is summarized in \Cref{alg:equivalence_rocas}.

\begin{algorithm}
    \caption{Algorithm for checking the equivalence of two ROCAs}%
    \label{alg:equivalence_rocas}
    \begin{algorithmic}[1]
        \Require Let \(\automaton = (Q^{\automaton}, \Sigma, \deltaZero^{\automaton}, \deltaNotZero^{\automaton}, q_0^{\automaton}, F^{\automaton})\) and \(\automaton[B] = (Q^{\automaton[B]}, \Sigma, \deltaZero^{\automaton[B]}, \deltaNotZero^{\automaton[B]}, q_0^{\automaton[B]}, F^{\automaton[B]})\) be two ROCAs, and \(k\) be the period of the description that was used to construct \(\automaton[B]\)
        \Ensure True is returned if our approximation for \(\languageOf{\automaton} = \languageOf{\automaton[B]}\) holds, or a counterexample otherwise
        \Statex
        \State \(Q \gets\) a queue initialized with \((\emptyword, (q_0^{\automaton}, 0), (q_0^{\automaton[B]}, 0))\)
        \State \(\ell \gets {(\sizeOfSet{Q^{\automaton}} \cdot \sizeOfSet{Q^{\automaton[B]}})}^2\)
        \Repeat
            \While{\(Q\) is not empty}
                \State \((w, c_{\automaton}, c_{\automaton[B]}) \gets \) the next element in \(Q\) \Comment{We pop the tuple from \(Q\)}
                \ForAll{\(a \in \Sigma\)}
                    \State Let \((q, n) \in Q^{\automaton} \times \N\) such that \(c_{\automaton} \transition\limits^a_{\automaton} (q, n)\)
                    \State Let \((p, m) \in Q^{\automaton[B]} \times \N\) such that \(c_{\automaton[B]} \transition\limits^a_{\automaton[B]} (p, m)\)
                    \If{\((q \in F^{\automaton} \land n = 0) \iff \neg(p \in F^{\automaton[B]} \land m = 0)\)}
                        \State \Return \(wa\)\Comment{\(wa\) is a witness that \(\automaton\) and \(\automaton[B]\) are not equivalent}
                    \ElsIf{\(n \leq \ell \land m \leq \ell\) and \((wa, (q, n), (p, m))\) has not yet been seen}
                        \State Add \((wa, (q, n), (p, m))\) in \(Q\)
                    \EndIf
                \EndFor
            \EndWhile
            \State \(\ell \gets \ell + k\)
        \Until{we stop with probability \(0.5\)}
        \State \Return true
    \end{algorithmic}
\end{algorithm}

\subsubsection{Results.}

\begin{table}
    \centering
    \begin{tabular}{r >{\raggedleft \arraybackslash}p{30pt} @{\hspace{10pt}} r}
\toprule
$|Q|$ &  $|\Sigma|$ &  TO (20 min) \\
\midrule
             4 &           1 &            0 \\
             4 &           2 &            5 \\
             4 &           3 &           16 \\
             4 &           4 &           41 \\
             5 &           1 &            0 \\
             5 &           2 &           23 \\
             5 &           3 &           55 \\
             5 &           4 &           83 \\
\bottomrule
\end{tabular}

    \caption{Number (over 100) of executions with a timeout (TO). The executions for the missing pairs $(\sizeOfSet{Q},\sizeOfSet{\Sigma})$ could all finish.}%
    \label{table:timeouts_errors}
\end{table}
For our random benchmarks, we let the size $|Q|$ of the ROCA vary between one and five, and the size $|\Sigma|$ of the alphabet between one and four.
For each pair $(\sizeOfSet{Q},\sizeOfSet{\Sigma})$, we execute the learning algorithm on 100 ROCAs (generated as explained above).
We set a timeout of 20 minutes and a memory limit of 16GB\@.
The number of executions with a timeout is given in \Cref{table:timeouts_errors} (we do not give the pairs $(\sizeOfSet{Q},\sizeOfSet{\Sigma})$ where every execution could finish).

\begin{figure}
    \centering
    \begin{subfigure}{.45\textwidth}
        \centering
        \includegraphics{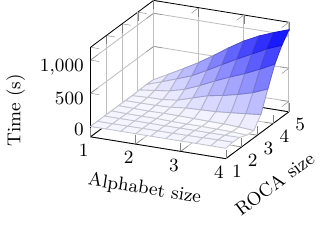}
        \caption{Mean of the total time taken by the learning algorithm.}%
        \label{fig:random:with_alphabet:time}
    \end{subfigure}%
    \hfill
    \begin{subfigure}{.45\textwidth}
        \centering
        \includegraphics{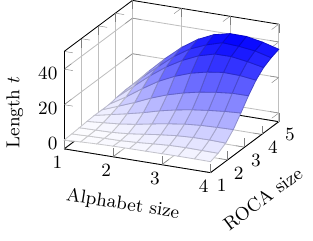}
        \caption{Mean of the length $t$ of the longest counterexample.}%
        \label{fig:random:with_alphabet:cex}
    \end{subfigure}

    \begin{subfigure}{.45\textwidth}
        \centering
        \includegraphics{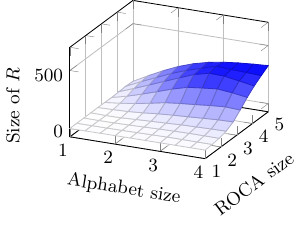}
        \caption{Mean of the final size of \(\representatives\).}%
        \label{fig:random:with_alphabet:R}
    \end{subfigure}%
    \hfill
    \begin{subfigure}{.45\textwidth}
        \centering
        \includegraphics{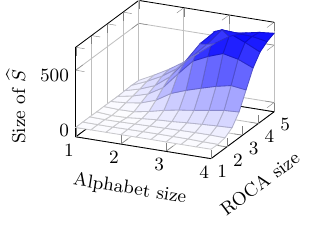}
        \caption{Mean of the final size of \(\separatorsMapping\).}%
        \label{fig:random:with_alphabet:HatS}
    \end{subfigure}
    \caption{Results for the benchmarks based on random ROCAs.}%
    \label{fig:random:with_alphabet}
\end{figure}

The mean of the total time taken by the algorithm is given in \Cref{fig:random:with_alphabet:time}.
One can see that it has an exponential growth in both sizes $\sizeOfSet{Q}$ and $\sizeOfSet{\Sigma}$.
Note that executions with a timeout had their execution time set to 20 minutes, in order to highlight the curve.
Let us now drop all the executions with a timeout. The mean length of the longest counterexample provided by the teacher for (partial) equivalence queries is presented in \Cref{fig:random:with_alphabet:cex} and the final size of the sets \(\representatives\) and \(\separatorsMapping\) are presented in \Cref{fig:random:with_alphabet:R,fig:random:with_alphabet:HatS}. 
Note that the curves go down  due to the limited number of remaining executions (for instance, the ones that could finish did not require long counterexamples).
We can see that \(\separatorsMapping\) grows larger than \(\representatives\). We conclude that these empirical results confirm the theoretical complexity claims from \Cref{thm:main} and \Cref{prop:nbrOpen-sizeR,prop:sizeSHatf}.

We conclude this section on the benchmarks based on random ROCAs by highlighting the fact that our random ROCAs do not reflect real-life automata.
That is, in more concrete cases, ROCAs may have a natural structure that is induced by the problem.
Thus, results may be very different.
We also applied our learning algorithm to a more realistic use case.

\subsection{JSON documents and JSON Schemas}\label{sec:json}

Let us now discuss the second set of benchmarks, which is a proof of concept for our learning algorithm based on JSON documents~\cite{DBLP:journals/rfc/rfc8259}. This format is the currently most popular one used for exchanging information on the web.

Our goal is to construct an ROCA that can validate a JSON document, according to some given constraints.
This use case is inspired by~\cite{DBLP:conf/webdb/ChiticR04} but applied on a more recent format.

A \emph{JSON document} is a text document that follows a specific structure.
Namely, five different types of data can be present in a document\footnote{The standard specifies a sixth type but we do not consider it in our use case.}:
\begin{itemize}
    \item An \emph{object} is an unordered collection of pairs key-value where key is a finite string and value can be any of the five different data types.
        An object must start with \verb!{! and end with \verb!}!.
    \item An \emph{array} is an ordered collection of values.
        Again, a value can be any of the five different data types.
        An array must start with \verb![! and end with \verb!]!.
    \item A \emph{string} is a finite sequence of any Unicode characters and must start and end with \verb!"!.
    \item A \emph{number} is any positive or negative decimal real number.
        In particular, an \emph{integer} is any positive or negative number without a decimal part.
    \item A \emph{boolean} can be \verb!true! or \verb!false!.
\end{itemize}
A JSON document must start with an object.

In the context we consider, the constraints a document must satisfy are given by a \emph{JSON Schema} which is itself a JSON document describing the kind of values that must be associated to each key.
An example of a schema is given in \Cref{fig:json}.
JSON schemas are described more thoroughly on the official website~\cite{JSONSchemaSite}.
They can be seen as a counterpart to DTDs for XML documents.
A schema can allow a recursive structure (to model a tree, for instance).
The ranges for the different values can be restricted.
For example, one can force an integer to be in the interval \([0, 10]\), or to be a multiple of two, and so on.

\begin{lstlisting}[language=JSON, caption=Example of a JSON Schema., label=fig:json, float]
{
  "title": "Fast",
  "type": "object",
  "required": ["string", "double", "integer", "boolean", "object", "array"],
  "additionalProperties": false,
  "properties": {
    "string": {"type": "string"},
    "double": {"type": "number"},
    "integer": {"type": "integer"},
    "boolean": {"type": "boolean"},
    "object": {
      "type": "object",
      "required": ["anything"],
      "additionalProperties": false,
      "properties": {
        "anything": {
          "type": ["number", "integer", "boolean", "string"]
        }
      }
    },
    "array": {
      "type": "array",
      "items": {"type": "string"},
      "minItems": 2
    }
  }
}
\end{lstlisting}

\subsubsection{Learning a JSON Schema as an ROCA.}

We propose to validate a JSON document against a given JSON schema as follows.
The learner learns an ROCA from the schema. With this ROCA, one can easily decide whether a JSON document is valid against the schema.

In this learning process, we suppose the teacher knows the target schema and the queries are specialized as follows:
\begin{itemize}
    \item Membership query: the learner provides a JSON document and the teacher returns true if and only the document is valid for the schema.
    \item Counter value query: the learner provides a JSON document and the teacher returns the number of unmatched \verb!{! and \verb![!. Adding the two values is a heuristic abstraction that allows us to summarize two-counter information into a single counter value. Importantly, the abstraction is a design choice regarding our implementation of a teacher for these experiments and not an assumption made by our learning algorithm.
    \item Partial equivalence query: the learner provides a DFA and a counter limit \(\ell\).
        The teacher randomly generates an a-priori fixed number of documents with a height not exceeding the counter limit $\ell$ and checks whether the DFA and the schema both agree on the documents' validity.
        If a disagreement is noticed, the incorrectly classified document is returned.
    \item Equivalence query: the learner provides an ROCA\@.
        The teacher also generates an a-priori number of randomly generated documents (but, this time, without a bound on the height) and verifies that the ROCA and the schema both agree on the documents' validity.
        If one is found, an incorrectly classified document is returned.
\end{itemize}

Note that the randomness of the (partial) equivalence queries implies that the learned ROCA may not completely recognize the same set of documents than the schema.
Setting the number of generated documents to be a large number would help reducing the probability that an incomplete ROCA is learned.
One could also control the randomness of the generation to force some important keys to appear, even if these keys are not required in the schema.

In order for an ROCA to be learned in a reasonable time, some abstractions must be made mainly for reducing the alphabet size.
\begin{itemize}
    \item If an object has a key named \verb!key!, we consider the sequence of characters \verb!"key"! as a single alphabet symbol.
    \item Strings, integers, and numbers are abstracted as follows. All strings must be equal to \verb!"\S"!, all integers to \verb!"\I"!, and all numbers to \verb!"\D"!.
        Booleans are left as-is since they can only take two values (\verb!true! or \verb!false!).
    \item The symbols \verb!,!, \verb!{!, \verb!}!, \verb![!, \verb!]!, \verb!:! are all considered as different alphabet symbols (note that since \verb!"! is considered directly into the keys' symbols or the values' symbols, that symbol does not appear here).
\end{itemize}

Moreover, notice that the alphabet is not known at the start of the learning process (due to the fact that keys can be any strings). Therefore we slightly modify the learning algorithm to support growing alphabets.
More precisely, the learner's alphabet starts with the symbols \verb!{! and \verb!}! (to guarantee we can at least produce a syntactically valid JSON document for the first partial equivalence query) and is augmented each time a new symbol is seen.

A last abstraction was applied for our benchmarks: we assume that each object is composed of an ordered (instead of unordered) collection of pairs key-value.
It is to be noted that the learning algorithm can learn without this restriction. However it requires substantially more time as all possible orders inside each object must be considered (and they all induce a different set of states in the ROCA). Our learning algorithm could be improved by taking into account that objects are unordered collections; we did not investigate this approach.

\subsubsection{Results.}

\begin{table}
    \centering
    \begin{tabularx}{\textwidth}{@{} R R R R R R R R @{}}
    \toprule
    Schema & TO (1h) &  Time (s) &  $\lengthCe$ &  $|R|$ &  $|\widehat{S}|$ &  $|\automaton|$ &  $|\Sigma|$ \\
    \midrule
    1 & 0 &     16.39 &        31.00 &  55.55 &            32.00 &           33.00 &       19.00\\
    2 & 27 &   1045.64 &        12.99 &  57.84 &            33.74 &           44.29 &       14.70 \\
    3 & 19 &    922.19 &        49.49 & 171.94 &            50.49 &           51.16 &        9.00\\
    \bottomrule
\end{tabularx}
    \caption{Results for JSON benchmarks.}%
    \label{table:results:json}
\end{table}

We considered three JSON schemas.
The first is the document from \Cref{fig:json} which lists all possible types of values (i.e., it contains an integer, a double, and so on).
The second is a real-world JSON schema\footnote{We downloaded the schema from the JSON Schema Store~\cite{SchemaStore}. We modified the file to remove all constraints of type \enquote{enum}.} used by a code coverage tool called Codecov~\cite{codecov}.
Finally, the third schema encodes a recursive list, i.e., an object containing a list with at most one object defined recursively.
This last example is used to force the behavior graph to be ultimately periodic, and is given in \Cref{sec:recursive_schema}.

The \Cref{table:results:json} gives the results of the benchmarks, obtained by fixing the number of random documents by (partial) equivalence query to be 1000.
For each schema, 100 experiments were conducted with a time limit of one hour by execution.
We can see that real-world JSON schemas and recursively-defined schemas can be both learned by our approach.
One last interesting statistics we can extract from the results is that the number of representatives is larger than the number of separators, unlike for the random benchmarks.

\section{Conclusion}\label{sec:conclusion}

We have designed a new learning algorithm for realtime one-counter automata. Our algorithm uses membership, counter value, partial equivalence, and equivalence queries. The algorithm executes in exponential time and space, and requires at most an exponential number of queries. We have implemented this algorithm and evaluated it on two benchmarks.

As future work, we believe one might be able to remove the use of partial equivalence queries. In this direction, perhaps replacing our use of Neider and L\"oding's VCA algorithm by Isberner's TTT algorithm~\cite{DBLP:conf/rv/IsbernerHS14} for visibly pushdown automata might help. Indeed, the TTT algorithm does not need partial equivalence queries. 

Another interesting direction concerns lowering the (query) complexity of our algorithm. In~\cite{DBLP:journals/iandc/RivestS93}, it is proved that \LStar algorithm~\cite{DBLP:journals/iandc/Angluin87} can be modified so that adding a single separator after a failed equivalence query is enough to update the observation table. This would remove the suffix-closedness requirements on the separator sets $S$ and $\widehat{S}$. It is not immediately clear to us whether the definition of \(\bot\)-consistency presented here holds in that context. Further optimizations, such as discrimination tree-based algorithms (such as Kearns and Vazirani's algorithm~\cite{DBLP:books/daglib/0041035}), also do not need the separator set to be suffix-closed.

It would also be interesting to directly learn the one-counter language instead of an ROCA\@. Indeed, our algorithm learns some ROCA that accepts the target language. It would be desirable to learn some canonical representation of the language (e.g.\ a minimal automaton, for some notion of minimality).

Finally, as far as we know, there currently is no active learning algorithm for deterministic one-counter automata (such that $\emptyword$-transitions are allowed). We want to study how we can adapt our learning algorithm in this context.

\printbibliography

\clearpage
\appendix

\section{Angluin's DFA-learning algorithm}\label{app:Lstar}

The $\LStar$ algorithm proposed in~\cite{DBLP:journals/iandc/Angluin87} is the following one.

\begin{algorithm}
    \caption{Learning a DFA~\cite{DBLP:journals/iandc/Angluin87}}%
    \label{alg:lstar}
    \begin{algorithmic}[1]
        \Require The target language \(L\)
        \Ensure A DFA accepting \(L\) is returned
        \Statex
        \State Initialize \(\obsTable[]\) with \(\representatives = \separatorsCounters = \{\emptyword\}\)
        \State Fill \(\obsTable[]\) by asking membership queries
        \While{true}
            \State Make $\obsTable[]$ closed and \consistent
            \State Construct the DFA \(\automaton_{\obsTable[]}\) from \(\obsTable[]\)
            \State Ask an equivalence query over \(\automaton_{\obsTable[]}\)
            \If{the answer is positive}
                \State \Return \(\automaton_{\obsTable[]}\)
            \Else
                \State Given the counterexample $w$, add $\prefixes{w}$ to $R$
                \State Update \(\obsTable[]\) by asking membership queries
            \EndIf
        \EndWhile
    \end{algorithmic}
\end{algorithm}

\section{Correctness of the ROCA construction from the periodic description of a behavior graph}\label{sec:correct-period2roca}

Let $\behaviorGraphAut$ be the behavior graph of an ROCA $\automaton$ and $\alpha$ be a periodic description of $\behaviorGraphAut$. We explain in the proof of \Cref{prop:DescriptionAut} how to construct an ROCA $\automaton_{\alpha}$ from $\alpha$ that accepts the same language $L$ as $\automaton$. We here prove that this construction is correct.

\begin{proof}[of \Cref{prop:DescriptionAut} (correctness)]
    We prove that
\(\forall u \in \Sigma^*\), it holds that \(u \in L \iff u \in \languageOf{\automaton_{\alpha}}\).
    To do so, in view of how are defined the final states of $\automaton_{\alpha}$, we just need to show that \(\forall u \in \Sigma^*\), we have \[\equivalenceClass{\emptyword}_{\equivBGROCA} \transition^u_{\behaviorGraphAut} \equivalenceClass{u}_{\equivBGROCA} \iff ((\nu_0(\equivalenceClass{\emptyword}_{\equivBGROCA}), 0), 0) \transition^u_{\automaton_{\alpha}} ((\nu_c(\equivalenceClass{u}_\equivBGROCA), c), \max\{0, n - m\})\] with \(n = \counterAut{u}\) and
    \[
        c = \begin{cases}
            n & \text{if \(n \leq m\)},\\
            m + ((n - m) \mod k) & \text{otherwise.}
        \end{cases}
    \]

    Notice that if \(\equivalenceClass{u}_{\equivBGROCA}\) is a reachable state in \(\behaviorGraphAut\), then \(u \in \prefixes{L}\) by definition.
    Moreover, by construction, only the states \(\equivalenceClass{u}_{\equivBGROCA}\) are considered when constructing \(\automaton_{\alpha}\).
    Thus, \(u \in \prefixes{L}\) and \(\counterAut{u}\) is well-defined.

    We do so by induction over the length of \(u\).
    Let \(u \in \Sigma^*\) such that \(\lengthOf{u} = 0\), i.e., \(u = \emptyword\).
    So, \((q_0, 0) \transition\limits^{\emptyword}_{\automaton_\alpha} (q_0, 0)\).
    By construction, \((q_0, 0) = (\nu_0(\equivalenceClass{\emptyword}_{\equivBGROCA}), 0)\) which implies the statement is verified.
    
    Now, let \(i \in \N\) and assume it is true for any \(v \in \Sigma^*\) of length \(i\).
    Let \(u \in \Sigma^*\) such that \(\lengthOf{u} = i + 1\), i.e., \(u = va\) with \(a \in \Sigma\) and \(v \in \Sigma^*\) such that \(\lengthOf{v} = i\).
    By the induction hypothesis, we know that \(\equivalenceClass{\emptyword}_{\equivBGROCA} \transition\limits^v_{\behaviorGraphAut} \equivalenceClass{v}_{\equivBGROCA} \iff ((\nu_0(\equivalenceClass{\emptyword}_{\equivBGROCA}), 0), 0) \transition\limits^v_{\automaton_{\alpha}} ((\nu_{c'}(\equivalenceClass{v}_{\equivBGROCA}), c'), \max\{0, n' - m\})\) with \(n' = \counterAut{v}\) and
    \[
        c' = \begin{cases}
            n' & \text{if \(n' \leq m\),}\\
            m + ((n' - m) \mod k) & \text{otherwise.}
        \end{cases}
    \]
    It is consequently sufficient to prove that the last transition is correct, i.e.,
    \[
        \equivalenceClass{v}_{\equivBGROCA} \transition^a_{\behaviorGraphAut} \equivalenceClass{u}_{\equivBGROCA} \iff ((\nu_{c'}(\equivalenceClass{v}_{\equivBGROCA}), c'), \max\{n' - m\}) \transition^a_{\automaton_{\alpha}} ((\nu_{c}(\equivalenceClass{u}_{\equivBGROCA}), c), \max\{n - m\})
    \] with \(n = \counterAut{u}\) and
    \[
        c = \begin{cases}
            n & \text{if \(n \leq m\),}\\
            m + ((n - m) \mod k) & \text{otherwise.}
        \end{cases}
    \]
    We do so in three cases:
    \begin{itemize}
        \item If \((n' < m) \lor (n' = m \land n \leq m)\), then \(\max\{0, n - m\} = 0 = \max\{0, n' - m\}\), and we only work with \(\deltaZero\).
            Moreover, we do not change the counter value.
            We have \(c' = n' = \counterAut{v}\) and \(c = n = \counterAut{u}\) and the following equivalences (note that \(n - n' \in \{-1, 0, +1\}\)):
            \begin{align*}
                \equivalenceClass{v}_{\equivBGROCA} \transition^a_{\behaviorGraphAut} \equivalenceClass{u}_{\equivBGROCA} &\iff \tau_{n'}(\nu_{n'}(\equivalenceClass{v}_{\equivBGROCA}), a) = (\nu_{n}(\equivalenceClass{u}_{\equivBGROCA}), n - n')\\
                &\iff \tau_{c'}(\nu_{c'}(\equivalenceClass{v}_{\equivBGROCA}), a) = (\nu_{c}(\equivalenceClass{u}_{\equivBGROCA}), n - n')\\
                &\iff \deltaZero((\nu_{c'}(\equivalenceClass{v}_{\equivBGROCA}), c'), a) = ((\nu_c(\equivalenceClass{u}_{\equivBGROCA}), c), 0)\\
                &\iff ((\nu_{c'}(\equivalenceClass{v}_{\equivBGROCA}), c'), 0) \transition^a_{\automaton_{\alpha}} ((\nu_c(\equivalenceClass{u}_{\equivBGROCA}), c), 0).
            \end{align*}
        \item If \(n' = m\) and \(n > m\), then  $n = m+1$ since $n' = \counterAut{v}$ and $n = \counterAut{u}$. Moreover, \(\max\{0, n' - m\} = 0\) and \(\max\{0, n - m\} = \max\{0, m + 1 - m\} = 1 \).
            Again, we work with \(\deltaZero\) but we increase the counter value.
            We have \(c' = m\) and \(c = m + ((m + 1 - m) \mod k) = m + (1 \mod k)\).
            We have the following equivalences:
            \begin{align*}
                \equivalenceClass{v}_{\equivBGROCA} \transition^a_{\behaviorGraphAut} \equivalenceClass{u}_{\equivBGROCA} &\iff \tau_{n'}(\nu_{n'}(\equivalenceClass{v}_{\equivBGROCA}), a) = (\nu_{n}(\equivalenceClass{u}_{\equivBGROCA}), n - n')\\
                &\iff \tau_{c'}(\nu_{c'}(\equivalenceClass{v}_{\equivBGROCA}), a) = (\nu_{c}(\equivalenceClass{u}_{\equivBGROCA}), 1)\\
                &\iff \deltaZero((\nu_{c'}(\equivalenceClass{v}_{\equivBGROCA}), c'), a) = ((\nu_{c}(\equivalenceClass{u}_{\equivBGROCA}), c), +1)\\
                &\iff ((\nu_{c'}(\equivalenceClass{v}_{\equivBGROCA}), c'), 0) \transition^a_{\automaton_{\alpha}} ((\nu_{c}(\equivalenceClass{u}_{\equivBGROCA}), c), +1).
            \end{align*}
        \item Finally, if \(n' > m\), then \(\max\{0, n' - m\} = n' - m\) and \(\max\{0, n - m\} = n - m\).
            This time, we work with \(\deltaNotZero\) and the counter value is modified according to \(n - n'\).
            We have \(c' = m + ((n' - m) \mod k)\) and \(c = m + ((n - m) \mod k)\).
            We have the following equivalences:
            \begin{align*}
                \equivalenceClass{v}_{\equivBGROCA} \transition^a_{\behaviorGraphAut} \equivalenceClass{u}_{\equivBGROCA} &\iff \tau_{n'}(\nu_{n'}(\equivalenceClass{v}_{\equivBGROCA}), a) = (\nu_n(\equivalenceClass{u}_{\equivBGROCA}), n - n')\\
                &\iff \tau_{c'}(\nu_{c'}(\equivalenceClass{v}_{\equivBGROCA}), a) = (\nu_{c}(\equivalenceClass{u}_{\equivBGROCA}), n - n')\\
                &\iff \deltaNotZero((\nu_{c'}(\equivalenceClass{v}_{\equivBGROCA}), c'), a) = ((\nu_c(\equivalenceClass{u}_{\equivBGROCA}), c), n - n')\\
                &\iff ((\nu_{c'}(\equivalenceClass{v}_{\equivBGROCA}), c'), n' - m) \transition^a_{\automaton_{\alpha}} ((\nu_c(\equivalenceClass{u}_{\equivBGROCA}), c), n - m).
            \end{align*}
    \end{itemize}
    We have shown that the last transition is correct, in every case.
    It is thus obvious that \(\languageOf{\automaton_{\alpha}} = L = \languageOf{\behaviorGraphAut}\).
\qed\end{proof}

\section{Upper bounds on recursive functions}\label{sec:bounds_recurrence} 

\begin{lemma}\label{lemma:bounds_recurrence}
    Let \(\alpha, \beta \geq 1\) be constants and \(S\) be a function defined as:
    \begin{align*}
        S(0) &= \beta,\\
        S(j) &= S(j - 1) \cdot \alpha \cdot j + \beta, \quad\forall j \geq 1.
    \end{align*}
    Then, for all \(j \in \N\), it holds that \(S(j) \leq (j + 1 ) \cdot {(\alpha \cdot j)}^j \cdot \beta\).
\end{lemma}
\begin{proof}
    Let \(P(k, \ell) = k \cdot (k + 1) \cdot (k + 2) \dotsm \ell\) for any $k, \ell \in \N$ such that \(1 \leq k \leq \ell\).
    Notice that \(P(k, \ell) \leq \ell^{\ell - k +1} \leq \ell^{\ell}\).
    
    Let $j \in \N$. We have:
    \begin{align*}
        S(j)    &= S(j - 1) \cdot \alpha \cdot j + \beta\\
                &= (S(j - 2) \cdot \alpha \cdot (j - 1) + \beta) \cdot \alpha \cdot j + \beta\\
                &= ((S(j - 3) \cdot \alpha \cdot (j - 2) + \beta) \cdot \alpha \cdot (j-1) + \beta) \cdot \alpha \cdot j + \beta\\
                &= \dotso\\
                &= S(0) \cdot \alpha^j P(1, j) + (\alpha^{j-1} P(2, j) + \alpha^{j-2} P(3, j) + \dotsb + \alpha P(j, j) + 1) \cdot \beta\\
                &= (\alpha^j P(1, j) + \alpha^{j-1} P(2, j) + \alpha^{j-2} P(3, j) + \dotsb + \alpha P(j, j) + 1) \cdot \beta\\
                &\leq (j + 1) \cdot {(\alpha \cdot j)}^j \cdot \beta.
    \end{align*}
\qed\end{proof}

\section{Third schema used for the experiments from \Cref{sec:json}}\label{sec:recursive_schema}
\Cref{fig:json:recursive} gives the third JSON schema that was used in our JSON-based benchmarks presented in \Cref{sec:json}.
The \verb!{"$ref": "#"}! indicates that the each item in the array \verb!children! is defined as the top object, i.e., we have a recursive definition.
\begin{lstlisting}[language=json, label=fig:json:recursive, caption=Third schema.]
{
    "type": "object",
    "properties": {
        "name": {"type": "string"},
        "children": {
            "type": "array",
            "maxItems": 1,
            "items": {"$ref": "#"}
        }
    },
    "required": ["name"],
    "additionalProperties": false
}
\end{lstlisting}

\end{document}